%% file: draft.tex
\documentclass[12pt]{article}

\usepackage{styles/natbib}
\usepackage{microtype}
\usepackage{graphicx}
\usepackage{subfigure}
\usepackage{booktabs} % for professional tables
\usepackage{fullpage}
\usepackage{longtable}
\usepackage{caption}
\usepackage{enumitem}
\usepackage[breaklinks,colorlinks,
            linkcolor=red,
            anchorcolor=blue,
            citecolor=blue
            ]{hyperref}
\usepackage{amsmath}
\usepackage{amssymb}
\usepackage{mathtools}
\usepackage{amsthm}

\usepackage[capitalize,noabbrev]{cleveref}

\theoremstyle{plain}
\newtheorem{theorem}{Theorem}[section]
\newtheorem{proposition}[theorem]{Proposition}

\theoremstyle{definition}

\newtheorem{assumption}[theorem]{Assumption}
\theoremstyle{remark}

\newtheorem{remark}[theorem]{Remark}
\usepackage[textsize=tiny]{todonotes}

\usepackage{relsize}
\usepackage[utf8]{inputenc} % allow utf-8 input
\usepackage[T1]{fontenc}    % use 8-bit T1 fonts
\usepackage{url}            % simple URL typesetting
\usepackage{booktabs}       % professional-quality tables
\usepackage{amsfonts}       % blackboard math symbols
\usepackage{nicefrac}       % compact symbols for 1/2, etc.
\usepackage{microtype}      % microtypography
\usepackage{styles/smile}
\usepackage{dsfont}
\usepackage{xcolor}
\usepackage{tikz}

\newcommand{\diff}{{\rm d}}
\newcommand{\indep}{\perp\!\!\!\perp}

\NewDocumentCommand{\indicator}{}{\mathds{1}}

\title{Partial Identification of Policy-Relevant Treatment Effects with Instrumental Variables via Optimal Transport}
\author{Jiyuan Tan \quad Jose Blanchet \quad Vasilis Syrgkanis \\ \texttt{\{jiyuantan, jose.blanchet, vsyrgk\}@stanford.edu} \\ Management Science and Engineering, Stanford University}

\begin{document}
\maketitle

\begin{abstract}
    Policy-Relevant Treatment Effects (PRTEs) are generally not point-identified under standard Instrumental Variable (IV) assumptions when the instrument generates limited support in treatment propensity. We show that PRTE partial identification in the generalized Roy model can instead be formulated as a Constrained Conditional Optimal Transport (CCOT) problem over the joint conditional law of the potential outcome and the latent resistance. The resulting multidimensional CCOT problem reduces analytically to separable one-dimensional OT problems with product costs, yielding sharp closed-form bounds and avoiding direct solution of the original high-dimensional CCOT problem. We also develop estimation and inference procedures for these bounds: for discrete instruments, we use a Double Machine Learning (DML) approach based on Neyman-orthogonal scores that accommodates high-dimensional covariates while achieving the parametric $\sqrt{n}$ rate and asymptotic normality; for continuous instruments, we explicitly characterize the corresponding nonparametric convergence rates. The framework accommodates covariates, discrete and continuous instruments, and extensions to general treatment settings. In simulations and a bed-net subsidy application, the resulting bounds are substantially tighter than the moment-relaxation method.
\end{abstract}

\input{sections/introduction}

\input{sections/preliminaries}

\input{sections/pid}

\input{sections/extension}

\section{Estimation and Inference}\label{sec:est_inf}
\input{sections/est_inf}

\section{Simulation}\label{sec:simulation}

\input{sections/experiments}

\section{Conclusion}\label{sec:conclusion}
\input{sections/conclusion}

\bibliographystyle{plainnat}
\bibliography{ref}

\newpage
\appendix
\input{sections/notation}

\newpage

\input{sections/appendix}
\end{document}

%% file: sections/introduction.tex
\section{Introduction}

Instrumental variables (IVs) are a standard tool for identifying causal effects when treatment is endogenous in observational studies. Under the classical IV assumptions, \citet{angrist1995identification} show that the Local Average Treatment Effect (LATE), defined as the average effect for the subpopulation of compliers, is point-identified. However, many policy questions are not about compliers themselves; they ask how outcomes would change under a counterfactual policy that shifts treatment take-up in the broader population. Such questions are naturally captured by the Policy-Relevant Treatment Effects (PRTEs), a central target in the generalized Roy and Marginal Treatment Effect (MTE) framework \citep{heckman1999local,heckman2005structural,heckman2006understanding}, which encompass a wide range of causal estimands including the Average Treatment Effect (ATE). Empirical policy evaluations within this framework include \citet{carneiro2010evaluating,carneiro2011estimating}. Unlike LATE, PRTEs are generally not point-identified when the instrument induces only limited support in the treatment propensity.

%Classical applications include the causal return to education \citep{angrist1991does,card1993using}, the earnings effects of military service \citep{angrist1990lifetime}, the long-run impact of institutions on economic development \citep{acemoglu2001colonial}, and policy interventions with imperfect treatment compliance \citep{angrist1995identification}. 

This limited-support problem is central in the generalized Roy model \citep{heckman2005structural}, where treatment selection follows a threshold-crossing mechanism formally equivalent to the classical monotonicity condition \citep{angrist1995identification,vytlacil2002independence}. In this framework, PRTEs can be written as weighted averages of the MTE, so identification depends on how much of the latent resistance distribution is reached by the propensity score induced by the instrument. When the propensity score lacks full support over the unit interval, PRTEs and related quantities such as the ATE remain only partially identified, and a long line of work has studied such identified sets under various assumptions \citep{manski1990nonparametric,manski1997monotone,manski_pepper2000monotone,heckman2001instrumental,magne2018using}. 

Despite its importance, deriving tight, closed-form bounds for general PRTEs remains a significant challenge. Existing strategies fall into three broad approaches: moment-relaxation over MTR functions \citep{magne2018using,mogstad2018identification}, second-order stochastic dominance constraints on the MTE function solved by rearrangement \citep{marx2024sharp}, and infinite-dimensional linear programming over the joint law of outcome and latent resistance solved by sieve approximation \citep{han_yang2024computational}. Each faces limitations---discarded distributional information, restricted scope to weights on the complier interval under a binary instrument, lack of closed-form solutions or incompatibility with continuous instruments---that we discuss in detail in \cref{subsec:related_work}.

In this paper, we develop a unified Optimal Transport (OT) framework for partial identification in the generalized Roy model, and use it to derive sharp, closed-form bounds for PRTEs. Specifically, we recast PRTE partial identification as a Constrained Conditional Optimal Transport (CCOT) problem, a structured instance within the broader OT framework that is equivalent to the infinite-dimensional LP of \citet{han_yang2024computational}. Rather than summarizing the observed distributions through moment conditions or imposing distributional constraints on conditional mean (MTE) functions, we place the distributional constraints directly on the joint conditional law of the potential outcome and latent resistance. Then, we exploit the OT structure to prove that this heavily constrained, multidimensional CCOT problem analytically reduces to separable standard one-dimensional OT problems. This reduction completely bypasses computationally intensive linear programming and sieve approximation, yielding explicit, closed-form solutions for the sharp bounds. Because the framework constrains the full joint latent law rather than a conditional mean, it applies in principle to a broader class of partial identification problems in the generalized Roy model; here we focus on PRTEs. 

Our main contributions can be summarized as follows.

\begin{enumerate}
    \item \textbf{A Unified OT Framework for the Generalized Roy Model:} We formulate partial identification in the generalized Roy model as a CCOT problem. The framework imposes distributional constraints directly at the joint law level and we prove that constraints decomposes analytically into separable marginal constraints.
    \item \textbf{Analytic Decomposition and Closed-Form PRTE Bounds:} Leveraging the decomposition, we prove the CCOT problem decomposes into separable one-dimensional OT problems with product costs, yielding explicit, closed-form sharp bounds and bypassing the sieve approximation required by existing computational approaches.
    \item \textbf{Semi-parametric Estimation and Inference:} We develop estimation results for the derived bounds. For discrete instruments, we leverage Double Machine Learning (DML) to construct a debiased estimator, achieving $\sqrt{n}$-consistency and asymptotic normality. For continuous instruments, we explicitly characterize the nonparametric convergence rates.
    \item \textbf{Generalization to General Treatments:} We extend our identification strategy beyond the binary treatment setting to continuous and multi-valued treatments, demonstrating the flexibility of the OT framework. For the multi-valued treatment case, this extension requires an additional algebraic condition on the ranges of the propensity scores induced by the instrument (\cref{asp:algebra}).
    \item \textbf{Empirical Validation:} We validate our theoretical results through both synthetic simulations and a real-world empirical application, demonstrating how our closed-form bounds tighten the identified set relative to the moment-relaxation methods.
\end{enumerate}

\paragraph{Organization of this paper.} The remainder of the paper is organized as follows. \cref{sec:pre} introduces the setup and optimal transport background; \cref{sec:pid} presents the main identification theory and closed-form bounds; \cref{sec:extension} extends the framework to general treatments; \cref{sec:est_inf} develops estimation and inference; and \cref{sec:simulation} reports simulation and empirical results.

\subsection{Related Work}\label{subsec:related_work}
\paragraph{Local IV and MTE.}
The Marginal Treatment Effect (MTE), introduced by \citet{heckman1999local} and systematically developed by \citet{heckman2005structural}, is the foundational building block of our identification analysis. The MTE is defined as the treatment effect for individuals at a specific quantile of unobserved resistance to treatment, and serves as a unifying object from which all standard IV estimands--including LATE, ATE, ATT, and PRTE--can be recovered as weighted averages \citep{heckman2005structural, heckman2006understanding}. The formal equivalence between the threshold-crossing selection model and the monotonicity assumption of \citet{angrist1995identification} is established by \citet{vytlacil2002independence}, who shows that the two frameworks are equivalent under a continuity condition on the unobservable. Point identification of the MTE is achieved via the Local Instrumental Variables (LIV) estimand \citep{heckman1999local}, which requires full propensity score support over $[0,1]$---a condition that fails under limited instrument variation, the regime our paper addresses. Empirical MTE-based policy evaluations include \citet{carneiro2010evaluating,carneiro2011estimating}.

\paragraph{IV with General Treatment.} \citet{angrist1995two} extends the LATE framework to treatments with multiple discrete ordered levels. \citet{heckman2007econometric} extends the MTE framework to the ordered choice model. \citet{kirkeboen2016field} provides methodology for identifying treatment effects in settings with multiple unordered discrete choices--specifically fields of study--by accounting for how individuals self-select into alternatives based on their unobserved comparative advantages. \citet{lee2018identifying} develop identification results for multivalued treatments. In the continuous treatment setting, \citet{florens2008identification} use the control function method to identify the ATE and ATT, and \citet{imbens2009identification} use a similar approach to identify structural equations in models with continuous endogenous variables. While these results establish point identification for certain causal quantities, our framework instead targets partial identification of PRTEs, and we extend our results to the continuous and multi-valued treatment settings in \cref{sec:extension}.

\paragraph{Partial Identification.} \citet{manski1989anatomy,manski1990nonparametric,manski1997monotone,manski_pepper2000monotone,manskiInferenceRegressionsInterval2002} pioneered partial identification, providing closed-form bounds for the IV problem under a range of assumptions. In the econometric literature, partial identification is often cast as estimating the identified set defined by moment inequalities \citep{chernozhukov_hong_tamer2007,rosen2008confidence,romano2008inference,canay2010inference,bugni2010bootstrap}; see \citep{tamer2010partial} for a comprehensive survey. As a special case, our bound recovers the classical Manski bound for ATE. The computational and analytical approaches of \citet{han_yang2024computational} and \citet{marx2024sharp} to PRTE partial identification in the MTE framework are discussed in the \emph{PRTE Partial Identification} paragraph below.

\citet{balke_pearl1997bounds} develop the linear programming approach for partial identification. Recently, this optimization-based approach has been revisited and generalized using modern techniques \citep{balazadeh_meresht_syrgkanis_krishnan2022,guo_yin_wang_jordan2022,duarte2024automated,voronin2025linear,levis2025covariate}. In particular, \citet{levis2025covariate} use covariates to tighten \citet{balke_pearl1997bounds}'s classical bound and develop estimation theory in the IV setting. Several works establish connections between partial identification and robust optimization \citep{guo_yin_wang_jordan2022,balazadeh_meresht_syrgkanis_krishnan2022,gao_ge_qian2024,penn_gunderson_bravohermsdorff_silva_watson2025,fan_pass_shi2025,tan2024consistency}. The partial identification problem can be formulated as an optimal transport problem with marginal constraints derived from the observational distribution \citep{ji_lei_spector2024,lin2025estimation,lin2025tightening}. For a recent primer on optimal transport methods for causal inference, see \citep{gunsilius2025primer}. In a related but distinct setting, \citet{fan_guerre_zhu2017partial} study partial identification of functionals of the joint distribution of potential outcomes in settings where the conditional marginal distributions are identified---including the latent threshold-crossing model of \citet{heckman2005structural} under full propensity-score support. \citet{oberreynolds2023estimating} characterizes sharp bounds on functionals of the joint complier outcome distribution via optimal transport under the classical binary-instrument LATE framework of \citet{angrist1995identification}, focusing on the identified complier subpopulation rather than the partially identified PRTEs we target here.

\paragraph{PRTE Partial Identification.} Three strategies have emerged for PRTE partial identification under limited support. \citet{magne2018using} optimize over MTR functions subject to moment constraints \citep{mogstad2018identification}, which discards distributional information beyond the first moment. \citet{han_yang2024computational} preserve full distributional content by formulating an infinite-dimensional LP over the joint law of latent states, solved by sieve approximation for discrete instruments. \citet{marx2024sharp} generalize \citep{magne2018using} by imposing full distributional content on the MTE function via a second-order stochastic dominance (SSD) constraint, obtaining sharp closed-form bounds by rearrangement bounds; his explicit theorems cover weights on the complier interval under a binary instrument, with extensions beyond compliers requiring rank similarity. 

Our work differs from all three in scope: we impose distributional constraints directly on the joint conditional law of the potential outcome and latent resistance rather than on a conditional mean, cast the resulting program as a CCOT problem equivalent to the LP of \citet{han_yang2024computational}, and show it reduces analytically to separable one-dimensional OT problems with product costs. This delivers (i) sharp closed-form bounds for the full PRTE on all of $[0,1]$ under multi-valued instruments via an interval decomposition across $K\!+\!1$ compliance types, including boundary intervals where only one potential outcome marginal is identified, (ii) extensions to continuous and mixed instruments with sub-problems indexed by the propensity score continuum, and (iii) discrete-instrument inference together with continuous-instrument rate results, which are not available in either \citet{marx2024sharp} or \citet{han_yang2024computational}. On the complier interval under a binary instrument, the CCOT solution and Marx's SSD rearrangement coincide. Beyond PRTE, our sharpness result characterizes the identified set at the level of the full joint latent law and sheds light on the joint-identification question raised in \citet{marx2024sharp} through its reduction to independent one-dimensional marginal constraints.

% Crucially, while \citet{marx2024sharp} imposes distributional constraints only on the conditional mean (the MTE function) and does not characterize the identified set at the level of the joint latent distribution, our sharpness result (\cref{prop:sharpness}) directly characterizes the identified set of the full joint latent law of $(Y(0), Y(1), U, X)$.

%% file: sections/preliminaries.tex
\section{Preliminaries} \label{sec:pre}
\subsection{IV Model and Policy-Relevant Treatment Effects}\label{subsec:iv_prte}
For most of this paper, we work with the generalized Roy model \citep{heckman2005structural}. In the model, the observable variables are treatment $W\in \mathcal{W} = \{0,1\}$, instrument $ Z \in \mathcal{Z}$, covariates $X\in\mathcal{X}$ and the outcome $Y \in \mathcal{Y}$, where $\mathcal{X} \subset\mathbb{R}^{d_X}, \mathcal{Z}$ and $ \mathcal{Y} \subset \mathbb{R}$ are the domains of the corresponding variables. We consider both discrete and continuous instrument settings in this paper. Throughout the paper, we assume the domains of all random variables are compact. The unobservables are the potential outcomes $Y(0), Y(1)$, and the variable $U$ in the selection equation (\ref{eq:thresh_xing}). 
\begin{assumption}[Structural Assumptions] \label{asp:structure}
    We will make the following structural assumptions for our IV model. Throughout, we assume $Y(w) \in \mathcal{Y}$, $Z \in \mathcal{Z}$, and $X \in \mathcal{X}$ almost surely.
    \begin{enumerate}
        \item (Consistency) $Y = WY(1) + (1-W)Y(0)$.
        \item (Conditional Instrumental Exogeneity) $Z \indep (U,Y(0),Y(1)) \mid X$, where $\indep$ denotes statistical independence.
        \item (Threshold Crossing) The treatment is selected by
        \begin{align}\label{eq:thresh_xing}
            W = \indicator (U\leqslant p(Z,X)),
        \end{align}
        where $p(Z,X) = \mathbb{P}(W=1 \mid Z, X)$. Moreover, $U \mid X \sim \mathrm{Unif}(0,1)$.
        \item (Common Instrument Support) The conditional support of $Z$ given $X$ does not depend on $X$: $\mathrm{Supp}(Z \mid X=x) = \mathcal{Z}$ for all $x \in \mathcal{X}$.
    \end{enumerate}
\end{assumption}

It is without loss of generality to assume that $U\sim\text{Unif}(0,1)$ if $U$ is continuously distributed, because we can always normalize its marginal distribution.\footnote{By the probability integral transform, if $U \mid X$ has a continuous CDF $F_{U\mid X}(\cdot \mid x)$, then $\tilde{U} := F_{U\mid X}(U \mid X) \sim \mathrm{Unif}(0,1)$ conditional on $X$. Since $\tilde{U}$ is a strictly increasing function of $U$, the threshold-crossing model is preserved under this reparametrization, and we may work with $\tilde{U}$ in place of $U$ without loss of generality.} The common instrument support condition is essential for the identification of the propensity score $p(z,x)$: without $z \in \mathrm{Supp}(Z \mid X=x)$ for all $x$, the conditional probability $\mathbb{P}(W=1 \mid Z=z, X=x)$ is not well-defined and hence $p(z,x)$ is not identified.

We also require the following boundedness condition on the potential outcomes, which is needed to obtain informative bounds.
\begin{assumption}[Boundedness of the Potential Outcome]\label{asp:bounded_y}
The outcome domain satisfies $\mathcal{Y} \subseteq [y_{\min}, y_{\max}]$ for known constants $y_{\min} < y_{\max}$.
\end{assumption}

The fundamental building block for evaluating causal parameters in this framework is the MTE, originally introduced by \citet{heckman1999local,heckman2005structural}. The MTE is defined as the average treatment effect for individuals with covariates $X=x$ and latent variable $U=u$:
$$ \text{MTE}(x,u) = \mathbb{E}[Y(1) - Y(0) \mid X=x, U=u]. $$

In policy analysis, we are often interested in the effect of an intervention that shifts the distribution of the instrument $Z$, or more generally, alters the propensity score from a baseline status to a new regime. Let the baseline policy be characterized by the current propensity score $p(Z,X)$ and the corresponding treatment status $W$. Consider an alternative policy that induces a new propensity score $q(Z,X)$ and a new counterfactual treatment status $W^q = \indicator(U \leqslant q(Z,X))$, where we assume the latent variable $U$ is invariant to the policy change, as is standard in the MTE framework.

The PRTE evaluates the average per-person impact of shifting from the baseline policy to the alternative policy. Let $Y^q$ denote the outcome realized under the alternative policy, such that $Y^q = W^qY(1) + (1-W^q)Y(0)$. The PRTE is formally defined as:
$$ \text{PRTE} = \frac{\mathbb{E}[Y^q - Y]}{\mathbb{E}[W^q - W]}. $$

Many common treatment effect parameters can be expressed as a weighted average of the MTE. Specifically, the PRTE can be rewritten as:
$$ \text{PRTE} = \mathbb{E}_{X,U}[\text{MTE}(X, U) \omega_{\text{PRTE}}(X,U)], $$
where the policy-specific weight function is given by:
$$ \omega_{\text{PRTE}}(x,u) = \frac{\mathbb{P}(q(Z,X) \geqslant u \mid X=x) - \mathbb{P}(p(Z,X) \geqslant u \mid X=x)}{\mathbb{E}[q(Z,X) - p(Z,X)]}, $$
provided $\mathbb{E}[q(Z,X) - p(Z,X)] \neq 0$.
Following \citet{heckman2005structural,magne2018using}, we consider the following more general target parameter.
\begin{align}\label{eq:target_mte}
    \theta_\omega = \mathbb{E}_X \int_0^1 \text{MTE}(X,u) \omega(X,u) \, \mathrm{d}u,
\end{align}
where $\omega(x,u)$ is an identifiable function. We can obtain different PRTEs by taking
\begin{align}\label{eq:omega_form}
    \omega(x,u) \propto \mathbb{E}[\indicator(u \leqslant q_1(Z,X)) - \indicator(u \leqslant q_0(Z,X)) \mid X=x],
\end{align}
for given propensity functions $q_0$ and $q_1$, where $q_0 = p$ and $q_1 = q$ correspond to the baseline and alternative propensity score functions, respectively, in the PRTE case. The coefficient hidden inside the proportionality symbol is identifiable. \cref{tab:causal_quantities} shows a variety of common causal parameters under different choices of $\omega$.
\begin{table}[htbp]
    \centering
    \caption{Common Causal Parameters as Weighted Averages of the MTE}
    \label{tab:causal_quantities}
    \renewcommand{\arraystretch}{1.8} % Increases row spacing for large fractions
    \begin{tabular}{lcc}
        \toprule
        \textbf{Target Parameter} & \textbf{Notation} & \textbf{Weight Function} $\omega(x,u)$ \\
        \midrule
        Average Treatment Effect & ATE & $1$ \\
        Average Treatment on the Treated & ATT & $ \frac{\mathbb{E}[\indicator (u\in(0,p(Z,X)))\mid X=x]}{\mathbb{E}[p(Z,X)]}$ \\
        Average Treatment on the Untreated & ATU & $ \frac{\mathbb{E}[\indicator (u\in(p(Z,X),1))\mid X=x]}{\mathbb{E}[1 - p(Z,X)]}$\\
        Local Average Treatment Effect\textsuperscript{a} & LATE$(z_1, z_0)$ & $ \frac{\indicator(u \in [p(z_0, x), p(z_1, x)])}{\mathbb{E}[p(z_1, X) - p(z_0, X)]}$ \\
        Policy-Relevant Treatment Effect & PRTE & $ \frac{\mathbb{P}(q(Z,X) \geqslant u \mid X=x) - \mathbb{P}(p(Z,X) \geqslant u \mid X=x)}{\mathbb{E}[q(Z,X) - p(Z,X)]}$ \\
        \bottomrule
        \multicolumn{3}{l}{\footnotesize \textsuperscript{a} Evaluated for a shift from instrument state $z_0$ to $z_1$, assuming $p(z_1, x) \geqslant p(z_0, x)$.}
    \end{tabular}
\end{table}

While our bounds recover the more classical bound for the first 4 quantities in \cref{tab:causal_quantities}, we will focus in particular on PRTEs. As a leading example, consider a uniform policy shift of the form $q(Z,X) = \min(p(Z,X) + \alpha, 1)$ for some $\alpha > 0$. This corresponds to a policy that uniformly raises the probability of treatment uptake by $\alpha$. In our empirical application (\cref{sec:simulation}), the instrument is the offered price of an insecticide-treated bed net \citep{dupas2014subsidies}, and $q = p + \alpha$ models a uniform price subsidy that lowers the cost of purchase and thereby increases take-up probability by $\alpha$.

\citet{heckman2001instrumental} shows that if the support of $p(Z,x)$ spans over the entire interval $(0,1)$ for all $x$, then MTE is identifiable. However, because the common support of $p(Z,X)$ is often a strict subset of $(0,1)$, the MTE is generally not identifiable. For instance, when both $X$ and $Z$ are discrete, the supports of $p(Z,x)$ are discrete and this common support assumption is violated. Consequently, the integral defining the PRTE is not point-identified without imposing strong parametric extrapolations.

\subsection{A Primer on Optimal Transport}\label{subsec:ot_primer}

To build the framework for our partial identification results, we briefly introduce the relevant concepts from OT theory. OT provides a mathematical framework for coupling two probability distributions at minimum cost, and has emerged as a powerful toolkit for bounding joint distributions when only their marginals are known \citep{galichon2016optimal}.

Let $\mu$ and $\nu$ be two probability measures on state spaces $\mathcal{A}$ and $\mathcal{B}$, respectively. In the context of partial identification, we often observe the marginal distributions of two variables but not their joint distribution. We define $\Pi(\mu, \nu)$ as the set of all \textit{couplings} between $\mu$ and $\nu$. Formally, a coupling $\pi \in \Pi(\mu, \nu)$ is a joint probability measure on the product space $\mathcal{A} \times \mathcal{B}$ such that its marginals are exactly $\mu$ and $\nu$. 

Given a cost function $c: \mathcal{A} \times \mathcal{B} \to \mathbb{R} \cup \{+\infty\}$ that represents the cost of associating state $a$ with state $b$, the standard Kantorovich optimal transport problem seeks the coupling that minimizes the expected cost:
\begin{align}
    \inf_{\pi \in \Pi(\mu, \nu)} \int_{\mathcal{A} \times \mathcal{B}} c(a, b) \, \mathrm{d}\pi(a, b). \label{eq:kantorovich}
\end{align} 

Generally, solving (\ref{eq:kantorovich}) in multidimensional spaces requires computationally intensive linear programming. However, a celebrated result in OT literature demonstrates that when the spaces are one-dimensional and the cost function satisfies specific structural conditions--such as supermodularity--the optimal coupling admits a sharp, closed-form solution. 

We formalize this one-dimensional closed-form solution in the following theorem, which will be instrumental in deriving the analytic bounds for the PRTE.

\begin{theorem}[1D Optimal Transport with Product Cost, e.g., \citep{villani2021topics, galichon2016optimal}] \label{thm:1d_ot}
    Let $\mu$ and $\nu$ be probability measures on $\mathbb{R}$ with quantile functions $Q_\mu$ and $Q_\nu$. Suppose the cost function is a simple product $c(a,b) = a \cdot b$. Then, the minimum and maximum of the expected cost over all valid couplings are completely characterized by the countermonotonic and comonotonic couplings, respectively:
    \begin{align*}
        \min_{\pi \in \Pi(\mu, \nu)} \int_{\mathbb{R}^2} a \cdot b \, \mathrm{d}\pi(a, b) &= \int_0^1 Q_\mu(t) Q_\nu(1-t) \, \mathrm{d}t, \\
        \max_{\pi \in \Pi(\mu, \nu)} \int_{\mathbb{R}^2} a \cdot b \, \mathrm{d}\pi(a, b) &= \int_0^1 Q_\mu(t) Q_\nu(t) \, \mathrm{d}t.
    \end{align*}
    The optimal coupling is attained by the countermonotonic coupling $(Q_\mu(\xi), Q_\nu(1-\xi))$ for the minimum, and the comonotonic coupling $(Q_\mu(\xi), Q_\nu(\xi))$ for the maximum, where $\xi \sim \mathrm{Unif}(0,1)$.
\end{theorem}

\cref{thm:1d_ot} provides a closed-form characterization of the optimal value of 1D OT problems. We will leverage this characterization to derive tight bounds for PRTEs later.

\subsection{Notation}\label{subsec:notation}

We use $\mathbb{P}$ and $\mathbb{E}$ for probability and expectation, and $\indep$ for conditional independence. $\indicator(\cdot)$ denotes the indicator function. We write $\mathrm{d}$ for the differential or Lebesgue measure. For a probability measure $\mu$ on $\mathbb{R}$, $Q_\mu(t)$ denotes its quantile function. $\Pi(\mu, \nu)$ denotes the set of all couplings of two probability measures $\mu$ and $\nu$. Calligraphic letters $\mathcal{X}, \mathcal{Y}, \mathcal{Z}, \mathcal{W}$ denote domains of the corresponding variables. We write $O_P(\cdot)$ and $o_P(\cdot)$ for stochastic order notation: $X_n = O_P(a_n)$ means $X_n/a_n$ is bounded in probability, and $X_n = o_P(a_n)$ means $X_n/a_n \xrightarrow{P} 0$.

%% file: sections/pid.tex
\section{Partial Identification via Optimal Transport}  \label{sec:pid}

\subsection{CCOT Formulation of Partial Identification}\label{subsec:ot_formulation}
In this section, we demonstrate how to derive closed-form, tight bounds for the PRTE. We begin by formally mapping the partial identification problem into an optimal transport framework.

The core intuition is to search over all possible joint distributions of the latent and observable variables $(Y(0), Y(1), U, X, Z, W)$ that are compatible with both the observed data and the structural conditions in \cref{asp:structure}. Because the potential outcomes are conditionally independent of the instrument given $X$, the identifying power of our model is entirely captured by the joint distributions of $(Y(w), U, X)$ for $w \in \{0, 1\}$.

Let $\pi_w$ denote the joint probability measure of $(Y(w), U, X)$ and $ \pi_w(\cdot\mid x,u) $ be the conditional distribution of $Y(w)$ given $x,u$. Recall that our target parameter is a linear functional of these measures:
$$\theta_\omega(\pi_0, \pi_1) = \mathbb{E}_{(X,U,Y(1))\sim\pi_1} [Y(1) \omega(X,U)] - \mathbb{E}_{(X,U,Y(0))\sim\pi_0} [Y(0) \omega(X,U)]. $$
To achieve sharp bounds, we must constrain $\pi_w$ using all available observational information. Let $\mathbb{P}_{obs}$ denote the true observed joint distribution of $(Y, W, Z, X)$. Under the threshold-crossing model in \cref{asp:structure}, while $U$ is generally unobserved, its conditional distribution is explicitly tied to the treatment status and the propensity score. Specifically, conditional on $X$ and $Z$, the event $W=1$ perfectly corresponds to $U \leqslant p(Z,X)$, and $W=0$ corresponds to $U > p(Z,X)$.

Consequently, the joint distributions $\pi_0$ and $\pi_1$ fully parameterize the observational distribution of $(X,Z,W,Y)$. To see this, we can link the unknown structural distributions directly to the observed conditional distribution of the data by exploiting the consistency equation $Y = W Y(1) + (1-W) Y(0)$ and the selection mechanism $W = \indicator(U \leqslant p(Z,X))$. For any measurable set $A \subseteq \mathcal{Y}$ and $W=1$, we have:
\begin{align*}
    \mathbb{P}_{obs}(Y \in A, W=1 \mid Z=z, X=x) &= \mathbb{P}(Y(1) \in A, U \leqslant p(z,x) \mid Z=z, X=x).
\end{align*}
By the conditional independence assumption $Z \indep (Y(1), U) \mid X$, we can omit the $Z$ in the second conditioning:
\begin{align*}
    \mathbb{P}(Y(1) \in A, U \leqslant p(z,x) \mid Z=z, X=x) &= \mathbb{P}(Y(1) \in A, U \leqslant p(z,x) \mid X=x).
\end{align*}
Since instruments sharing the same propensity score $p(z,x)$ yield the same right-hand side value, we can replace the conditioning on $Z$ with conditioning on $p(Z)$. Next, we disintegrate the joint measure $\pi_1$ into the conditional distribution of the potential outcome given the covariates and the latent variable. Because $U \mid X \sim \text{Unif}(0,1)$, integrating over the region $U \leqslant p(z,x)$ yields:
\begin{align*}
    \mathbb{P}_{obs}(Y \in A, W=1 \mid p(Z,X)=p, X=x) &= \int_0^{p} \pi_1(A \mid u, x) \mathrm{d}u.
\end{align*}
Following the exact same sequence of steps for $W=0$, where the observed outcome is $Y(0)$ and the selection mechanism dictates $U > p(z,x)$, we obtain the symmetric result:
\begin{align*}
    \mathbb{P}_{obs}(Y \in A, W=0 \mid p(Z,X)=p, X=x) &= \int_{p}^1 \pi_0(A \mid u, x) \mathrm{d}u.
\end{align*}
Furthermore, by construction, the marginal distribution of $(U, X)$ under $\pi_w$ must equal the known population distribution $\mathbb{P}_{U,X}$.

This allows us to formulate the partial identification of the PRTE as the following infinite dimensional linear programming problem:
\begin{equation}\label{eq:ot_formulation}
    \begin{split}
        \max_{\pi_0, \pi_1} \ / \min_{\pi_0, \pi_1}\quad & \theta_\omega(\pi_0, \pi_1)\\
        \text{subject to}  \quad &  \int_0^{p(z,x)} \pi_1(\mathrm{d}y \mid u, x) \mathrm{d}u = \mathbb{P}_{obs}(\mathrm{d}y, W=1 \mid p(Z,X)=p(z,x),X= x) ,\\
        & \int_{p(z,x)}^1 \pi_0(\mathrm{d}y \mid u, x) \mathrm{d}u = \mathbb{P}_{obs}(\mathrm{d}y, W=0 \mid  p(Z,X)=p(z,x), X=x), \forall z\in\mathcal{Z}, x\in\mathcal{X},\\
        & \pi_w(\mathrm{d}x, \mathrm{d}u) = \mathrm{d}u \ \mathbb{P}_{obs}(\mathrm{d}x) \quad \text{for } w \in \{0,1\},\\
        & \pi_w \text{ is supported on } \mathcal{Y} \times [0,1] \times \mathcal{X} \quad \text{for } w \in \{0,1\}.
    \end{split}
\end{equation}
We call the optimization in \eqref{eq:ot_formulation} the Constrained Conditional Optimal Transport (CCOT) problem. The feasible set is constrained by the observational compatibility conditions in \eqref{eq:ot_formulation}; the problem is conditional in that, for fixed covariates $x$ and latent resistance $u$, it seeks a conditional measure $\pi_w(\cdot \mid u, x)$ on $\mathcal{Y}$; and it is a transport problem because the objective is linear in $\pi_w(\cdot \mid u, x)$ with effective cost $c(y, u, x) = y \cdot \omega(x, u)$, so that mass is transported from the latent space $U$ to the outcome space $Y(w)$ conditional on $X$.

To formally guarantee that the CCOT formulation in (\ref{eq:ot_formulation}) yields sharp bounds for the PRTE, we must show that any pair of measures $(\pi_0, \pi_1)$ satisfying the CCOT constraints corresponds to a valid data-generating process that satisfies our structural assumptions and perfectly reproduces the observed data. Let $\Gamma(\mathbb{P}_{obs})$ be the set of all measure pairs $(\pi_0, \pi_1)$ satisfying the constraints in (\ref{eq:ot_formulation}). The following proposition characterizes the joint distribution of $(Y(0), Y(1), U, X)$, from which observational
equivalence of the marginal pair $(\pi_0, \pi_1)$ follows immediately.

\begin{proposition}[Sharp Identified Set of the Joint Latent Law] \label{prop:sharpness}
    Suppose the observed distribution $\mathbb{P}_{obs}$ is generated by a true structural model satisfying \cref{asp:structure}. Let $\tilde{\pi}$ be any probability measure on $\mathcal{Y} \times \mathcal{Y} \times [0,1] \times \mathcal{X}$ whose $(Y(w), U, X)$-marginals $\tilde{\pi}_w$ satisfy $(\tilde{\pi}_0, \tilde{\pi}_1) \in \Gamma(\mathbb{P}_{obs})$. Then there exists a joint distribution of the latent and observable variables $(Y(0), Y(1), U, X, Z, W)$ such that:
    \begin{enumerate}
        \item[(i)] the joint marginal distribution of $(Y(0), Y(1), U, X)$ is exactly $\tilde{\pi}$;
        \item[(ii)] \cref{asp:structure} holds for $(Y(0), Y(1), U, X, Z, W)$;
        \item[(iii)] the induced distribution of $(Z, X, W, Y)$ exactly matches $\mathbb{P}_{obs}$.
    \end{enumerate}
    Conversely, every joint marginal $\tilde{\pi}$ arising from a distribution of $(Y(0), Y(1), U, X, Z, W)$ satisfying \textup{(ii)}--\textup{(iii)} has $(Y(w), U, X)$-marginals in $\Gamma(\mathbb{P}_{obs})$ for $w \in \{0, 1\}$.
\end{proposition}

The content of \cref{prop:sharpness} is twofold. First, the CCOT constraints act only on the two marginals $\tilde{\pi}_0$ and $\tilde{\pi}_1$: \emph{any} coupling of these marginals---including every joint law of $(Y(0), Y(1))$ given $(U, X)$---is observationally equivalent. Second, the converse makes the characterization sharp: no joint law outside this set can be generated by a structural model matching $\mathbb{P}_{obs}$. In particular, applying \cref{prop:sharpness} to the independent coupling $\tilde{\pi} := \tilde{\pi}_0 \otimes_{(U, X)} \tilde{\pi}_1$ shows that every marginal pair $(\tilde{\pi}_0, \tilde{\pi}_1) \in \Gamma(\mathbb{P}_{obs})$ is observationally equivalent to the true data-generating process. Consequently, the identified set for the PRTE target is $\{\theta_\omega(\pi_0, \pi_1) : (\pi_0, \pi_1) \in \Gamma(\mathbb{P}_{obs})\}$, and the maximum and minimum of $\theta_\omega(\pi_0, \pi_1)$ over $\Gamma(\mathbb{P}_{obs})$ are the sharp upper and lower bounds on the PRTE.

\begin{remark}
The fact that the moment relaxation approach of \citet{magne2018using} can yield arbitrarily loose bounds compared to our CCOT formulation is well-recognized; see also \citet{han_yang2024computational}. A detailed comparison example is given in \cref{exp:loose_bd} in the appendix.
\end{remark}

\begin{remark}[Beyond PRTE: CCOT for general joint functionals]\label{rem:general_functional}
The CCOT perspective is not specific to the PRTE. At a conceptual level, one could consider a bounded measurable cost $f:\mathcal{Y}^2 \to \mathbb{R}$ and the joint target $\theta_f := \mathbb{E}[f(Y(1),Y(0))]$ by replacing the PRTE objective in (\ref{eq:ot_formulation}) with $\mathbb{E}_{\pi}[f(Y(1),Y(0))]$ and enlarging the decision variable from the pair $(\pi_0,\pi_1)$ to a joint conditional kernel $\pi(\diff y_1,\diff y_0 \mid u,x)$ whose marginals satisfy the same observational constraints. A sharp characterization should again follow from the lifting logic behind \cref{prop:sharpness}. What is special about the PRTE is the separable product cost $y \cdot \omega(x,u)$: it allows the CCOT problem to collapse into one-dimensional OT problems for each potential outcome separately. For a general joint functional, that separability is lost, and the resulting interval-specific problems are instead OT problems over the coupling between the identified marginals of $Y(1)$ and $Y(0)$, with no comparable closed-form reduction in general. Since the present paper is about the PRTE case, we do not pursue this extension here.
\end{remark}

Next, we show how to obtain the closed-form solution of the CCOT problem (\ref{eq:ot_formulation}).
Notice that in (\ref{eq:ot_formulation}), the objective is the difference of two linear functionals of $\pi_0$ and $\pi_1$, and the constraints on $\pi_0$ and $\pi_1$ are separate. Therefore, we only need to consider one side of the problem, as the other side follows identically. In what follows, we only consider the $W=1$ side minimization problem, i.e., 
\begin{equation}\label{eq:ot_formulation_1}
    \begin{split}
    \min_{ \pi_1}\quad &  \mathbb{E}_{\pi_1}[Y(1)\omega(X,U)]\\
        \text{subject to}  \quad &  \int_0^{p(z,x)} \pi_1(\mathrm{d}y \mid u, x) \mathrm{d}u = \mathbb{P}_{obs}(\mathrm{d}y, W=1 \mid  p(Z,X)= p(z,x), X = x) ,\  \forall z\in\mathcal{Z}, \forall x\in\mathcal{X}\\ 
        & \pi_1(\mathrm{d}x, \mathrm{d}u) = \mathrm{d}u \ \mathbb{P}_{obs}(\mathrm{d}x),\\
        &\pi_1 \text{ is supported on } \mathcal{Y} \times [0,1] \times \mathcal{X}.
    \end{split}
\end{equation} 
The solution for the maximization problem can be derived similarly. Writing $\theta_\omega = \theta_{\omega,1} - \theta_{\omega,0}$, let $[\underline{\theta}_{\omega,w}, \overline{\theta}_{\omega,w}]$ denote the sharp identified interval for the $w$-th component. Because the feasible sets for $\pi_0$ and $\pi_1$ are Cartesian products, the sharp identified interval for the full target is
\[
[\underline{\theta}_\omega, \overline{\theta}_\omega]
=
[\underline{\theta}_{\omega,1} - \overline{\theta}_{\omega,0},\,
\overline{\theta}_{\omega,1} - \underline{\theta}_{\omega,0}].
\]
Accordingly, we state the main derivations for $\theta_{\omega,1}$; the formulas for $\theta_{\omega,0}$ are obtained by the same arguments with the untreated analogues. 
\subsection{Closed-Form Bounds without Covariates}\label{subsec:cf_no_cov}

To better illustrate the idea, we first derive the results without covariates and generalize them in the next subsection. Therefore, in this subsection, $\omega$ and $p$ do not depend on $x$. 

\subsubsection{Discrete Instrument Setting}\label{subsubsec:discrete_iv}
We first consider the setting where there are no covariates and the instrument $Z$ takes values in a finite discrete set $\mathcal{Z}$. Let the unique values of the propensity score be ordered as:
\[
    \mathrm{Range}(p(z)) := \{p(z) : z \in \mathcal{Z}\} = \{p_1, \dots, p_K\},
\]
with the conventions $p_0 := 0$ and $p_{K+1} := 1$, such that $0 = p_0 \leqslant p_1 < p_2 < \dots < p_K \leqslant p_{K+1} = 1$. We allow different instruments to be mapped to the same propensity value. This discretizes the unit interval of the latent variable $U$ into $K+1$ disjoint sub-intervals, $I_i = (p_i, p_{i+1}]$ for $i=0, \dots, K$.

Because the objective function (\ref{eq:ot_formulation_1}) is an integral over $U$, we can additively decompose it across these sub-intervals:
\begin{align*}
    \mathbb{E}_{\pi_1}[Y(1)\omega(U)] = \sum_{i=0}^K (p_{i+1} - p_i) \mathbb{E}_{\pi_1}[Y(1)\omega(U) \mid U \in I_i].
\end{align*}

The fundamental advantage of the discrete setting is that the global observational constraints difference out, perfectly isolating the conditional distribution of $Y(1)$ within each interval $I_i$. Recall that from (\ref{eq:ot_formulation_1}),
\begin{align}\label{eq:ot_y_constraint}
    \int_0^{p_i} \pi_{1}(\mathrm{d}y \mid u) \mathrm{d}u = \mathbb{P}_{obs}(\mathrm{d}y, W=1 \mid p(Z)=p_i), \quad i=1,\dots,K.
\end{align}
The key insight of this work is that (\ref{eq:ot_y_constraint}) can be decomposed into separate marginal constraints, which enables us to convert (\ref{eq:ot_formulation_1}) into a series of seperate standard one-dimensional OT problems. For $i = 1, \dots, K-1$, subtracting the $i$-th constraint from the $(i+1)$-th in (\ref{eq:ot_y_constraint}) isolates the integral over $(p_i, p_{i+1}]$; for $i=0$, the first constraint directly governs $(0, p_1]$. Since cumulative sums recover the original constraints, this invertible transformation shows that the system (\ref{eq:ot_y_constraint}) is equivalent to
\begin{align*}
   \frac{1}{p_{i+1}-p_i} \int^{p_{i+1}}_{p_{i}} \pi_{1}(\mathrm{d}y \mid u) \mathrm{d}u = \mu_{1,i}(\mathrm{d}y) , \quad i=0,\dots,K-1,
\end{align*}
where $\mu_{1,i}$ is given by
\begin{align}\label{eq:y_identified_measure}
    \mu_{1, i}(\mathrm{d}y) = 
    \begin{cases}
         \frac{\mathbb{P}_{obs}(\mathrm{d}y, W=1 \mid p(Z) = p_1 )}{p_1}, & \text{if } i=0, \\[10pt]
         \frac{\mathbb{P}_{obs}(\mathrm{d}y, W=1 \mid p(Z) = p_{i+1} ) - \mathbb{P}_{obs}(\mathrm{d}y, W=1 \mid p(Z) = p_{i} )}{p_{i+1} - p_i}, & \text{if } 0 < i < K.
    \end{cases}
\end{align}

Note that for the final interval $I_K = (p_K, 1]$, the potential outcome $Y(1)$ is never observed because no individual in the population has a propensity score high enough to induce treatment when $U > p_K$. By (\ref{eq:ot_y_constraint}), since $\pi_1$ is nonnegative, $\mathbb{P}_{obs}(\mathrm{d}y, W=1 \mid p(Z)=p)$ is non-decreasing in $p$. Note that $\mu_{1,i}(\mathcal{Y}) = 1$ since $\mathbb{P}_{obs}(W=1 \mid p(Z)=p) = p$. Therefore, the  $\mu_{1,i}$  are valid distributions. \citet{mourifie2017testing} provide a simple method to test the monotonicity of $\mathbb{P}_{obs}(\mathrm{d}y, W=1 \mid p(Z)=p)$.

For $0<i<K$, notice that 
\begin{align*}
   \mu_{1,i}(A) &= \frac{\mathbb{P}_{obs}(Y\in A, W=1 \mid p(Z) = p_{i+1} ) - \mathbb{P}_{obs}(Y\in A, W=1 \mid p(Z) = p_{i} )}{p_{i+1} - p_i}\\
   & = \frac{\mathbb{P}(Y(1)\in A, U\in(0,p_{i+1}) \mid p(Z) = p_{i+1} ) - \mathbb{P}(Y(1)\in A, U\in(0,p_{i})\mid p(Z) = p_{i} )}{p_{i+1} - p_i}\\
   & = \mathbb{P}(Y(1)\in A \mid U\in(p_i,p_{i+1}) ),
\end{align*}
for any measurable set $A$, where we use \cref{asp:structure}.(3) and consistency in the second equality and \cref{asp:structure}.(2) in the last equality.  Therefore, $\mu_{1,i}$ can be interpreted as the distribution of the potential outcome under treatment for the $i$-th complier group. The $i$-th complier group consists of units that refuse treatment when $p(Z) \leqslant p_i$ and comply only when $p(Z) > p_i$. In particular, when the instrument is binary, this recovers the identification results of \citet{imbens1997estimating}

Substituting the objective decomposition and the interval-wise constraints derived above, the original problem (\ref{eq:ot_formulation_1}) can be equivalently written as
\begin{equation}\label{eq:ot_decomposed}
    \min_{\pi_1} \quad \sum_{i=0}^{K} (p_{i+1} - p_i) \int_{\mathcal{Y} \times I_i} y \cdot \omega(u) \, \pi_{1}(\mathrm{d}y \mid u) \frac{\mathrm{d}u}{p_{i+1} - p_i}
\end{equation}
subject to the $K$ independent marginal constraints
\begin{equation}\label{eq:ot_decomposed_constraint}
    \begin{split}
        &\frac{1}{p_{i+1} - p_i}\int_{p_i}^{p_{i+1}} \pi_{1}(\mathrm{d}y \mid u) \mathrm{d}u = \mu_{1,i}(\mathrm{d}y),\quad i = 0, \dots, K-1, \\
    &\pi_1(\mathrm{d}u)=\mathrm{d}u,
    \end{split}
\end{equation}
with the support constraint $\text{supp}(\pi_1(\cdot \mid u)) \subseteq \mathcal{Y}$ for all $u \in [0,1]$, and no further distributional constraint on $\pi_1(\cdot \mid u)$ for $u \in I_K = (p_K, 1]$. Crucially, each constraint in (\ref{eq:ot_decomposed_constraint}) involves $\pi_1$ only on the single interval $I_i$, and the objective (\ref{eq:ot_decomposed}) is additively separable across these intervals. Therefore, the global minimization problem decomposes into $K+1$ independent subproblems, one for each interval.

Specifically, for each interval $i = 0, \dots, K-1$, the data perfectly fixes the marginal distribution of $Y(1)$ as $\mu_{1, i}$ and the marginal distribution of $U$ as uniform on $(p_i, p_{i+1}]$. Let $\nu_i$ denote the uniform probability measure $\text{Unif}(p_i, p_{i+1})$. The set of valid joint distributions for $(Y(1), U)$ conditional on $U \in I_i$ is exactly the set of all couplings $\Pi(\mu_{1, i}, \nu_i)$. Therefore, for each of these $K$ intervals, we solve a standard 1D optimal transport problem:
\begin{equation} \label{eq:ot_subproblem}
    \min_{\gamma_i \in \Pi(\mu_{1, i}, \nu_i)} \int_{\mathcal{Y} \times I_i} y \cdot \omega(u) \, \mathrm{d}\gamma_i(y, u).
\end{equation}

Conversely, for the final interval $I_K = (p_K, 1]$, the propensity score never reaches a value high enough to induce treatment, meaning the data provides absolutely no constraints on the distribution of $Y(1)$ for this subpopulation. Without distributional constraints, the optimal transport problem degenerates into a trivial pointwise minimization. By \cref{asp:bounded_y}, for each $u \in I_K$, the optimal solution assigns all probability mass to $y_{\min}$ if $\omega(u) \geqslant 0$, or to $y_{\max}$ if $\omega(u) < 0$:
\begin{equation} \label{eq:ot_trivial}
    \min_{\pi_{1}} \int_{p_K}^1 \int_{\mathcal{Y}} y \, \omega(u) \, \pi_{1}(\mathrm{d}y \mid u) \mathrm{d}u \ = \ \int_{p_K}^1 \Big( y_{\min} \max\{0, \omega(u)\} + y_{\max} \min\{0, \omega(u)\} \Big) \mathrm{d}u.
\end{equation}
Summing the optimal values of the $K$ constrained subproblems in (\ref{eq:ot_subproblem}) scaled by their respective interval lengths $(p_{i+1} - p_i)$ and adding the trivial unconstrained bound from (\ref{eq:ot_trivial}), we obtain the global minimum. Applying \cref{thm:1d_ot} to each subproblem in (\ref{eq:ot_subproblem}) yields the countermonotonic coupling, which leads directly to the following closed-form sharp bounds.

\begin{theorem}[Closed-Form Sharp Bounds]\label{thm:discrete_lower_bound}
    Suppose \cref{asp:structure} and \cref{asp:bounded_y} hold, $\mathcal{X} = \emptyset$, and $\mathcal{Z}$ is finite. For $i=0, \dots, K-1$, let $Q_{Y, i}(t)$ be the quantile function of the identified conditional distribution $\mu_{1, i}$ defined in (\ref{eq:y_identified_measure}). Let $Q_{\omega, i}(t)$ be the quantile function of the random variable $\omega(U)$ where $U \sim \text{Unif}(p_i, p_{i+1})$. The sharp lower bound for the $W=1$ component of the PRTE is given by:
    \begin{equation}\label{eq:discrete_lower_bound}
        \underline{\theta}_{\omega, 1} = \sum_{i=0}^{K-1} (p_{i+1} - p_i) \int_0^1 Q_{Y, i}(t) Q_{\omega, i}(1-t) \, \mathrm{d}t \ + \ \int_{p_K}^1 \Big( y_{\min} \max\{0, \omega(u)\} + y_{\max} \min\{0, \omega(u)\} \Big) \mathrm{d}u.
    \end{equation}
    Similarly, the sharp upper bound is given by replacing $Q_{\omega,k}(1-t)$ with $Q_{\omega,k}(t)$ in the OT term and swapping $y_{\min}$ and $y_{\max}$ in the trivial tail term.

    % \begin{equation}
    %     \overline{\theta}_{\omega, 1} = \sum_{i=0}^{K-1} (p_{i+1} - p_i) \int_0^1 Q_{Y, i}(t) Q_{\omega, i}(t) \, \mathrm{d}t \ + \ \int_{p_K}^1 \Big( y_{\max} \max\{0, \omega(u)\} + y_{\min} \min\{0, \omega(u)\} \Big) \mathrm{d}u.
    % \end{equation}
\end{theorem}

We defer the formal proof of this theorem to \cref{thm:covariate_continuous_lower_bound}, since the present theorem is a special case ($\mathcal{X}=\emptyset$, every LIV interval degenerate). \cref{thm:discrete_lower_bound} provides a closed-form, analytically tractable expression for the sharp bounds, bypassing any need for linear programming and requiring only estimation of the conditional quantile function from the observed data.

As a special case, if we set $ \omega(u) = 1 $, \cref{thm:discrete_lower_bound} recovers the classical bound by \citet{manski1990nonparametric} for $\mathbb{E}[Y(1)]$, which is known to be tight \citep{heckman2001instrumental}:
\begin{align*}
    \underline{\theta}_{\omega, 1} &= \sum_{i=0}^{K-1} (p_{i+1} - p_i) \int_0^1 Q_{Y, i}(t) Q_{\omega, i}(1-t) \, \mathrm{d}t \ + \ y_{\min} \int_{p_K}^1 \omega(u) \, \mathrm{d}u\\
    & = \sum_{i=0}^{K-1} (p_{i+1} - p_i) \int_0^1 Q_{Y, i}(t) \, \mathrm{d}t \ + \ y_{\min} (1-p_K)\\
    & = p_K \mathbb{E}[Y \mid W = 1, p(Z) = p_K] \ + \ y_{\min} (1-p_K).
\end{align*}
\cref{fig:decomposition} illustrates the three identification regions that compose the bound in \cref{thm:discrete_lower_bound}; \cref{ex:discrete_bd} in the appendix works through a concrete binary-instrument example with a constant propensity score, deriving the explicit three-part decomposition.

\begin{figure}[!htbp]
\centering
\begin{tikzpicture}[scale=1.0, font=\small]
  %% --- Layout (K=3 propensity levels for illustration) ---
  \def\W{12.0}    % full width (U = 1)
  \def\SH{1.10}   % main strip height
  \def\pA{2.5}    % p_1
  \def\pB{6.0}    % p_2
  \def\pC{9.2}    % p_3
  \def\qS{7.5}    % q* strictly between p_2 and p_3 (intermediate step in omega)
  % omega panel: baseline at y=2.20, levels measured upward
  \def\omY{2.20}  % omega y-baseline (omega = 0 line)
  \def\omA{1.80}  % omega height on I_0
  \def\omB{1.10}  % omega height on I_1
  \def\omC{0.65}  % omega height on (p_2, q*] subset of I_2
  \def\omD{0.35}  % omega height on (q*, p_3] subset of I_2
  \def\omE{0.16}  % omega height on trivial region I_3 (non-zero!)

  %% --- Main identification strip ---
  % I_0, I_1: LIV (omega constant on interval -> integral identifiable)
  \fill[teal!12]   (0,0)    rectangle (\pA,\SH);
  \fill[teal!22]   (\pA,0)  rectangle (\pB,\SH);
  % I_2: OT (omega has interior step at q* -> need quantile pairing)
  \fill[blue!22]   (\pB,0)  rectangle (\pC,\SH);
  % I_3: trivial bound
  \fill[orange!15] (\pC,0)  rectangle (\W,\SH);
  \draw[thick]     (0,0)    rectangle (\W,\SH);
  \draw[thick, teal!55]   (\pA,0) -- (\pA,\SH);
  \draw[thick, teal!55]   (\pB,0) -- (\pB,\SH);
  \draw[thick, orange!60] (\pC,0) -- (\pC,\SH);

  % Region labels inside strip
  \node[teal!70!black,  font=\footnotesize\bfseries] at (\pA/2,         0.70) {LIV};
  \node[teal!70!black,  font=\scriptsize]            at (\pA/2,         0.28) {$I_0=(0,p_1]$};
  \node[teal!70!black,  font=\footnotesize\bfseries] at ({(\pA+\pB)/2}, 0.70) {LIV};
  \node[teal!70!black,  font=\scriptsize]            at ({(\pA+\pB)/2}, 0.28) {$I_1=(p_1,p_2]$};
  \node[blue!80!black,  font=\footnotesize\bfseries] at ({(\pB+\pC)/2}, 0.70) {OT Region};
  \node[blue!80!black,  font=\scriptsize]            at ({(\pB+\pC)/2}, 0.28) {$I_2=(p_2,p_3]$};
  \node[orange!80!black, font=\scriptsize\bfseries, align=center]
      at ({(\pC+\W)/2}, \SH/2) {Trivial\\[-1pt]Bound};

  %% --- Identified-measure markers at propensity levels ---
  % p_1
  \draw[blue!55, thick] (\pA,\SH) -- (\pA,{\SH+0.30});
  \filldraw[fill=white, draw=blue!65, thick] (\pA,{\SH+0.30}) circle (3pt);
  \node[above, font=\scriptsize, blue!70!black, align=center] at (\pA,{\SH+0.30})
      {$p_1$\\[-1pt]{\tiny identified}};
  % p_2
  \draw[blue!55, thick] (\pB,\SH) -- (\pB,{\SH+0.30});
  \filldraw[fill=white, draw=blue!65, thick] (\pB,{\SH+0.30}) circle (3pt);
  \node[above, font=\scriptsize, blue!70!black, align=center] at (\pB,{\SH+0.30})
      {$p_2$\\[-1pt]{\tiny identified}};
  % p_3
  \draw[blue!55, thick] (\pC,\SH) -- (\pC,{\SH+0.30});
  \filldraw[fill=white, draw=blue!65, thick] (\pC,{\SH+0.30}) circle (3pt);
  \node[above, font=\scriptsize, blue!70!black, align=center] at (\pC,{\SH+0.30})
      {$p_3$\\[-1pt]{\tiny identified}};

  %% --- U axis ---
  \draw[->,thick] (-0.3,-0.50) -- (\W+0.6,-0.50) node[right]{$U$};
  \draw (0,   -0.38) -- (0,   -0.60); \node[below,font=\scriptsize] at (0,   -0.60) {$0$};
  \draw (\pA, -0.38) -- (\pA, -0.60); \node[below,font=\scriptsize] at (\pA, -0.60) {$p_1$};
  \draw (\pB, -0.38) -- (\pB, -0.60); \node[below,font=\scriptsize] at (\pB, -0.60) {$p_2$};
  % q* tick (shorter, different style to distinguish from propensity levels)
  \draw[gray!65] (\qS,-0.38) -- (\qS,-0.55);
  \node[below,font=\scriptsize,gray!75] at (\qS,-0.55) {$q^*$};
  \draw (\pC, -0.38) -- (\pC, -0.60); \node[below,font=\scriptsize] at (\pC, -0.60) {$p_3$};
  \draw (\W,  -0.38) -- (\W,  -0.60); \node[below,font=\scriptsize] at (\W,  -0.60) {$1$};

  %% --- Dashed vertical guides through all panels ---
  \draw[gray!40,dashed,thin] (\pA,-0.60) -- (\pA,{\omY+\omA+0.25});
  \draw[gray!40,dashed,thin] (\pB,-0.60) -- (\pB,{\omY+\omA+0.25});
  % q* guide (dotted, lighter -- not a propensity level)
  \draw[gray!35,dotted,thin] (\qS,-0.55) -- (\qS,{\omY+\omA+0.25});
  \draw[gray!40,dashed,thin] (\pC,-0.60) -- (\pC,{\omY+\omA+0.25});

  %% --- Omega piecewise-constant panel ---
  \draw[->,thick] (-0.3,\omY) -- (-0.3,{\omY+\omA+0.55})
      node[above,font=\scriptsize]{$\omega(u)$};
  \draw[gray!30,thin] (0,\omY) -- (\W,\omY);       % zero baseline
  \node[left,font=\scriptsize,gray] at (-0.3,\omY) {$0$};

  % y-axis tick marks
  \draw[gray!45] (-0.42,{\omY+\omA}) -- (-0.18,{\omY+\omA});
  \node[left,font=\scriptsize,gray!60] at (-0.42,{\omY+\omA}) {$\omega_0$};
  \draw[gray!45] (-0.42,{\omY+\omB}) -- (-0.18,{\omY+\omB});
  \node[left,font=\scriptsize,gray!60] at (-0.42,{\omY+\omB}) {$\omega_1$};
  \draw[gray!35] (-0.42,{\omY+\omC}) -- (-0.18,{\omY+\omC});  % unlabeled
  \draw[gray!35] (-0.42,{\omY+\omD}) -- (-0.18,{\omY+\omD});  % unlabeled
  \draw[gray!45] (-0.42,{\omY+\omE}) -- (-0.18,{\omY+\omE});
  \node[left,font=\scriptsize,gray!60] at (-0.42,{\omY+\omE}) {$\omega_4$};

  % piecewise-constant step function (non-increasing, with intermediate step at q*)
  \draw[thick, violet!80!black]
      (0,    {\omY+\omA}) -- (\pA, {\omY+\omA})   % constant on I_0
      -- (\pA, {\omY+\omB})                         % drop at p_1
      -- (\pB, {\omY+\omB})                         % constant on I_1
      -- (\pB, {\omY+\omC})                         % drop at p_2
      -- (\qS, {\omY+\omC})                         % constant on (p_2, q*]
      -- (\qS, {\omY+\omD})                         % drop at q* (strictly between p_2 and p_3!)
      -- (\pC, {\omY+\omD})                         % constant on (q*, p_3]
      -- (\pC, {\omY+\omE})                         % drop at p_3 (to non-zero level!)
      -- (\W,  {\omY+\omE});                        % non-zero weight in trivial region!

  %% --- Connecting arrows (strip to bound labels) ---
  \draw[->,dashed,teal!50]    (\pA/2,         0) -- (\pA/2,         -1.00);
  \draw[->,dashed,teal!50]    ({(\pA+\pB)/2}, 0) -- ({(\pA+\pB)/2}, -1.00);
  \draw[->,dashed,blue!40]    ({(\pB+\pC)/2}, 0) -- ({(\pB+\pC)/2}, -1.00);
  \draw[->,dashed,orange!55]  ({(\pC+\W)/2},  0) -- ({(\pC+\W)/2},  -1.00);

  %% --- Brief per-region bound labels ---
  % LIV regions: omega constant -> integral is directly identifiable mean
  \node[teal!70!black, font=\scriptsize, align=center]
      at (\pA/2,         -1.38) {direct\\mean};
  \node[teal!70!black, font=\scriptsize, align=center]
      at ({(\pA+\pB)/2}, -1.38) {direct\\mean};
  % OT region: omega has interior step -> quantile pairing needed
  \node[blue!75!black, font=\scriptsize, align=center]
      at ({(\pB+\pC)/2}, -1.38) {quantile\\pairing};
  \node[orange!80!black, font=\scriptsize, align=center]
      at ({(\pC+\W)/2},  -1.38) {assign $y_{\min}$};

  % plus signs between labels
  \node[font=\footnotesize] at ({(\pA/2+(\pA+\pB)/2)/2},        -1.38) {$+$};
  \node[font=\footnotesize] at ({((\pA+\pB)/2+(\pB+\pC)/2)/2},  -1.38) {$+$};
  \node[font=\footnotesize] at ({((\pB+\pC)/2+(\pC+\W)/2)/2},   -1.38) {$+$};

  %% --- Annotation explaining LIV vs OT distinction in omega panel ---
  % Arrow at p_1 drop: "drop at p_k -> LIV"
  \draw[->,gray!60,thin] (\pA+0.15,{\omY+(\omA+\omB)/2}) -- (\pA+1.2,{\omY+(\omA+\omB)/2});
  \node[right,font=\tiny,teal!65!black] at (\pA+1.2,{\omY+(\omA+\omB)/2})
      {drop at $p_k$ $\Rightarrow$ \textbf{LIV}};
  % Arrow at q* drop: "drop at q* -> OT"
  \draw[->,gray!60,thin] (\qS+0.15,{\omY+(\omC+\omD)/2}) -- (\qS+1.0,{\omY+(\omC+\omD)/2});
  \node[right,font=\tiny,blue!70!black] at (\qS+1.0,{\omY+(\omC+\omD)/2})
      {drop at $q^*$ $\Rightarrow$ \textbf{OT}};

  %% --- Global sharp lower bound formula ---
  \node[align=center, font=\footnotesize] at ({\W/2}, -2.30) {
    $
    \underline{\theta}_{\omega,1}
    \;=\; \sum_{i=0}^{K-1}(p_{i+1}-p_i)
      \int_0^1 Q_{Y,i}(t)\,Q_{\omega,i}(1-t)\,\mathrm{d}t
    \;+\; \int_{p_K}^{1}\!\bigl(y_{\min}[\omega]_+\!+y_{\max}[\omega]_-\bigr)\mathrm{d}u$
  };

\end{tikzpicture}
\caption{Decomposition of the unit interval $[0,1]$ of the latent variable $U$ for the sharp lower bound in \cref{thm:discrete_lower_bound} ($K=3$ propensity levels). The key distinction is driven by where $\omega(u)$ drops. \textbf{\color{teal!65}LIV regions} ($I_0$, $I_1$): $\omega$ is constant on the entire interval because its step falls exactly at a propensity level $p_k$; the integral therefore reduces to a directly identifiable conditional mean $\omega_k\,\mathbb{E}_{\mu_{1,k}}[Y]$. \textbf{\color{blue!65}OT region} ($I_2$): $\omega$ has an interior step at $q^*\!\in\!(p_2,p_3)$ (dotted guide), so the bound depends on how the mass of $\mu_{1,2}$ is distributed relative to $q^*$; the tight bound requires countermonotone quantile coupling. \textbf{\color{orange!65}Trivial bound} ($I_3$): $Y(1)$ is never observed; the lower bound assigns $y_{\min}$ where $\omega(u)\geqslant 0$ and $y_{\max}$ where $\omega(u)<0$. The top panel shows the piecewise-constant $\omega(u)$; the bottom row gives the bound contribution per region.}
\label{fig:decomposition}
\end{figure}

By \cref{thm:discrete_lower_bound}, the length of the bound is
\begin{align*}
\overline{\theta }_{\omega ,1} -\underline{\theta }_{\omega ,1} & =\sum _{i=0}^{K-1} (p_{i+1} -p_{i} )\int _{0}^{1}( Q_{Y,i} (t)-Q_{Y,i}( 1-t)) Q_{\omega ,i} (t) \mathrm{d} t + ( y_{\max} -y_{\min})\int _{p_{K}}^{1} |\omega (u)| \mathrm{d} u
\end{align*}
which depends on the range of the propensity function and distribution $\mu_{1,i}$. If $p_K$ is close to $1$ and the support of $\mu_{1,i}$ is small, the length of the bound is also small. In particular, if each $Q_{Y,i}$ is constant and $p_K = 1,p_0 = 0$, $\theta_{\omega,1}$ is identifiable from the data. 

\begin{remark}[Decomposition into 1D marginal constraints and \citet{marx2024sharp}] \label{rem:marx-open}
\citet{marx2024sharp} raises the open question of characterizing the distribution of treatment effects for the complier group. Although \cref{prop:sharpness} states sharpness at the joint level of $(Y(0), Y(1), U, X)$, the CCOT constraints defining $\Gamma(\mathbb{P}_{obs})$ reduce to independent one-dimensional marginal constraints on the conditional kernels $\pi_w(\diff y \mid u, x)$. Consequently, once these marginals are pinned down, the remaining freedom lies entirely in the coupling between $Y(1)$ and $Y(0)$ given $(U, X)$. In particular, in the binary-instrument, no-covariate setting with propensity scores $p_1 < p_2$, the complier interval is $I_c=(p_1,p_2]$. By (\ref{eq:y_identified_measure}), the treated-outcome distribution for compliers is identified as $\mu_{1,1}$. Symmetrically, the untreated-outcome distribution for the same complier group is identified as
\[
\mu_{0,1}(\diff y)
:=
\frac{\mathbb{P}_{obs}(\diff y, W=0 \mid p(Z)=p_1)-\mathbb{P}_{obs}(\diff y, W=0 \mid p(Z)=p_2)}{p_2-p_1}.
\]
Hence the identified set for the complier treatment-effect distribution is exactly
\[
\mathcal I_{\Delta,c}
=
\left\{
(y_1-y_0)_{\#}\gamma
:
\gamma \in \Pi(\mu_{1,1},\mu_{0,1})
\right\}.
\]
Thus, Marx's question can be restated as follows: once the identified complier marginals are fixed, characterize all possible pushforwards of their couplings under the difference map $(y_1,y_0)\mapsto y_1-y_0$. We do not have a simpler closed-form characterization of $\mathcal I_{\Delta,c}$, but this makes precise that the remaining difficulty is entirely a coupling problem rather than a missing marginal-identification problem.
\end{remark}

\subsubsection{General Instrument Setting}\label{subsubsec:continuous_iv}
We now unify consider a general instrument setting under a single regularity condition on the instrument domain. Rather than assuming $\mathcal{Z}$ is finite, we only require it to be compact with finitely many connected components. This framework nests the finite discrete case (each component a singleton), the connected continuous case of the LIV literature \citep{heckman1999local}, and any intermixture of the two.

\begin{assumption}[Compactness and Piecewise Continuity]\label{asp:continuity_p}
    The domain of the instrument $\mathcal{Z}$ is a compact set with finitely many connected components. Furthermore, for almost every $x \in \mathcal{X}$, the conditional propensity score function $z \mapsto p(z,x) = \mathbb{P}(W=1 \mid Z=z, X=x)$ is continuous on each connected component.
\end{assumption}

Under \cref{asp:continuity_p}, the image of each connected component of $\mathcal{Z}$ under $p(\cdot)$ is a closed interval (possibly a singleton if $p$ is constant on the component or if the component is a single point). Merging overlaps and ordering from left to right, the range of the propensity score admits a \emph{maximal-interval decomposition}
\begin{equation}\label{eq:range_decomposition}
    p(\mathcal{Z}) \;=\; \bigcup_{k=1}^{K}\bigl[\underline{p}_k,\overline{p}_k\bigr], \qquad 0 \leqslant \underline{p}_1 \leqslant \overline{p}_1 < \underline{p}_2 \leqslant \overline{p}_2 < \cdots < \underline{p}_K \leqslant \overline{p}_K \leqslant 1,
\end{equation}
where we explicitly allow degenerate intervals $\underline{p}_k = \overline{p}_k$ corresponding to isolated propensity values, as in the finite discrete case. Adopting the boundary conventions $\overline{p}_0 := 0$ and $\underline{p}_{K+1} := 1$, this decomposition partitions the unit interval of the latent variable $U$ into three types of regions: the \emph{LIV intervals} $L_k := [\underline{p}_k,\overline{p}_k]$ for $k=1,\dots,K$, the \emph{OT gaps} $G_k := (\overline{p}_k, \underline{p}_{k+1})$ for $k=0,\dots,K-1$, and the \emph{trivial tail} $(\overline{p}_K,1]$.

The CCOT problem for the treated ($W=1$) component becomes
\begin{equation}\label{eq:ot_formulation_continuous}
    \begin{split}
        \min_{\pi_1}\quad &  \mathbb{E}_{\pi_1}[Y(1)\omega(U)]  \\
        \text{subject to}  \quad &  \int_0^{p} \pi_{1}(\mathrm{d}y \mid u) \mathrm{d}u = \mathbb{P}_{obs}(\mathrm{d}y, W=1 \mid p(Z)=p), \ \forall p \in p(\mathcal{Z}), \\
        & \pi_1(\mathrm{d}u) = \mathrm{d}u.
    \end{split}
\end{equation}
Following the same strategy as in \cref{subsubsec:discrete_iv}, we rewrite (\ref{eq:ot_formulation_continuous}) into an obviously separable form across the three region types. The objective splits additively as
\begin{align*}
\mathbb{E}_{\pi_1}[Y(1)\omega(U)]
&= \sum_{k=0}^{K-1}\int_{G_k}\!\int_{\mathcal{Y}} y\,\omega(u)\,\pi_1(\mathrm{d}y\mid u)\,\mathrm{d}u
+ \sum_{k=1}^{K}\int_{L_k}\!\int_{\mathcal{Y}} y\,\omega(u)\,\pi_1(\mathrm{d}y\mid u)\,\mathrm{d}u \\
&\quad\quad+ \int_{\overline{p}_K}^{1}\!\int_{\mathcal{Y}} y\,\omega(u)\,\pi_1(\mathrm{d}y\mid u)\,\mathrm{d}u.
\end{align*}
On each gap $G_k$, the same telescoping argument that produced (\ref{eq:y_identified_measure}) --- differencing the observational constraint in (\ref{eq:ot_formulation_continuous}) at $p=\overline{p}_k$ and $p=\underline{p}_{k+1}$ --- fixes the marginal of $Y(1)$ on $G_k$ to be
\begin{equation}\label{eq:gap_identified_measure}
    \mu_{1,k}(\mathrm{d}y) \;:=\; \frac{\mathbb{P}_{obs}(\mathrm{d}y, W=1 \mid p(Z)=\underline{p}_{k+1}) - \mathbb{P}_{obs}(\mathrm{d}y, W=1 \mid p(Z)=\overline{p}_k)}{\underline{p}_{k+1}-\overline{p}_k}, \qquad k=0,\dots,K-1,
\end{equation}
with the convention $\mathbb{P}_{obs}(\mathrm{d}y, W=1 \mid p(Z)=0) \equiv 0$ so that $k=0$ reduces to $\mu_{1,0}(\mathrm{d}y) = \mathbb{P}_{obs}(\mathrm{d}y, W=1 \mid p(Z)=\underline{p}_1)/\underline{p}_1$. On each non-degenerate LIV interval $L_k$, the constraint in (\ref{eq:ot_formulation_continuous}) holds for every $p$ in the continuum $[\underline{p}_k,\overline{p}_k]$; differentiating in $p$ pins $\pi_1(\mathrm{d}y\mid u)$ down pointwise via the LIV identity
\begin{equation}\label{eq:liv_identity}
\pi_1(\mathrm{d}y\mid u) \;=\; \frac{\partial}{\partial p}\bigg|_{p=u}\mathbb{P}_{obs}(\mathrm{d}y, W=1 \mid p(Z)=p), \qquad u\in L_k\ \text{a.e.,}
\end{equation}
whose absolute continuity in $p$ follows from the integral representation in (\ref{eq:ot_formulation_continuous}). The trivial tail $(\overline{p}_K,1]$ carries no constraint. Because cumulative sums of the gap marginals plus integration of the LIV densities recover the original constraint at every $p\in p(\mathcal{Z})$, this transformation is invertible and the two constraint systems are equivalent. Each rewritten constraint involves $\pi_1$ only on its own region and the objective is additively separable, so the global problem decomposes into three types of sub-problems, one for each region type.

\begin{enumerate}
    \item \textbf{The OT Gaps} $G_k = (\overline{p}_k, \underline{p}_{k+1})$ for $k=0,\dots,K-1$: the marginal of $Y(1)$ on $G_k$ is fixed to $\mu_{1,k}$ of (\ref{eq:gap_identified_measure}) and the $U$-marginal is $\text{Unif}(G_k)$, so the contribution of $G_k$ to the objective is a standard 1D optimal transport problem between $\mu_{1,k}$ and $\text{Unif}(G_k)$, which is solved by \cref{thm:1d_ot}.
    \item \textbf{The LIV Intervals} $L_k = [\underline{p}_k, \overline{p}_k]$ for $k=1,\dots,K$: for non-degenerate $L_k$, the LIV identity in (\ref{eq:liv_identity}) point-identifies the conditional mean
    \[\mathbb{E}[Y(1) \mid U=u] \;=\; \frac{\diff}{\diff u}\mathbb{E}_{obs}[Y W \mid p(Z)=u], \qquad u \in L_k,\]
    so the contribution of $L_k$ to the objective is uniquely determined and no optimization is required. Degenerate $L_k$ ($\underline{p}_k = \overline{p}_k$) contribute a Lebesgue-null integral and drop out.
    \item \textbf{The Trivial Tail} $u \in (\overline{p}_K, 1]$: for individuals with extreme resistance to treatment, $Y(1)$ is never observed and the data imposes no constraint. Utilizing \cref{asp:bounded_y}, the minimization problem degenerates to assigning all mass to $y_{\min}$ where $\omega(u) \geqslant 0$ and to $y_{\max}$ where $\omega(u) < 0$.
\end{enumerate}

Combining the solutions across these regions yields the closed-form sharp bounds for the general instrument setting, which are illustrated in \cref{fig:continuous_decomposition}.

\begin{theorem}[General Closed-Form Sharp Bounds]\label{thm:continuous_lower_bound}
    Suppose \cref{asp:structure}, \cref{asp:bounded_y}, and \cref{asp:continuity_p} hold, and $\mathcal{X} = \emptyset$. Let $p(\mathcal{Z}) = \bigcup_{k=1}^{K}[\underline{p}_k, \overline{p}_k]$ be the maximal-interval decomposition in (\ref{eq:range_decomposition}). For each $k = 0,\dots,K-1$, let $Q_{Y,k}(t)$ be the quantile function of $\mu_{1,k}$ defined in (\ref{eq:gap_identified_measure}), and $Q_{\omega,k}(t)$ the quantile function of $\omega(U)$ for $U \sim \text{Unif}(G_k)$. The sharp lower bound for the $W=1$ component of the target parameter is given by
    \begin{equation}\label{eq:continuous_lower_bound}
        \begin{split}
            \underline{\theta}_{\omega,1}\ =\
            & \underbrace{\sum_{k=0}^{K-1}(\underline{p}_{k+1}-\overline{p}_k)\int_0^1 Q_{Y,k}(t)\,Q_{\omega,k}(1-t)\,\mathrm{d}t}_{\text{OT Bound on the Gaps}} + \underbrace{\sum_{k=1}^{K}\int_{\underline{p}_k}^{\overline{p}_k}\left(\frac{\diff}{\diff u}\mathbb{E}_{obs}[YW \mid p(Z)=u]\right)\omega(u)\,\mathrm{d}u}_{\text{Point-Identified LIV on the Intervals}} \\
            & \quad\quad + \underbrace{\int_{\overline{p}_K}^{1}\Big(y_{\min}\max\{0,\omega(u)\} + y_{\max}\min\{0,\omega(u)\}\Big)\mathrm{d}u}_{\text{Trivial Lower Bound on the Tail}}.
        \end{split}
    \end{equation}
    Similarly, the sharp upper bound is given by replaceing $Q_{\omega,k}(1-t)$ with $Q_{\omega,k}(t)$ in the OT term and swapping $y_{\min}$ and $y_{\max}$ in the trivial tail term.
    % \begin{equation*}
    %     \begin{split}
    %         \overline{\theta}_{\omega,1}\ =\
    %         & \sum_{k=0}^{K-1}(\underline{p}_{k+1}-\overline{p}_k)\int_0^1 Q_{Y,k}(t)\,Q_{\omega,k}(t)\,\mathrm{d}t \ + \ \sum_{k=1}^{K}\int_{\underline{p}_k}^{\overline{p}_k}\left(\frac{\diff}{\diff u}\mathbb{E}_{obs}[YW \mid p(Z)=u]\right)\omega(u)\,\mathrm{d}u \\
    %         & \quad\quad + \int_{\overline{p}_K}^{1}\Big(y_{\max}\max\{0,\omega(u)\} + y_{\min}\min\{0,\omega(u)\}\Big)\mathrm{d}u.
    %     \end{split}
    % \end{equation*}
\end{theorem}

The proof combines the subtraction-of-constraints decomposition from \cref{thm:discrete_lower_bound} (applied region-by-region to identify $\mu_{1,k}$ on each gap) with the LIV-derivative argument on each non-degenerate interval $L_k$, followed by the 1D OT coupling of \cref{thm:1d_ot}. We defer the formal proof to \cref{thm:covariate_continuous_lower_bound}, which restates the result in the unified form in the covariate subsection below. \cref{thm:continuous_lower_bound} subsumes both classical settings as limit cases. When every LIV interval is degenerate ($\underline{p}_k = \overline{p}_k$ for all $k$), the LIV sum vanishes identically, each $\mu_{1,k}$ coincides with the discrete complier distribution in (\ref{eq:y_identified_measure}), and (\ref{eq:continuous_lower_bound}) reduces exactly to \cref{thm:discrete_lower_bound}. Conversely, when $K=1$ with $\underline{p}_1 < \overline{p}_1$, the gap sum contains a single OT term on $(0,\underline{p}_1)$, the LIV sum contributes a single integral on $[\underline{p}_1,\overline{p}_1]$, and we recover the connected-continuous case; in particular, we match the identification results of \citet{heckman1999local} when additionally $\underline{p}_1 = 0$ and $\overline{p}_1 = 1$, so the propensity score has full support with no gaps. The length of the bound depends on the gap widths, the identified gap distributions $\mu_{1,k}$, and the tail beyond $\overline{p}_K$:
\begin{align*}
    \overline{\theta}_{\omega,1} - \underline{\theta}_{\omega,1} = \sum_{k=0}^{K-1}(\underline{p}_{k+1}-\overline{p}_k)\int_0^1\Big(Q_{Y,k}(t) - Q_{Y,k}(1-t)\Big)Q_{\omega,k}(t)\,\mathrm{d}t \ + \ (y_{\max}-y_{\min})\int_{\overline{p}_K}^{1}|\omega(u)|\,\mathrm{d}u.
\end{align*}

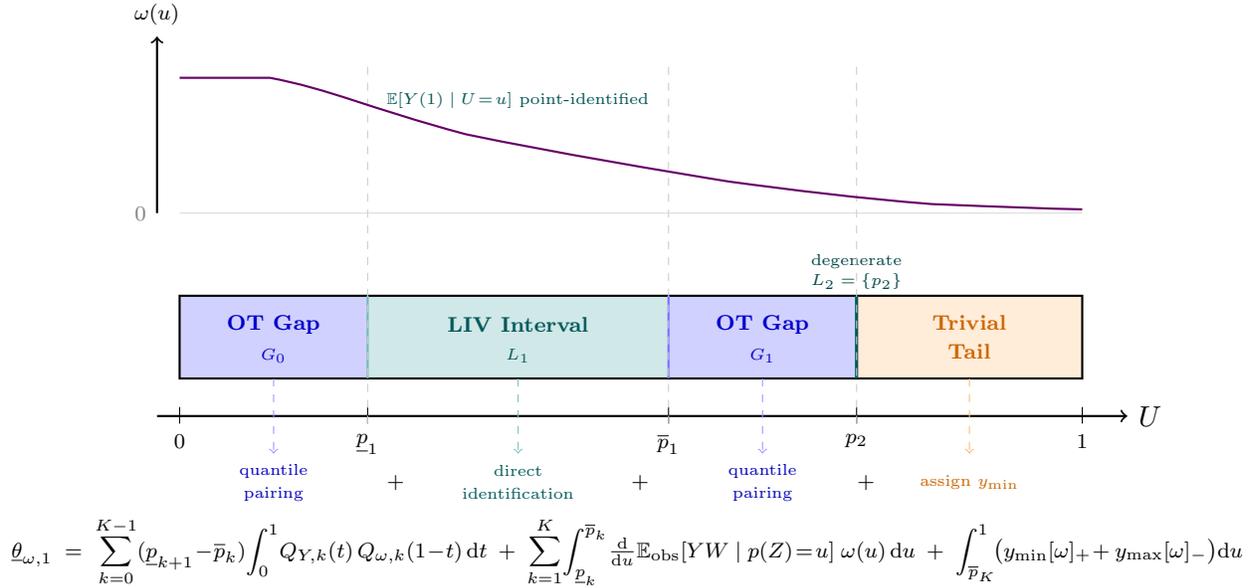
\begin{figure}[!htbp]
\centering
\begin{tikzpicture}[scale=1.0, font=\small]
  %% --- Layout (general instrument setting: $K=2$, order OT / LIV / OT / trivial) ---
  \def\W{12.0}    % full width (U = 1)
  \def\SH{1.10}   % main strip height
  \def\pOneL{2.5} % \underline{p}_1
  \def\pOneH{6.5} % \overline{p}_1
  \def\pTwo{9.0}  % p_2 — degenerate L_2 = \{p_2\}, marks the start of the trivial tail
  \def\qL{1.2}    % where omega begins descending
  \def\omY{2.20}  % omega baseline
  \def\omMax{1.80}% omega plateau height

  %% --- Main identification strip ---
  % G_0 = (0, \underline{p}_1): OT gap
  \fill[blue!18]    (0,0)        rectangle (\pOneL,\SH);
  % L_1 = [\underline{p}_1, \overline{p}_1]: non-degenerate LIV interval
  \fill[teal!18]    (\pOneL,0)   rectangle (\pOneH,\SH);
  % G_1 = (\overline{p}_1, p_2): OT gap
  \fill[blue!18]    (\pOneH,0)   rectangle (\pTwo,\SH);
  % Trivial tail (p_2, 1]
  \fill[orange!15]  (\pTwo,0)    rectangle (\W,\SH);
  \draw[thick]      (0,0)        rectangle (\W,\SH);
  % Vertical boundaries between regions
  \draw[thick, teal!55]   (\pOneL,0) -- (\pOneL,\SH);
  \draw[thick, blue!55]   (\pOneH,0) -- (\pOneH,\SH);
  % L_2 = \{p_2\}: degenerate LIV rendered as a thick teal line at the boundary to the trivial tail
  \draw[very thick, teal!70!black] (\pTwo,0) -- (\pTwo,\SH);

  % Region labels inside strip
  \node[blue!80!black,   font=\scriptsize\bfseries]  at (\pOneL/2,            0.72) {OT Gap};
  \node[blue!80!black,   font=\tiny]                 at (\pOneL/2,            0.30) {$G_0$};
  \node[teal!70!black,   font=\scriptsize\bfseries]  at ({(\pOneL+\pOneH)/2}, 0.72) {LIV Interval};
  \node[teal!70!black,   font=\tiny]                 at ({(\pOneL+\pOneH)/2}, 0.30) {$L_1$};
  \node[blue!80!black,   font=\scriptsize\bfseries]  at ({(\pOneH+\pTwo)/2},  0.72) {OT Gap};
  \node[blue!80!black,   font=\tiny]                 at ({(\pOneH+\pTwo)/2},  0.30) {$G_1$};
  \node[orange!80!black, font=\scriptsize\bfseries, align=center]
      at ({(\pTwo+\W)/2}, \SH/2) {Trivial\\[-1pt]Tail};

  % Degenerate LIV singleton label (above strip)
  \node[teal!55!black, font=\tiny, align=center]
      at (\pTwo, \SH+0.32) {degenerate\\[-1pt]$L_2=\{p_2\}$};

  %% --- U axis ---
  \draw[->,thick] (-0.3,-0.50) -- (\W+0.6,-0.50) node[right]{$U$};
  \draw (0,      -0.38) -- (0,      -0.60); \node[below,font=\scriptsize] at (0,      -0.60) {$0$};
  \draw (\pOneL, -0.38) -- (\pOneL, -0.60); \node[below,font=\scriptsize] at (\pOneL, -0.60) {$\underline{p}_1$};
  \draw (\pOneH, -0.38) -- (\pOneH, -0.60); \node[below,font=\scriptsize] at (\pOneH, -0.60) {$\overline{p}_1$};
  \draw (\pTwo,  -0.38) -- (\pTwo,  -0.60); \node[below,font=\scriptsize] at (\pTwo,  -0.60) {$p_2$};
  \draw (\W,     -0.38) -- (\W,     -0.60); \node[below,font=\scriptsize] at (\W,     -0.60) {$1$};

  %% --- Dashed vertical guides through all panels ---
  \draw[gray!40,dashed,thin] (\pOneL,-0.60) -- (\pOneL,{\omY+\omMax+0.25});
  \draw[gray!40,dashed,thin] (\pOneH,-0.60) -- (\pOneH,{\omY+\omMax+0.25});
  \draw[gray!40,dashed,thin] (\pTwo, -0.60) -- (\pTwo, {\omY+\omMax+0.25});

  %% --- Omega smooth-curve panel ---
  \draw[->,thick] (-0.3,\omY) -- (-0.3,{\omY+\omMax+0.55})
      node[above,font=\scriptsize]{$\omega(u)$};
  \draw[gray!30,thin] (0,\omY) -- (\W,\omY);      % zero baseline
  \node[left,font=\scriptsize,gray] at (-0.3,\omY) {$0$};

  % Smooth non-increasing omega curve
  \draw[thick, violet!80!black]
      (0, {\omY+\omMax}) -- (\qL, {\omY+\omMax})
      .. controls (2.0, {\omY+1.65}) and (\pOneL, {\omY+1.40}) ..
         (3.8, {\omY+1.05})
      .. controls (5.0, {\omY+0.80}) and (\pOneH, {\omY+0.55}) ..
         (7.3, {\omY+0.42})
      .. controls (8.2, {\omY+0.30}) and (\pTwo,  {\omY+0.20}) ..
         (10.0, {\omY+0.12})
      .. controls (11.0, {\omY+0.08}) and (11.5, {\omY+0.06}) ..
         (\W, {\omY+0.05});

  % LIV annotation: point-identified via derivative
  \node[font=\tiny,teal!65!black] at ({(\pOneL+\pOneH)/2},{\omY+1.50})
      {$\mathbb{E}[Y(1)\mid U\!=\!u]$ point-identified};

  %% --- Connecting arrows (strip to bound labels) ---
  \draw[->,dashed,blue!40]    (\pOneL/2,              0) -- (\pOneL/2,              -1.00);
  \draw[->,dashed,teal!50]    ({(\pOneL+\pOneH)/2},   0) -- ({(\pOneL+\pOneH)/2},   -1.00);
  \draw[->,dashed,blue!40]    ({(\pOneH+\pTwo)/2},    0) -- ({(\pOneH+\pTwo)/2},    -1.00);
  \draw[->,dashed,orange!55]  ({(\pTwo+\W)/2},        0) -- ({(\pTwo+\W)/2},        -1.00);

  %% --- Brief per-region bound labels ---
  \node[blue!75!black, font=\tiny, align=center]
      at (\pOneL/2, -1.38) {quantile\\pairing};
  \node[teal!70!black, font=\tiny, align=center]
      at ({(\pOneL+\pOneH)/2}, -1.38) {direct\\identification};
  \node[blue!75!black, font=\tiny, align=center]
      at ({(\pOneH+\pTwo)/2}, -1.38) {quantile\\pairing};
  \node[orange!80!black, font=\tiny, align=center]
      at ({(\pTwo+\W)/2}, -1.38) {assign $y_{\min}$};

  % plus signs between labels
  \node[font=\scriptsize] at ({(\pOneL/2 + (\pOneL+\pOneH)/2)/2},         -1.38) {$+$};
  \node[font=\scriptsize] at ({((\pOneL+\pOneH)/2 + (\pOneH+\pTwo)/2)/2}, -1.38) {$+$};
  \node[font=\scriptsize] at ({((\pOneH+\pTwo)/2 + (\pTwo+\W)/2)/2},      -1.38) {$+$};

  %% --- Global sharp lower bound formula ---
  \node[align=center, font=\scriptsize] at ({\W/2}, -2.30) {
    $
    \underline{\theta}_{\omega,1}
    \;=\; \displaystyle\sum_{k=0}^{K-1}(\underline{p}_{k+1}\!-\!\overline{p}_k)\!\int_0^1\! Q_{Y,k}(t)\,Q_{\omega,k}(1\!-\!t)\,\mathrm{d}t
    \;+\; \displaystyle\sum_{k=1}^{K}\!\int_{\underline{p}_k}^{\overline{p}_k}\!
      \tfrac{\mathrm{d}}{\mathrm{d}u}\mathbb{E}_{\mathrm{obs}}[YW\mid p(Z)\!=\!u]
      \;\omega(u)\,\mathrm{d}u
    \;+\; \displaystyle\int_{\overline{p}_K}^{1}\!\bigl(y_{\min}[\omega]_+\!+y_{\max}[\omega]_-\bigr)\mathrm{d}u$
  };

\end{tikzpicture}
\caption{Decomposition of the unit interval $[0,1]$ for the sharp lower bound in the general instrument setting (\cref{thm:continuous_lower_bound}), illustrated with $K=2$ in the order OT / LIV / OT / trivial: the non-degenerate LIV interval $L_1 = [\underline{p}_1, \overline{p}_1]$ (teal strip) lies between two OT gaps, and the \emph{degenerate} LIV singleton $L_2 = \{p_2\}$ (teal line) sits at the boundary of the trivial tail. \textbf{\color{blue!65}OT gaps} $G_0 = (0, \underline{p}_1)$ and $G_1 = (\overline{p}_1, p_2)$: only the aggregate distributions $\mu_{1,0}, \mu_{1,1}$ are identified, so the tight bound requires countermonotone quantile coupling. \textbf{\color{teal!65}LIV interval} $L_1$: the propensity score sweeps a full sub-interval, so the standard LIV derivative point-identifies $\mathbb{E}[Y(1)\mid U=u]$ exactly; the singleton $L_2$ contributes a Lebesgue-null integral and drops out. \textbf{\color{orange!65}Trivial tail} $(p_2, 1]$: $Y(1)$ is never observed; the lower bound assigns $y_{\min}$ where $\omega(u)\geqslant 0$ and $y_{\max}$ where $\omega(u)<0$. Collapsing every LIV interval to a singleton recovers the discrete case (\cref{fig:decomposition}); taking $K=1$ with $\underline{p}_1<\overline{p}_1$ and no trailing singleton recovers the connected-continuous case.}
\label{fig:continuous_decomposition}
\end{figure}

\subsection{Closed-Form Bounds with Covariates}\label{subsec:cf_cov}

Generalizing the previous results to accommodate covariates $X$ is mathematically straightforward. Because our structural assumption \cref{asp:structure} imposes conditional independence $Z \indep (Y(w), U) \mid X$, the global CCOT problem naturally disintegrates into a family of conditional one-dimensional optimal transport problems, indexed by $x \in \mathcal{X}$. We can therefore solve the optimization conditional on $X=x$ and aggregate the resulting conditional bounds over the marginal distribution of $X$ via the law of iterated expectations.

Under \cref{asp:continuity_p}, for $\mathbb{P}_{obs}$-a.e.\ $x \in \mathcal{X}$ the conditional range of the propensity score admits a maximal-interval decomposition
\begin{equation}\label{eq:range_decomposition_cov}
    p(\mathcal{Z}, x) \;=\; \bigcup_{k=1}^{K_x}[\underline{p}_k(x), \overline{p}_k(x)], \qquad 0 \leqslant \underline{p}_1(x) \leqslant \overline{p}_1(x) < \cdots < \underline{p}_{K_x}(x) \leqslant \overline{p}_{K_x}(x) \leqslant 1,
\end{equation}
with degenerate intervals $\underline{p}_k(x) = \overline{p}_k(x)$ corresponding to isolated conditional propensity values and boundary conventions $\overline{p}_0(x) := 0$ and $\underline{p}_{K_x+1}(x) := 1$. For each $x$, this partitions $[0,1]$ into the conditional LIV intervals $L_k(x) := [\underline{p}_k(x), \overline{p}_k(x)]$, the OT gaps $G_k(x) := (\overline{p}_k(x), \underline{p}_{k+1}(x))$, and the trivial tail $(\overline{p}_{K_x}(x), 1]$. The conditional gap-identified measure of $Y(1)$ on $G_k(x)$ is
\begin{equation}\label{eq:gap_identified_measure_cov}
    \mu_{1,k\mid x}(\mathrm{d}y) \;:=\; \frac{\mathbb{P}_{obs}(\mathrm{d}y, W=1 \mid p(Z,X)=\underline{p}_{k+1}(x), X=x) - \mathbb{P}_{obs}(\mathrm{d}y, W=1 \mid p(Z,X)=\overline{p}_k(x), X=x)}{\underline{p}_{k+1}(x)-\overline{p}_k(x)},
\end{equation}
for $k=0,\dots,K_x-1$, with the convention $\mathbb{P}_{obs}(\mathrm{d}y, W=1 \mid p(Z,X)=0, X=x) \equiv 0$.

\begin{theorem}[Conditional General Bound]\label{thm:covariate_continuous_lower_bound}
    Suppose \cref{asp:structure}, \cref{asp:bounded_y}, and \cref{asp:continuity_p} hold. For each $k = 0,\dots,K_x-1$, let $Q_{Y, k \mid x}(t)$ be the quantile function of $\mu_{1,k\mid x}$ defined in (\ref{eq:gap_identified_measure_cov}), and $Q_{\omega, k \mid x}(t)$ the quantile function of $\omega(x, U)$ for $U \sim \text{Unif}(G_k(x))$. The sharp lower bound for the $W=1$ component of the target parameter is
    \begin{equation} \label{eq:covariate_continuous_lower_bound}
        \begin{split}
        \underline{\theta}_{\omega,1}  = \mathbb{E}_{X} \Bigg[ & \sum_{k=0}^{K_X-1}(\underline{p}_{k+1}(X)-\overline{p}_k(X))\int_0^1 Q_{Y,k\mid X}(t)\,Q_{\omega,k\mid X}(1-t)\,\mathrm{d}t \\
        & + \sum_{k=1}^{K_X}\int_{\underline{p}_k(X)}^{\overline{p}_k(X)} \left( \frac{\diff }{\diff u} \mathbb{E}_{obs}[Y W \mid p(Z,X)=u, X] \right) \omega(X, u) \, \mathrm{d}u \\
        & + \int_{\overline{p}_{K_X}(X)}^1 \Big( y_{\min} \max\{0, \omega(X, u)\} + y_{\max} \min\{0, \omega(X, u)\} \Big) \mathrm{d}u \Bigg].
        \end{split}
    \end{equation}
    Similarly, the sharp upper bound is given by replacing $Q_{\omega,k\mid X}(1-t)$ with $Q_{\omega,k\mid X}(t)$ in the OT term and swapping $y_{\min}$ and $y_{\max}$ in the trivial tail term.
    % \begin{equation*}
    %     \begin{split}
    %     \overline{\theta}_{\omega,1}  = \mathbb{E}_{X} \Bigg[ & \sum_{k=0}^{K_X-1}(\underline{p}_{k+1}(X)-\overline{p}_k(X))\int_0^1 Q_{Y,k\mid X}(t)\,Q_{\omega,k\mid X}(t)\,\mathrm{d}t \\
    %     & + \sum_{k=1}^{K_X}\int_{\underline{p}_k(X)}^{\overline{p}_k(X)} \left( \frac{\diff }{\diff u} \mathbb{E}_{obs}[Y W \mid p(Z,X)=u, X] \right) \omega(X, u) \, \mathrm{d}u  \\
    %     & + \int_{\overline{p}_{K_X}(X)}^1 \Big( y_{\max} \max\{0, \omega(X, u)\} + y_{\min} \min\{0, \omega(X, u)\} \Big) \mathrm{d}u \Bigg].
    %     \end{split}
    % \end{equation*}
\end{theorem}
The proof of \cref{thm:covariate_continuous_lower_bound} is given in \cref{sec:cov_proof}. As in the no-covariate case, \cref{thm:covariate_continuous_lower_bound} subsumes both the finite-discrete and the connected-continuous settings: when every conditional LIV interval is degenerate ($\underline{p}_k(X) = \overline{p}_k(X)$ a.s.), the LIV sum vanishes and the expression reduces to the discrete covariate bound; when $K_X = 1$ with $\underline{p}_1(X) < \overline{p}_1(X)$, the gap sum collapses to a single term on $(0, \underline{p}_1(X))$ and we recover the standard LIV covariate bound.

%% file: sections/extension.tex
\section{Extension to General Treatment} \label{sec:extension}

In this section, we generalize our partial identification framework to accommodate continuous treatments. We allow the treatment domain $\mathcal{W}$ to be a compact subset of $\mathbb{R}$. We impose the following generalized structural assumptions.

\begin{assumption}[Generalized Structural Assumptions] \label{asp:structure_generalized}
The potential outcomes and the treatment selection mechanism satisfy the following conditions:
\begin{enumerate}
    \item \textbf{Consistency:} $Y = Y(W)$.
    \item \textbf{Conditional Instrumental Exogeneity:} $Z \indep (U, \{Y(w)\}_{w \in \mathcal{W}}) \mid X$.
    \item \textbf{Selection Mechanism:} The treatment $W$ is selected via the structural equation:
    \begin{equation} \label{eq:selection_generalized}
        W = F_{W \mid Z,X}^{-1}(U \mid Z,X),
    \end{equation}
    where $U \mid X \sim \text{Unif}(0,1)$, and $F_{W \mid Z,X}^{-1}(\cdot \mid z,x)$ is the conditional quantile function of $W$ given $Z=z$ and $X=x$.
    \item \textbf{Common Instrument Support:} The conditional support of $Z$ given $X$ does not depend on $X$: $\mathrm{Supp}(Z \mid X=x) = \mathcal{Z}$ for all $x \in \mathcal{X}$.
\end{enumerate}
\end{assumption}
As in the binary case, the common instrument support condition ensures that the conditional distribution $F_{W \mid Z=z, X=x}$ is well-defined for every $z \in \mathcal{Z}$ and every $x \in \mathcal{X}$, which is necessary for the structural equation (\ref{eq:selection_generalized}) to be identified uniformly across covariates. If $W$ is a binary treatment, \cref{asp:structure_generalized} reduces to the threshold-crossing model utilized in our previous assumptions after the harmless reparameterization $U' = 1-U$.\footnote{Under the generalized ordered-choice convention, $W = F^{-1}_{W\mid Z,X}(U\mid Z,X)$ is increasing in $U$, whereas our earlier binary threshold-crossing convention writes treatment as $W=\indicator(U \le p(Z,X))$. Setting $U' = 1-U$ aligns the two models.} If $W$ is multi-valued, this model is equivalent to the discrete ordered choice model \citep{heckman2007econometric}. In the continuous treatment regime, this assumption provides the structural foundation for the control function approach in nonseparable models \citep{imbens2009identification, blundell2003endogeneity, chesher2003identification}. We discuss this connection in detail in \cref{subsec:triangle}.

In this model, $U$ can be interpreted as a latent rank that orders individuals' treatment intensity within each $(Z,X)$ cell. Under the ordered-choice convention in (\ref{eq:selection_generalized}), larger values of $U$ correspond to larger realizations of $W$.

Following the binary treatment case, we consider the following general target parameter. Let $\pi_{w,x}$ be the structural joint distribution of $(U, Y(w)) \mid X=x$, which can be seen as a probability kernel. For any identifiable weight function $\omega(w,x,u)$, we define:
\begin{equation}\label{eq:target_generalized}
    \theta_\omega = \int_{\mathcal{W}} \mathbb{E}_X \left[ \int_0^1 \mathbb{E}_{\pi_{w,X}}[Y(w) \mid U=u] \,\omega(w,X,u) \,\mathrm{d}u \right] \lambda(\mathrm{d}w),
\end{equation}
where $\lambda(\cdot)$ denotes the Lebesgue measure for continuous treatments and the counting measure for discrete treatments. This directly generalizes $\theta_\omega$ from (\ref{eq:target_mte}). To see this, consider the binary treatment case $\mathcal{W} = \{0,1\}$, where $\int_{\mathcal{W}} \cdot \,\mathrm{d}w$ reduces to the sum over $\{0,1\}$. Setting $\omega(1,x,u) = \omega(x,u)$ and $\omega(0,x,u) = -\omega(x,u)$, the expression (\ref{eq:target_generalized}) becomes:
\begin{equation*}
    \theta_\omega = \mathbb{E}_X \int_0^1 \Big(\mathbb{E}_{\pi_{1,X}}[Y(1) \mid U=u] - \mathbb{E}_{\pi_{0,X}}[Y(0) \mid U=u]\Big) \omega(X,u) \,\mathrm{d}u = \mathbb{E}_X \int_0^1 \text{MTE}(X,u)\,\omega(X,u)\,\mathrm{d}u,
\end{equation*}
recovering (\ref{eq:target_mte}). 

As an example, consider a counterfactual policy that alters two components while holding the covariate distribution fixed: (i) the treatment assignment mechanism changes from $F_{W \mid Z,X}$ to $\tilde{F}_{W \mid Z,X}$, and (ii) the instrument distribution shifts from $F_{Z \mid X}$ to $\tilde{F}_{\tilde{Z} \mid X}$. Under this new policy, the counterfactual treatment is $\tilde{W} = \tilde{F}_{W \mid Z,X}^{-1}(U \mid \tilde{Z},X)$, and the realized outcome is $\tilde{Y} = Y(\tilde{W})$. The policy effect $\theta = \mathbb{E}[\tilde{Y} - Y]$ can be expressed as $\theta_\omega$ with the specific weight:
\begin{equation}\label{eq:policy_weight}
    \omega(w,x,u)\,\mathrm{d}u = \mathrm{d}\tilde{F}_{U,\tilde{W} \mid X}(u,w \mid x) - \mathrm{d}F_{U,W \mid X}(u,w \mid x),
\end{equation}
where the right-hand side denotes the $w$-section of the signed measure on $\mathcal{W} \times [0,1]$. The policy weight (\ref{eq:policy_weight}) reduces to (\ref{eq:omega_form}) in the binary case.

The observed data constrains the structural kernel $\pi_{w,x}$. Note that the conditional distribution $\mathbb{P}_{\text{obs}}(\mathrm{d}y \mid W=w, Z=z, X=x)$ is only well-defined when $w$ belongs to the conditional support of $W$ given $(Z,X)$. We therefore define the effective instrument set:
\begin{equation*}
    \mathcal{Z}(x,w) \coloneqq \{z \in \mathcal{Z} : w \in \text{supp}(F_{W \mid Z=z, X=x})\},
\end{equation*}
and impose the observational constraints only for $z \in \mathcal{Z}(x,w)$. Specifically, for any $w \in \mathcal{W}$ and $x \in \mathcal{X}$, integrating the conditional distribution of the potential outcome over the latent variable $U$ must perfectly reconstruct the observed conditional outcome measure:
\begin{equation*}
    \int_0^1 \pi_{w,x}(\mathrm{d}y \mid U=u) \mathrm{d}F_{U \mid Z, X, W}(u \mid z, x, w) = \mathbb{P}_{\text{obs}}(\mathrm{d}y \mid W=w, Z=z, X=x), \quad \forall z \in \mathcal{Z}(x,w).
\end{equation*}
Crucially, because both the objective functional and the observational constraints are additively separable across the support of $X$ and $W$, the global optimization over all joint distributions decomposes into a family of independent CCOT sub-problems. Formally, this is equivalent to solving the following CCOT problem for $\lambda$-a.e.\ $w \in \mathcal{W}$ and $\mathbb{P}_{\mathrm{obs}}$-a.e.\ $x \in \mathcal{X}$:
\begin{equation}\label{eq:ot_generalized_separate}
    \begin{split}
        \max_{\pi_{w,x}} / \min_{\pi_{w,x}} \quad & \int_0^1 \mathbb{E}_{\pi_{w,x}}[Y(w) \mid U=u] \,\omega(w,x,u)\,\mathrm{d}u \\
\text{s.t.} \quad & \int_0^1 \pi_{w,x}(\mathrm{d}y \mid U=u) \mathrm{d}F_{U \mid Z, X, W}(u \mid z, x, w) \\&
\quad\quad\quad\quad\quad\quad= \mathbb{P}_{\text{obs}}(\mathrm{d}y \mid W=w, Z=z, X=x), \quad \forall z \in \mathcal{Z}(x,w), \\
& \pi_{w,x}(\mathrm{d}u) = \mathrm{d}u,\\
&\pi_{w,x} \text{ is supported on } \mathcal{Y} \times [0,1].
    \end{split}
\end{equation}
Let $\overline{\theta}(x,w)$ and $\underline{\theta}(x,w)$ denote the optimal values of the maximization and minimization in (\ref{eq:ot_generalized_separate}), respectively. The tight bounds for the global parameter $\theta_\omega$ are then obtained by aggregating over all $(x,w)$ strata:
\begin{equation}\label{eq:global_bounds_aggregation}
    \underline{\theta}_\omega = \int_{\mathcal{W}} \mathbb{E}_X\!\left[\underline{\theta}(X,w)\right] \lambda(\mathrm{d}w), \qquad \overline{\theta}_\omega = \int_{\mathcal{W}} \mathbb{E}_X\!\left[\overline{\theta}(X,w)\right] \lambda(\mathrm{d}w).
\end{equation}
To formally guarantee that (\ref{eq:ot_generalized_separate}) yields sharp bounds for the continuous treatment setting, we must establish that any family of measures satisfying the generalized observational and marginal constraints corresponds to a valid data-generating process. Let $\Gamma_{\text{gen}}(\mathbb{P}_{\text{obs}})$ denote the set of all measure families $\{\pi_{w,x}\}_{w \in \mathcal{W}, x \in \mathcal{X}}$ that satisfy the constraints in the CCOT problem above. The following proposition verifies this observational equivalence, extending the sharpness guarantee from the binary framework, \cref{prop:sharpness}.

\begin{proposition}[Generalized Observational Equivalence and Sharpness] \label{prop:sharpness_generalized}
    Suppose the observed distribution $\mathbb{P}_{\text{obs}}$ is generated by a true structural model satisfying \cref{asp:structure_generalized}. Then, for any candidate family of probability kernels $\{\tilde{\pi}_{w,x}\}_{w \in \mathcal{W}, x \in \mathcal{X}} \in \Gamma_{\text{gen}}(\mathbb{P}_{\text{obs}})$, there exists a probability space supporting a jointly measurable process $\{Y(w)\}_{w \in \mathcal{W}}$ together with random variables $(U, X, Z, W)$ such that:
    \begin{enumerate}
        \item[(i)] For $\lambda$-a.e.\ $w \in \mathcal{W}$ and $\mathbb{P}_{\text{obs}}$-a.e.\ $x \in \mathcal{X}$, the conditional distribution of $(Y(w), U) \mid X=x$ is exactly $\tilde{\pi}_{w,x}$. Moreover, conditional on $(U, X)$, the process $\{Y(w)\}_{w \in \mathcal{W}}$ is independent of $(Z, W)$ and has mutually independent coordinates.
        \item[(ii)] \cref{asp:structure_generalized} holds for $(\{Y(w)\}_{w \in \mathcal{W}}, U, X, Z, W)$.
        \item[(iii)] The induced distribution of the observable variables $(Z,X,W,Y)$ exactly matches the observed data distribution $\mathbb{P}_{\text{obs}}$.
    \end{enumerate}
\end{proposition}

This formal CCOT formulation highlights that variation in $F_{U \mid Z=z, X=x, W=w}$ across different instruments $Z$ imposes multiple marginal constraints on the structural distribution $\pi_{w,x}$. In general, the solution depends on the union of these supports, 
$$\mathcal{U}_{\text{id}}(x,w) \coloneqq \bigcup_{z \in \mathcal{Z}(x,w)} \text{supp}(F_{U \mid Z=z, X=x, W=w}).$$ 
Note that each $\text{supp}(F_{U \mid Z=z, X=x, W=w})$ is either a single point or a closed interval in $[0,1]$. Throughout this section, we assume that $\mathcal{U}_{\text{id}}(x,w)$ is measurable for each $(x,w)$. This measurability condition holds, for instance, if $\mathcal{Z}(x,w)$ is compact and the map $z \mapsto F_{W \mid Z,X}(w \mid z,x)$ is continuous, in which case $\mathcal{U}_{\text{id}}(x,w)$ is a compact subset of $[0,1]$.

In the following subsections, we demonstrate how to obtain closed-form solutions under specific structural settings. The general case is much more complicated and can be seen as a mix of these settings, which we will leave for future work. 

\subsection{Strictly Monotonic Treatment Selection} \label{subsec:triangle}

\citet{imbens2009identification} considers the following assumption for the IV model.

\begin{assumption}[Strict Monotonicity] \label{asp:strict_monotonicity}
The structural treatment function $W = h(Z,X,U)$ is strictly increasing in the latent variable $U$ almost surely.
\end{assumption}

Under \cref{asp:strict_monotonicity}, it is observationally equivalent to the conditional quantile representation in (\ref{eq:selection_generalized}). Crucially, the unobserved confounder $U$ can be perfectly inverted and recovered from the observables as the conditional rank of the treatment, $U = F_{W \mid Z,X}(W \mid Z,X)$. This recovered latent rank serves directly as a control variable: conditioning on $U$ alongside the covariates $X$ effectively absorbs the endogeneity, allowing us to isolate the exogenous variation in $W$. This formal equivalence connects our generalized partial identification framework to classical control function approach \citep{blundell2003endogeneity, chesher2003identification}.

Because $U$ is perfectly recoverable, the conditional distribution $F_{U \mid Z=z, X=x, W=w}$ degenerates to a point mass at $u = F_{W \mid Z,X}(w \mid z, x)$. Consequently, the integral constraint in (\ref{eq:ot_generalized_separate}) simplifies to:
\begin{equation*}
    \pi_{w,x}(\mathrm{d}y \mid U = u) = \mathbb{P}_{\text{obs}}(\mathrm{d}y \mid W=w, Z=z_u, X=x), \quad \text{for } u \in \mathcal{U}_{\text{id}}(x,w),
\end{equation*}
where $z_u \in \mathcal{Z}(x,w)$ satisfies $F_{W \mid Z,X}(w \mid z_u, x) = u$. When multiple instrument values map to the same $u$, the observational constraints ensure that $\mathbb{P}_{\text{obs}}(\mathrm{d}y \mid W=w, Z=z_u, X=x)$ is identical for all such $z_u$, so the choice is immaterial. That is, the structural kernel $\pi_{w,x}$ is point-identified on the identifiable region $\mathcal{U}_{\text{id}}(x,w) = \bigcup_{z \in \mathcal{Z}(x,w)} \{F_{W\mid Z,X}(w \mid z,x)\}$. Therefore, if $u \in \mathcal{U}_{\text{id}}(x,w)$, we have:
\begin{equation}\label{eq:id_point_mass}
    \mathbb{E}_{\pi_{w,x}}[Y(w) \mid U=u] = \mathbb{E}_{\text{obs}}[Y \mid W=w, Z=z_u, X=x],
\end{equation}
where $z_u \in \mathcal{Z}$ is the specific baseline instrument realization that satisfies $w = F_{W \mid Z,X}^{-1}(u \mid z_u, x)$. 

Outside this identifiable support, there is no constraint, and the optimal solution trivially assigns probability mass to the global outcome bounds to minimize or maximize the objective. Substituting this into our separated CCOT formulation yields the closed-form bounds.

\begin{theorem}[Closed-Form Bounds under Strictly Monotonic Treatment Selection] \label{thm:tsem_bounds}
Suppose the generalized structural conditions in \cref{asp:structure_generalized}, the strict monotonicity in \cref{asp:strict_monotonicity}, and the outcome boundedness in \cref{asp:bounded_y} hold. For any identifiable weight $\omega(w,x,u)$, the sharp lower bound for $\theta_\omega$ is:
\begin{align*}
    \underline{\theta} &= \int_{\mathcal{W}} \mathbb{E}_{X} \Bigg[ \int_0^1 \indicator\Big(u \in \mathcal{U}_{\text{id}}(x,w)\Big) \mathbb{E}_{\text{obs}}[Y \mid w, Z=z_u, X] \,\omega(w,X,u)\,\mathrm{d}u \Bigg] \mathrm{d}w \\
    &\quad + \int_{\mathcal{W}} \mathbb{E}_{X} \Bigg[ \int_0^1 \indicator\Big(u \notin \mathcal{U}_{\text{id}}(x,w)\Big) \Big( y_{\min} \max\{0, \omega(w,X,u)\} + y_{\max} \min\{0, \omega(w,X,u)\} \Big) \mathrm{d}u \Bigg] \mathrm{d}w,
\end{align*}
where $z_u$ satisfies $F_{W \mid Z,X}(w \mid z_u, x) = u$. The sharp upper bound $\overline{\theta}$ is obtained symmetrically by swapping $y_{\min}$ and $y_{\max}$ in the trivial bound term.
\end{theorem}

\subsection{Ordered Choice Model}\label{subsec:multi_treatment}
Next, let us consider the multi-valued treatment setting, which is known as the ordered choice model \citep{heckman2007econometric}. In this case, $\mathcal{W} $ is a finite set and $\text{supp}(F_{U \mid Z=z, X=x, W=w})$ is an interval for each $z$. Let us denote this interval as $I_{x,w}(z)$. Then, the first constraint of (\ref{eq:ot_generalized_separate}) simplifies to:
\begin{equation}\label{eq:contraints_multitreatment}
    \frac{1}{|I_{x,w}(z)|} \int_{I_{x,w}(z)} \pi_{w,x}(\mathrm{d}y \mid U=u) \mathrm{d}u = \mathbb{P}_{\text{obs}}(\mathrm{d}y \mid W=w, Z=z, X=x), \quad \forall z \in \mathcal{Z}(x,w).
\end{equation}
In particular, if $|I_{x,w}(z)|=0$, the interval $I_{x,w}(z)$ degenerates to a single point $\{u_0\}$ and the constraint reduces to $\pi_{w,x}(\mathrm{d}y \mid U=u_0) = \mathbb{P}_{\text{obs}}(\mathrm{d}y \mid W=w, Z=z, X=x)$, i.e., the conditional distribution is point-identified at $u_0$. By convention, the left-hand side of (\ref{eq:contraints_multitreatment}) is interpreted as $\pi_{w,x}(\mathrm{d}y \mid U=u_0)$ in this degenerate case.

The main difficulty of (\ref{eq:contraints_multitreatment}) is that the intervals in $\{I_{x,w}(z)\}_{z \in \mathcal{Z}(x,w)}$ may overlap with each other, making it hard to disentangle the constraints like in the binary treatment setting. Surprisingly, we find that if $\{I_{x,w}(z)\}_{z \in \mathcal{Z}(x,w)}$ forms a $\pi$-system, closed-form solutions still exist.

\begin{assumption}[$\pi$-System] \label{asp:algebra}
Assume that the collection of intervals $\{I_{x,w}(z)\}_{z \in \mathcal{Z}(x,w)}$ forms a $\pi$-system for all $x \in \mathcal{X}$ and $w \in \mathcal{W}$. That is, for any $z, z' \in \mathcal{Z}(x,w)$, the intersection is also in the family: $I_{x,w}(z) \cap I_{x,w}(z') \in \{I_{x,w}(z'')\}_{z'' \in \mathcal{Z}(x,w)} \cup \{\emptyset\}$.
\end{assumption}

In particular, if the $\{I_{x,w}(z)\}_{z \in \mathcal{Z}(x,w)}$ share the same start point or end point, they naturally satisfy this algebraic property, which is exactly the case in the binary treatment setting. More generally, \cref{asp:algebra} endows $\{I_{x,w}(z)\}$ with a hierarchical Directed Acyclic Graph (DAG) ordering under strict inclusion---$z$ is an ancestor of $z'$ whenever $I_{x,w}(z') \subsetneq I_{x,w}(z)$---which, as we show below, is exactly the structure needed to disentangle the overlapping constraints.

For each $z \in \mathcal{Z}(x,w)$, let $\text{children}(z)$ and $\text{Dec}(z)$ denote its direct children and all descendants in this DAG. Define the disjoint \emph{isolated sub-region}
\begin{equation*}
    J_{x,w}(z) = I_{x,w}(z) \setminus \bigcup_{z' \in \text{children}(z)} I_{x,w}(z'),
\end{equation*}
and, recursively from the leaves upward, the \emph{isolated measure}
\begin{equation}\label{eq:mu_isolated}
    \mu_{w,z,x}(\mathrm{d}y) = \frac{1}{|J_{x,w}(z)|} \Bigg( |I_{x,w}(z)| \mathbb{P}_{\text{obs}}(\mathrm{d}y \mid W=w, Z=z, X=x) - \sum_{z' \in \text{Dec}(z)} |J_{x,w}(z')| \mu_{w,z',x}(\mathrm{d}y) \Bigg).
\end{equation}
The measure $\mu_{w,z,x}$ is the general analogue of $\mu_{1,i}$ from the binary case: it isolates the conditional distribution of $Y(w)$ on the disjoint slice $J_{x,w}(z)$. As in the binary setting, $\mu_{w,z,x}$ is a valid (nonnegative) probability measure under correct model specification, and its nonnegativity serves as a testable implication of the structural model. We also let $J_{x,w}(\emptyset) = [0,1] \setminus \bigcup_{z \in \mathcal{Z}(x,w)} I_{x,w}(z)$ denote the unconstrained domain.

In \cref{sec:proofs_extension}, we show that under \cref{asp:algebra} the original constraint system (\ref{eq:contraints_multitreatment}) is equivalent to the {disentangled} marginal constraints
\begin{equation}\label{eq:simplified_marginal_constraint}
    \frac{1}{|J_{x,w}(z)|} \int_{J_{x,w}(z)} \pi_{w,x}(\mathrm{d}y \mid U=u) \mathrm{d}u = \mu_{w,z,x}(\mathrm{d}y), \quad \forall z \in \mathcal{Z}(x,w),
\end{equation}
each acting on a single disjoint region $J_{x,w}(z)$. Because the $\{J_{x,w}(z)\}_{z \in \mathcal{Z}(x,w)}$ partition the identified support $\mathcal{U}_{\text{id}}(x,w)$ and the objective in (\ref{eq:ot_generalized_separate}) is additively separable across these regions, the global problem decomposes into independent 1D optimal transport sub-problems with product cost---one per $J_{x,w}(z)$, pairing $\mu_{w,z,x}$ with $\mathrm{Unif}(J_{x,w}(z))$---plus a pointwise trivial bound on the unconstrained domain $J_{x,w}(\emptyset)$. Applying \cref{thm:1d_ot} to each sub-problem yields the following closed-form result.

\begin{theorem}[Closed-Form Bounds for Multi-Valued Treatment] \label{thm:multivalued_bounds}
Suppose \cref{asp:structure_generalized}, \cref{asp:algebra}, and \cref{asp:bounded_y} hold, $\mathcal{Z}$ is a discrete set, $\mathcal{W}$ is finite, and for every $(x,w)$ and $z \in \mathcal{Z}(x,w)$ the support $\mathrm{supp}(F_{U \mid Z=z, X=x, W=w})$ is an interval $I_{x,w}(z) \subseteq [0,1]$. Let $Q_{Y, J_{x,w}(z)}(t)$ be the quantile function of the isolated measure $\mu_{w,z,x}$, and let $Q_{\omega, J_{x,w}(z)}(t)$ be the quantile function of $\omega(w,x,U)$ for $U \sim \text{Unif}(J_{x,w}(z))$. 

For a given $x \in \mathcal{X}$ and $w \in \mathcal{W}$, let $J_{x,w}(\emptyset) = [0,1] \setminus \bigcup_{z \in \mathcal{Z}(x,w)} I_{x,w}(z)$ be the unconstrained domain. The sharp lower bound $\underline{\theta}(x,w)$ is given by:
\begin{align*}
    \underline{\theta}(x,w) &= \sum_{z \in \mathcal{Z}(x,w)} |J_{x,w}(z)| \int_0^1 Q_{Y, J_{x,w}(z)}(t) Q_{\omega, J_{x,w}(z)}(1-t) \, \mathrm{d}t \\
    &\quad + \int_{J_{x,w}(\emptyset)} \Big( y_{\min} \max\{0, \omega(w,x,u)\} + y_{\max} \min\{0, \omega(w,x,u)\} \Big) \mathrm{d}u.
\end{align*}
The global sharp lower bound is $\underline{\theta}_\omega = \sum_{w \in \mathcal{W}} \mathbb{E}_{X}[\underline{\theta}(X,w)]$. The sharp upper bound $\overline{\theta}$ is obtained symmetrically by taking the comonotonic coupling integral $\int_0^1 Q_{Y, J_{x,w}(z)}(t) Q_{\omega, J_{x,w}(z)}(t) \mathrm{d}t$ and swapping $y_{\min}$ and $y_{\max}$ in the trivial bound term.
\end{theorem}

\begin{remark}[Mixed Continuous-Discrete Treatments] \label{rmk:mixed_distributions}
    In general, the solution to the global CCOT problem in (\ref{eq:ot_generalized_separate}) is a hybrid of the results in \cref{subsec:triangle} and \cref{subsec:multi_treatment}. The exact nature of the subproblem for a given stratum $(x,w)$ depends directly on the localized behavior of the treatment distribution.

    If the CDF $F_{W\mid Z,X}(w)$ is continuous at the evaluation point $w$, the structural quantile function is strictly increasing, collapsing the conditional distribution $F_{U \mid Z=z, X=x, W=w}$ to a point mass. In these regions, the subproblem is perfectly constrained, and the potential outcome expectation is point-identified as in (\ref{eq:id_point_mass}).  Conversely, if the data exhibit a discrete probability mass at $w$, the latent variable $U$ maps to an interval. Over these discrete mass points, the subproblem natively transitions into the localized 1D optimal transport problem described in \cref{subsec:multi_treatment}. Thus, our unified CCOT framework naturally accommodates complex empirical settings with mixed continuous-discrete treatments.
\end{remark}

%% file: sections/est_inf.tex
In this section, we develop estimation and inference results for the closed-form bounds derived in \cref{sec:pid}. For the discrete instrument setting, we leverage DML \citep{chernozhukov2018double} to accommodate high-dimensional covariates, constructing Neyman-orthogonal scores that yield $\sqrt{n}$-consistent and asymptotically normal estimators. In the continuous instrument setting, we characterize the corresponding nonparametric convergence rates.

We focus our exposition on the bounds for the PRTE, as the derivation for other causal quantities follows similarly. For simplicity, we first develop theory for the unscaled numerator $\mathbb{E}[Y^{q}-Y]$. In the empirical designs considered below, the PRTE denominator is the known policy-shift size $\alpha > 0$, so the reported bounds and confidence intervals for $\text{PRTE}_\alpha$ are obtained by dividing the corresponding numerator bounds by $\alpha$. More generally, when $\mathbb{E}[q(Z,X)-p(Z,X)]$ is not fixed by design, it is point-identified and can be estimated separately. Thus, our target weight function simplifies to:
\begin{align*}
    \omega(x,u) \equiv \mathbb{E}_{\text{obs}}[\indicator(u \leqslant  q(Z,X)) -  \indicator(u \leqslant  p(Z,X)) \mid X=x].
\end{align*}
We model the alternative policy as $q(z,x) = \phi(z,x,p(z,x))$ for a known function $\phi$. Common examples for $\phi$ include uniform propensity shifts $\phi(z,x,p) = p + \alpha$ or $p(1+\alpha)$ for a constant $\alpha$, or setting $\phi(z,x,p) = r(z,x)$ to evaluate a specific alternative targeting policy $r$.

\subsection{Discrete Instrument Setting}

For a fixed $x$, the weight function $\omega(x,u)$ is a step function that is constant and monotonic within each latent interval $u \in (p_i(x), p_{i+1}(x)]$. This piecewise structure allows us to decompose the closed-form lower bound integral into a finite weighted sum over the instrument level sets.

Define $q_{j,i}(x)$ as the $j$-th smallest value of the alternative policy propensity $q(z,x)$ that falls strictly inside the baseline interval $[p_i(x), p_{i+1}(x))$, such that:
\begin{equation*}
    p_i(x) \leqslant q_{1,i}(x) < \dots < q_{l_i, i}(x) <  p_{i+1}(x),
\end{equation*}
with the convention $ q_{0,i}(x) = p_i(x)$. Let the baseline and alternative instrument level sets be:
\begin{align*}
    S_k = \{z \in \mathcal{Z} : p(z,x) = p_k(x)\}, \quad T_{j,k} = \{z \in \mathcal{Z} : q(z,x) = q_{j,k}(x)\}.
\end{align*}

Before presenting the rewriting of the target parameter, we impose two regularity conditions on the alternative policy and the propensity scores.

\begin{assumption}[Smoothness of $\phi$ and Propensity Score]\label{asp:phi_smooth}
    (i) $\phi(z,x,p)$ is differentiable with respect to $p$, and $\frac{\partial \phi}{\partial p}$ is bounded almost surely for all $z \in \mathcal{Z}$ and $x \in \mathcal{X}$. (ii) For each $z \in \mathcal{Z}$, the conditional propensity score $p(z,\cdot)$ is continuous on $\mathcal{X}$.
\end{assumption}
Part~(i) will be used in the pathwise derivation of the orthogonal score and in the asymptotic analysis. Because $q(z,x) = \phi(z,x,p(z,x))$, part~(ii) combined with (i) also implies that $q(z,\cdot)$ is continuous on $\mathcal{X}$ for each $z \in \mathcal{Z}$.

To ensure the level sets can be consistently estimated from the data without asymptotic ambiguity, we further impose the following gap assumption on the propensity scores.

\begin{assumption}[Propensity Score Gap]\label{asp:prop_gap}
    There exists a constant $c_{\text{gap}} > 0$ such that for all $x \in \mathcal{X}$ and for all $v, v' \in R_x = \{p(z,x), q(z,x)\}_{z \in \mathcal{Z}}$, either $v' = v$ or $|v - v'| > c_{\text{gap}}$.
\end{assumption}

Fix any two pairs $(z,\tau)$ and $(z',\tau')$ with $\tau, \tau' \in \{p,q\}$. The map $x \mapsto \tau(z,x) - \tau'(z',x)$ is continuous on $\mathcal{X}$ by \cref{asp:phi_smooth}, while \cref{asp:prop_gap} requires its value at every $x$ to be either exactly zero or of absolute value strictly greater than $c_{\text{gap}}$. A continuous function whose range avoids the punctured neighborhood $(-c_{\text{gap}}, c_{\text{gap}}) \setminus \{0\}$ cannot change sign, so the relative ordering of the elements of $R_x$ is invariant across $\mathcal{X}$. We therefore suppress the dependence on $x$ in $S_k$, $T_{j,k}$, $K$, and $l_k$ from now on.

The bound in \cref{thm:discrete_lower_bound} admits a representation as an expectation of a finite weighted sum over instrument level sets and their sub-intervals. Before stating this representation, we introduce the nuisance components that carry the observable content of the bound.

\paragraph{Nuisance components.} For each interval index $k \in \{0,\dots,K\}$ and sub-interval index $j \in \{0,\dots,l_k\}$, define:
\begin{itemize}
    \item[--] \emph{Instrument-level-set weights.}
    \begin{align*}
        \gamma_{\text{full},k}(x) &\coloneqq \mathbb{P}_{\text{obs}}\!\left(Z \in \bigcup_{i=k+1}^K \bigcup_{l=0}^{l_i} T_{l,i} \mathrel{\Big|} X=x\right) - \mathbb{P}_{\text{obs}}\!\left(Z \in \bigcup_{i=k+1}^K S_i \mathrel{\Big|} X=x\right), \\
        \gamma_{j,k}(x) &\coloneqq \mathbb{P}_{\text{obs}}\!\left(Z \in \bigcup_{l=j}^{l_k} T_{l,k} \mathrel{\Big|} X=x\right), \qquad \gamma_K(x) \coloneqq \mathbb{P}_{\text{obs}}\!\left(Z \in \bigcup_{j=1}^{l_K} T_{j,K} \mathrel{\Big|} X=x\right).
    \end{align*}
    These are the aggregate masses that the step weight $\omega$ assigns to instruments; each is estimable by a standard regression of an instrument-set indicator on $X$.
    \item[--] \emph{Scaled complier mean and sub-interval quantile integral.} With the relative threshold $\kappa_{j,k}(x) \coloneqq (q_{j,k}(x) - p_k(x))/(p_{k+1}(x) - p_k(x)) \in [0,1]$,
    \begin{equation*}
        J_{\text{full},k}(x) \coloneqq (p_{k+1}(x)-p_k(x))\,\mathbb{E}_{\mu_{1,k\mid x}}[Y], \qquad J_{j,k}(x) \coloneqq (p_{k+1}(x)-p_k(x))\int_{\kappa_{j,k}(x)}^{\kappa_{j+1,k}(x)} Q_{Y,k\mid x}(u)\,\mathrm{d}u,
    \end{equation*}
    where $\mu_{1,k\mid x}$ is the identified complier distribution of \cref{thm:discrete_lower_bound} and $Q_{Y,k\mid x}$ is its quantile function.
    \item[--] \emph{Trivial-bound boundary term.}
    \begin{equation*}
        \Delta_K(x) \coloneqq y_{\min}\, \mathbb{E}_{\text{obs}}\!\left[ \indicator\!\Big(Z \in \bigcup_{j=1}^{l_K} T_{j,K}\Big) (q(Z,X) - p_K(X)) \mathrel{\Big|} X=x \right].
    \end{equation*}
\end{itemize}

\begin{proposition}[Discrete-Instrument Reduction to Nuisance Functional]\label{prop:discrete_reduction}
    Suppose \cref{asp:structure}, \cref{asp:bounded_y}, \cref{asp:phi_smooth}, and \cref{asp:prop_gap} hold and $\mathcal{Z}$ is finite. Then the sharp lower bound in \cref{thm:continuous_lower_bound} admits the representation
    \begin{align}\label{eq:target_discrete_iv}
        \underline{\theta}_{\omega,1} = \mathbb{E}\Bigg[\sum_{k=0}^{K-1}\bigg(\gamma_{\text{full},k}(X)\, J_{\text{full},k}(X) + \sum_{j=1}^{l_k} \gamma_{j,k}(X)\, J_{j-1,k}(X)\bigg) + \Delta_K(X) \Bigg].
    \end{align}
    The sharp upper bound admits an analogous representation, obtained by replacing $Q_{Y,k\mid x}(u)$ by $Q_{Y,k\mid x}(1-u)$ inside $J_{j,k}$ and swapping $y_{\min}$ with $y_{\max}$ in $\Delta_K$.
\end{proposition}

\cref{prop:discrete_reduction} is the key analytical step that converts the optimal-transport integral in \cref{thm:discrete_lower_bound} into a Neyman-orthogonalizable functional of standard nuisance quantities---conditional probabilities, conditional means, and conditional quantile integrals. The proof, given in \cref{app:closed_form_rewriting}, exploits the piecewise-constant, monotone structure of $\omega(X,\cdot)$ on each baseline interval and identifies the step-function value attached to each sub-interval $(q_{j,k}(X), q_{j+1,k}(X))$.

The sub-interval term $J_{j,k}$ still involves the conditional quantile function $Q_{Y,k\mid X}$, which is delicate to estimate directly. The next proposition replaces it with conditional expectations in the two outcome regimes of interest.

\begin{proposition}[Specialization of the Sub-Interval Contribution]\label{prop:J_specialization}
    Fix $k \in \{0,\dots,K-1\}$ and let $\nu_{j,k}(x) \coloneqq Q_{Y,k\mid x}(\kappa_{j,k}(x))$ for $j \in \{0,\dots,l_k\}$.
    \begin{enumerate}[label=(\roman*)]
        \item \textbf{Continuous outcome.} If $\mu_{1,k\mid X}$ is continuous at $\nu_{j,k}(X)$ and $\nu_{j+1,k}(X)$ almost surely, then, for $j \in \{0,\dots,l_k-1\}$,
        \begin{equation}\label{eq:J_cexpectation}
            J_{j,k}(X) = (p_{k+1}(X) - p_k(X))\,\mathbb{E}_{\mu_{1,k\mid X}}\!\big[ Y\,\indicator\!\big(\nu_{j,k}(X) < Y \leqslant \nu_{j+1,k}(X)\big) \big].
        \end{equation}
        \item \textbf{Binary outcome.} If $Y \in \{0,1\}$, then $J_{\text{full},k}(X) = P_{1,k+1}(X) - P_{1,k}(X)$ with $P_{1,k}(x) \coloneqq \mathbb{E}_{\text{obs}}[YW \mid Z \in S_k, X=x]$, and for $j \in \{1,\dots,l_k\}$,
        \begin{equation}\label{eq:h_jk_binary}
            J_{j-1,k}(X) = h_{j,k}^{-}(X) \coloneqq \max\!\Big(0,\; q_{j,k}(X) - \max\!\big(q_{j-1,k}(X),\; p_{k+1}(X) - J_{\text{full},k}(X)\big)\Big).
        \end{equation}
        Consequently, \eqref{eq:target_discrete_iv} collapses to
        \begin{align}\label{eq:binary_lower_bound}
            \underline{\theta}_{\omega,1} = \mathbb{E}\Bigg[\sum_{k=0}^{K-1}\bigg(\gamma_{\text{full},k}(X)\, J_{\text{full},k}(X) + \sum_{j=1}^{l_k} \gamma_{j,k}(X)\, h_{j,k}^{-}(X)\bigg) + \Delta_K(X)\Bigg].
        \end{align}
    \end{enumerate}
\end{proposition}

Part~(i) converts the sub-interval quantile integral into a plain conditional expectation of $Y$ on the quantile-defined band $(\nu_{j,k}, \nu_{j+1,k}]$, so only the $l_k+1$ boundary quantiles $\nu_{j,k}$---rather than the full quantile process---enter the estimation. Part~(ii) leverages the step-function form of $Q_{Y,k\mid X}$ when $Y$ is binary: the complier mean $\theta_k(X) = \mathbb{E}_{\mu_{1,k\mid X}}[Y]$ reduces to a difference of observed joint probabilities, and the integral of the step function over $[\kappa_{j-1,k}, \kappa_{j,k}]$ evaluates to the closed-form $h_{j,k}^{-}$. Both parts are proved in \cref{app:J_specialization}.

Together, \cref{prop:discrete_reduction,prop:J_specialization} expose \eqref{eq:target_discrete_iv} as a functional of conditional expectations, boundary quantiles, and instrument-level-set probabilities---the form required for DML-style $\sqrt{n}$-inference. We construct the corresponding estimators and establish their asymptotic properties next, treating the two outcome regimes in turn.

\subsubsection{Continuous Outcome}\label{subsubsec:continuous_outcome}

We first treat the case where the conditional complier distribution $\mu_{1,k\mid X}$ is continuous at the boundary quantiles $\nu_{j,k}(X)$, so \cref{prop:J_specialization}(i) applies. The representation~\eqref{eq:J_cexpectation} converts each sub-interval term $J_{j,k}(X)$ into a conditional expectation over the band $(\nu_{j,k}(X), \nu_{j+1,k}(X)]$, thereby reducing the nuisance list from the full quantile process to the boundary quantiles $\nu_{j,k}$ alone. The Neyman-orthogonal score construction additionally requires the conditional indicator expectations
\begin{align*}
    M_{j,k}^{+}(x) \coloneqq \mathbb{E}[\indicator(Y \leqslant \nu_{j,k}(x)) \mid Z \in S_{k+1}, W=1, X=x], \\
     M_{j,k}^{-}(x) \coloneqq \mathbb{E}[\indicator(Y \leqslant \nu_{j,k}(x)) \mid Z \in S_k, W=1, X=x],
\end{align*}
which serve as Riesz representers for the pathwise derivatives with respect to $\nu_{j,k}$.

Given the functional form in~(\ref{eq:target_discrete_iv}) and~(\ref{eq:J_cexpectation}), our DML estimator takes the aggregated Neyman-orthogonal form
\begin{equation}\label{eq:estimator_continuous_outcome}
    \hat{\underline{\theta}}_{\omega, 1} = \frac{1}{|I_2|} \sum_{i \in I_2} \Bigg[ \sum_{k=0}^{K-1}\bigg(\hat{\psi}_{\text{full},k}^{\text{prod}}(O_i; \hat{\eta}) + \sum_{j=1}^{l_k} \hat{\psi}_{j,k}^{\text{prod}}(O_i; \hat{\eta})\bigg) + \hat{\psi}_{\Delta_K}(O_i; \hat{\eta}) \Bigg],
\end{equation}
where each $\hat{\psi}$ is a Neyman-orthogonal score that corrects for the estimation error of its associated nuisance (explicit formulas are given in \cref{app:scores_continuous}). The rest of this subsection develops the nuisance estimators $\hat{\eta}$ trained on an auxiliary sample $I_1$; \cref{thm:dml_normality} below then establishes $\sqrt{n}$-consistency and asymptotic normality of~(\ref{eq:estimator_continuous_outcome}).

\paragraph{Estimation Procedure.} Following the standard DML template, we randomly partition the sample into an auxiliary set $I_1$ and a main estimation set $I_2$, train all nuisance estimators on $I_1$, and average plug-in orthogonal scores on $I_2$. \cref{alg:dml_continuous_outcome} summarizes the full procedure; the explicit moment equation for the conditional quantile $\nu_{j,k}$, the conditional-expectation components $M_{j,k}^{\pm}$, $J_{j,k}^{\pm}$, $J_{\mathrm{full},k}^{\pm}$, and the Neyman-orthogonal scores $\psi^{\mathrm{prod}}$ are deferred to \cref{app:procedure_continuous_outcome} and \cref{app:scores_continuous}.

\begin{algorithm}[h]
\caption{DML estimator: discrete instrument, continuous outcome}\label{alg:dml_continuous_outcome}
\begin{algorithmic}[1]
\STATE Split sample into $I_1$ (auxiliary) and $I_2$ (main).
\STATE \textbf{On $I_1$:} estimate propensity $\hat p(z,x)$ and set $\hat q(z,x) = \phi(z,x,\hat p(z,x))$.
\STATE Cluster $\{\bar p(z), \bar q(z)\}_{z \in \mathcal Z}$ at radius $c_{\mathrm{gap}}/2$ to obtain level sets $\widehat S_k$ and $\widehat T_{j,k}$.
\STATE Estimate weighting probabilities $\hat\pi_k(x)$, $\hat\gamma_{\mathrm{full},k}(x)$, $\hat\gamma_{j,k}(x)$, $\hat\gamma_K(x)$ by ML regression on indicators ({$\pi_k(x)\coloneqq\mathbb{P}(Z\in S_k\mid X=x)$; $\gamma_{\mathrm{full},k},\gamma_{j,k},\gamma_K$ defined preceding \cref{prop:discrete_reduction}}).
\STATE Estimate boundary quantiles $\hat\nu_{j,k}(x)$ via the moment-equation regression in \cref{app:procedure_continuous_outcome}.
\STATE Estimate conditional expectations $\hat M_{j,k}^{\pm}(x)$, $\hat J_{j,k}^{\pm}(x)$, $\hat J_{\mathrm{full},k}^{\pm}(x)$ by flexible ML.
\STATE \textbf{On $I_2$:} compute orthogonal scores $\hat\psi^{\mathrm{prod}}_{\mathrm{full},k}, \hat\psi^{\mathrm{prod}}_{j,k}, \hat\psi_{\Delta_K}$ from \cref{app:scores_continuous}.
\STATE Return $\hat{\underline\theta}_{\omega,1}$ as in~(\ref{eq:estimator_continuous_outcome}) and $\hat\sigma^2 = |I_2|^{-1}\sum_{i\in I_2}(\hat\psi(O_i;\hat\eta) - \hat{\underline\theta}_{\omega,1})^2$.
\end{algorithmic}
\end{algorithm}

To ensure the debiased terms remain well-behaved and that inverse probability weights do not explode, we require the following strict overlap assumption.

\begin{assumption}[Overlap]\label{asp:overlap}
    There exists a constant $c_\pi > 0$ such that $\pi_k(x) \geqslant c_\pi$ for all $x \in \mathcal{X}$ and for all interval indices $k \in \{0, \dots, K\}$.
\end{assumption}

The following theorem formally establishes the asymptotic normality of our DML estimator.

\begin{theorem}[Asymptotic Normality] \label{thm:dml_normality}
    Suppose \cref{asp:structure}, \cref{asp:bounded_y}, \cref{asp:phi_smooth}, \cref{asp:prop_gap} and \cref{asp:overlap} hold, $\mathcal{Z}$ is finite and the estimators for the nuisance parameters $\eta$ trained on $I_1$ converge at a rate of at least $o_P(n^{-1/4})$ in the $L_2$ norm. Furthermore, assume that for each $j = 1, \ldots, l_k$, $k = 0, \ldots, K-1$ and $k' \in \{k, k+1\}$, the conditional density $f_{Y \mid X, Z \in S_{k'}, W=1}(y \mid X)$ exists and is continuous at $y = \nu_{j,k}(X)$ almost surely. Then, the sample-split estimator is $\sqrt{n}$-consistent and asymptotically normal:
    \begin{equation*}
        \sqrt{n/2} \big( \hat{\underline{\theta}}_{\omega, 1} - \underline{\theta}_{\omega, 1} \big) \xrightarrow{d} \mathcal{N}(0, \sigma^2),
    \end{equation*}
    where the asymptotic variance is the variance of the true orthogonal score, $\sigma^2 = \mathbb{E}\big[ ({\psi}(O; \eta) - \underline{\theta}_{\omega, 1})^2 \big]$. It can be consistently estimated by its sample analog $\hat{\sigma}^2 = \frac{1}{|I_2|} \sum_{i \in I_2} \big(\hat{\psi}(O_i; \hat{\eta}) - \hat{\underline{\theta}}_{\omega, 1}\big)^2$.
\end{theorem}

\subsubsection{Discrete Outcome}\label{subsubsec:discrete_outcome}

When the outcome $Y$ takes finitely many values---in particular, when $Y \in \{0,1\}$ is binary---the conditional quantile function $Q_{Y,k\mid X}$ becomes a step function with atoms, and the continuity hypothesis of \cref{prop:J_specialization}(i) fails. Instead, \cref{prop:J_specialization}(ii) applies: the representation~\eqref{eq:binary_lower_bound} expresses $J_{\text{full},k}$ as a difference of observed joint probabilities $P_{1,k+1} - P_{1,k}$ and evaluates the sub-interval contribution in closed form through $h_{j,k}^{-}$, so the nuisance list collapses to conditional expectations alone and quantile estimation is avoided entirely. We present the binary case; the extension to general discrete outcomes is analogous.

Given the functional form in~\eqref{eq:binary_lower_bound}, our DML estimator takes an aggregated Neyman-orthogonal form analogous to the continuous-outcome case:
\begin{equation}\label{eq:binary_estimator}
    \hat{\underline{\theta}}_{\omega,1} = \frac{1}{|I_2|} \sum_{i \in I_2} \Bigg[ \sum_{k=0}^{K-1}\bigg(\hat{\psi}_{\text{full},k}^{\text{prod}}(O_i; \hat{\eta}) + \sum_{j=1}^{l_k} \hat{\psi}_{j,k}^{\text{prod}}(O_i; \hat{\eta})\bigg) + \hat{\psi}_{\Delta_K}(O_i; \hat{\eta}) \Bigg],
\end{equation}
where each $\hat{\psi}$ is a Neyman-orthogonal score (explicit formulas in \cref{app:scores_discrete}). \cref{thm:dml_discrete_outcome} below establishes $\sqrt n$-consistency and asymptotic normality.

\paragraph{Estimation Procedure.} The procedure mirrors \cref{alg:dml_continuous_outcome}, with the conditional-quantile and indicator-expectation steps removed. \cref{alg:dml_binary_outcome} summarizes; full details and the construction of the orthogonal scores are in \cref{app:procedure_binary_outcome} and \cref{app:scores_discrete}.

\begin{algorithm}[h]
\caption{DML estimator: discrete instrument, binary outcome}\label{alg:dml_binary_outcome}
\begin{algorithmic}[1]
\STATE Split sample into $I_1$, $I_2$. On $I_1$, estimate $\hat p$, $\hat q$, and the level sets $\widehat S_k$, $\widehat T_{j,k}$ exactly as in \cref{alg:dml_continuous_outcome}.
\STATE Estimate $\hat P_{1,k}(x) = \widehat{\mathbb E}[YW \mid Z\in\widehat S_k, X=x]$ and $\hat\gamma_{\mathrm{full},k}(x), \hat\gamma_{j,k}(x), \hat\gamma_K(x)$ by ML regression.
\STATE For each $i$, sort $(\hat P_{1,0}(X_i),\dots,\hat P_{1,K}(X_i))$ to enforce monotonicity.
\STATE \textbf{On $I_2$:} compute orthogonal scores $\hat\psi^{\mathrm{prod}}, \hat\psi_{\Delta_K}$ from \cref{app:scores_discrete} and average to obtain $\hat{\underline\theta}_{\omega,1}$ as in~(\ref{eq:binary_estimator}).
\end{algorithmic}
\end{algorithm}

To ensure differentiability of the target functional with respect to the nuisance parameters, we impose the following gap condition, which replaces the conditional-density smoothness assumption of \cref{thm:dml_normality}.

\begin{assumption}[Discrete Outcome Gap]\label{asp:discrete_outcome_gap}
    For all $k = 0, \ldots, K-1$ and $x \in \mathcal{X}$, the integration boundaries $p_{k+1}(x) - J_{\text{full},k}(x)$ and $p_k(x) + J_{\text{full},k}(x)$ do not coincide with any alternative policy level $q_{j,k}(x)$ for $j = 1, \ldots, l_k$.
\end{assumption}
This condition guarantees that the target is pathwise-differentiable in the nuisances despite the $\max$ operators appearing in $h_{j,k}^{-}$. The following theorem establishes the asymptotic normality of our DML estimator in the discrete-outcome setting.

\begin{theorem}[Asymptotic Normality, Discrete Outcome]\label{thm:dml_discrete_outcome}
    Suppose \cref{asp:structure}, \cref{asp:bounded_y}, \cref{asp:phi_smooth}, \cref{asp:prop_gap}, \cref{asp:overlap}  and \cref{asp:discrete_outcome_gap} hold, $\mathcal{Z}$ is finite, outcome is binary and the nuisance estimators trained on $I_1$ converge at a rate of at least $o_P(n^{-1/4})$ in the $L_2$ norm. Then, the sample-split estimator $\hat{\underline{\theta}}_{\omega,1}$ defined in \eqref{eq:binary_estimator} is $\sqrt{n}$-consistent and asymptotically normal:
    \begin{equation*}
        \sqrt{n/2}\big(\hat{\underline{\theta}}_{\omega,1} - \underline{\theta}_{\omega,1}\big) \xrightarrow{d} \mathcal{N}(0, \sigma^2),
    \end{equation*}
    where $\sigma^2 = \mathbb{E}\big[(\psi(O; \eta) - \underline{\theta}_{\omega,1})^2\big]$. It can be consistently estimated by $\hat{\sigma}^2 = \frac{1}{|I_2|}\sum_{i \in I_2}\big(\hat{\psi}(O_i; \hat{\eta}) - \hat{\underline{\theta}}_{\omega,1}\big)^2$.
\end{theorem}

\subsection{Continuous Instrument Setting}
We now proceed to estimation in the continuous instrument setting. For simplicity, we will assume that $\mathcal{Z}$ is connected throughout this subsection; the multiply connected component case can be derived similarly. When the instrumental variable $Z$ is continuous, the propensity score $p(Z,X)$ takes a continuous support, and the piecewise-constant structure of the weight function no longer holds. The general bound~(\ref{eq:covariate_continuous_lower_bound}) contains a derivative $\frac{\partial}{\partial u}\mathbb{E}_{\text{obs}}[YW \mid p(Z,X)=u, X]$ that is non-trivial to estimate directly. The following proposition leverages Fubini's theorem to collapse the derivative into a finite difference of conditional-expectation evaluations, yielding a representation amenable to plug-in ML estimation.

\begin{proposition}[Continuous-Instrument Reduction via Fubini]\label{prop:continuous_reduction}
    Suppose \cref{asp:structure}, \cref{asp:bounded_y}, \cref{asp:continuity_p}, and \cref{asp:phi_smooth} hold, and $\mathcal{Z}$ is connected. Let $g_1(u,X) \coloneqq \mathbb{E}_{\text{obs}}[Y W \mid p(Z,X)=u, X]$, and let $\underline{p}(X), \overline{p}(X)$ denote the conditional infimum and supremum of $p(Z,X)$ given $X$. Then the sharp lower bound in \cref{thm:covariate_continuous_lower_bound} admits the representation
    \begin{equation} \label{eq:continuous_simplified_target}
     \begin{split}
            \underline{\theta}_{\omega, 1} &= \mathbb{E}_{Z,X} \Bigg[ - \int_{\min\{\underline{p}(X), q(Z,X)\}}^{\underline{p}(X)} Q_{Y,\underline{p} \mid X}\!\left({v}/{\underline{p}(X)}\right) \, \mathrm{d}v \\
        &\qquad + g_1\big(\min\{\overline{p}(X), q(Z,X)\}, X\big) - g_1\big(p(Z,X), X\big)  \\
        &\qquad + y_{\min} \max\{0, q(Z,X) - \overline{p}(X)\} \Bigg].
     \end{split}
    \end{equation}
    The sharp upper bound admits an analogous representation, obtained by replacing $Q_{Y,\underline{p} \mid X}(v/\underline{p}(X))$ by $Q_{Y,\underline{p} \mid X}(1 - v/\underline{p}(X))$ in the first term and $y_{\min}$ by $y_{\max}$ in the third term.
\end{proposition}

\cref{prop:continuous_reduction} is the continuous-instrument analog of \cref{prop:discrete_reduction}: it converts the sharp bound into an expectation over $(Z,X)$ that depends on estimable nuisances only. The first term is a boundary-quantile tail integral capturing the optimal-coupling contribution on $(0, \underline{p}(X))$ through the conditional quantile function of the identified complier distribution $\mu_{1,\underline{p}\mid X}$. The second term replaces the derivative of $g_1$ by a finite difference of $g_1$-values at the truncated limits---this is the key simplification enabled by Fubini, since $g_1$ itself is a standard conditional expectation estimable by flexible ML regressions. The third term is the closed-form trivial-bound contribution on $(\overline{p}(X),1]$, where $\omega(X,u) \geqslant 0$ for the PRTE weight. Crucially, the entire target requires only conditional expectations and tail quantiles, avoiding derivative estimation entirely. The proof is given in \cref{app:procedure_fubini}.

\cref{prop:continuous_reduction} motivates the following plug-in estimator, obtained by replacing each nuisance in~(\ref{eq:continuous_simplified_target}) with its machine-learning estimate trained on an auxiliary sample:
\begin{equation}\label{eq:estimator_continuous_instrument}
    \hat{\underline{\theta}}_{\omega, 1} = \frac{1}{|I_2|} \sum_{i \in I_2} \Big( \hat{\psi}_{1}(O_i) + \hat{\psi}_{2}(O_i) + \hat{\psi}_{3}(O_i) \Big),
\end{equation}
where $\hat{\psi}_1, \hat{\psi}_2, \hat{\psi}_3$ correspond respectively to the boundary-quantile integral, the $g_1$ difference, and the trivial-bound term in~(\ref{eq:continuous_simplified_target}). \cref{thm:continuous_rate} below establishes its nonparametric convergence rate. The rest of this subsection develops the nuisance estimators $\hat{p}, \hat{g}_1$, and $\widehat{Q}_{Y,\underline{p}\mid X}$ that appear in~(\ref{eq:estimator_continuous_instrument}).

\paragraph{Estimation Procedure.} Because the continuous instrument setting requires nonparametric estimation of conditional expectations and quantiles, we employ a three-fold sample split: a propensity-estimation sample $I_0$, a nuisance-estimation sample $I_1$, and a main evaluation sample $I_2$. Splitting off $I_0$ separately is what allows the localized boundary-quantile subsample $\mathcal{I}_\delta = \{i\in I_1 : \hat p(Z_i,X_i) \leqslant \widehat{\underline p}(X_i) + \delta_n\}$ to remain conditionally i.i.d.\ given $I_0$. The conditional expectation $g_1(u,x)$ is estimated by a kernel-weighted local $M$-estimator in $u$ combined with a flexible ML fit in $x$, yielding a piecewise-linear-in-$u$ estimator that is deterministically Lipschitz in $u$---a property no joint ML regression on $(u, X)$ is known to deliver. This construction, related to local-polynomial DR-learning \citep{kennedy2023towards} and generalized random forests \citep{athey2019generalized}, decouples the covariate-dimension ML rate $r_{X,n}$ from the smoothing bias $h_n^2$ in $u$. \cref{alg:continuous_instrument} summarizes; the explicit kernel $M$-estimator, grid quantile regression, and target evaluation are in \cref{app:procedure_continuous_instrument}.

\begin{algorithm}[h]
\caption{Estimator: continuous instrument}\label{alg:continuous_instrument}
\begin{algorithmic}[1]
\STATE Split sample into $I_0, I_1, I_2$.
\STATE \textbf{On $I_0$:} estimate $\hat p(z,x)$, set $\hat q(z,x) = \phi(z,x,\hat p(z,x))$, and form boundary estimates $\widehat{\underline p}(x), \widehat{\overline p}(x)$.
\STATE \textbf{On $I_1$:} for each grid point $u_m = m/M$, train $\hat f_{u_m} = \arg\min_{f\in\mathcal F} \sum_{i\in I_1} K\!\left(\frac{u_m - \hat p(Z_i,X_i)}{h_n}\right)(Y_iW_i - f(X_i))^2$.
\STATE Define $\hat g_1(u,x)$ by linear interpolation between adjacent $\hat f_{u_m}(x)$.
\STATE Restrict to $\mathcal I_\delta = \{i\in I_1 : \hat p(Z_i,X_i) \leqslant \widehat{\underline p}(X_i) + \delta_n\}$ and estimate $\widehat Q_{Y,\underline p\mid X}$ by grid quantile regression on $\mathcal I_\delta$.
\STATE \textbf{On $I_2$:} compute $\hat\psi_1, \hat\psi_2, \hat\psi_3$ from~(\ref{eq:continuous_simplified_target}) and average to obtain $\hat{\underline\theta}_{\omega,1}$ as in~(\ref{eq:estimator_continuous_instrument}).
\end{algorithmic}
\end{algorithm}

To establish the formal convergence rate of our estimator in the continuous instrument setting, we must account for the estimation error of the generated regressor $\hat{p}$, the kernel-smoothed conditional expectation $\hat{g}_1$, and the localized boundary quantile $\widehat{Q}_{Y,\underline{p} \mid X}$. We impose the following regularity condition on the data generating process.

\begin{assumption}[Regularity]\label{asp:continuous_regularity}
    The marginal density of the propensity score $p(Z,X)$ is bounded away from zero and bounded above on its support. Furthermore, the density at the lower boundary $\underline{p}(X)$ is strictly positive almost surely.
\end{assumption}

\begin{assumption}[Nuisance Rates]\label{asp:continuous_rates}
    Recall that $d_X$ is the dimension of the covariates $X$. We assume the nuisance estimators satisfy the following conditions:
    \begin{enumerate}
        \item \textbf{Rate for Propensity Score:} The estimator $\hat{p}(z,x)$ converges to $p(z,x)$ at the nonparametric rate in the supremum norm, denoted by $\| \hat{p} - p \|_\infty = O_P(r_{p,n})$.
        \item \textbf{Regularity of $g_1$:} The conditional expectation function $g_1(u, x)$ is twice continuously differentiable with respect to $u$. Furthermore, its first and second derivatives, $\frac{\partial g_1(u,x)}{\partial u}$ and $\frac{\partial^2 g_1(u,x)}{\partial u^2}$, are bounded uniformly over all $u \in [0,1]$ and $x \in \mathcal{X}$.
        \item \textbf{Localized Kernel Bandwidth:} The 1D kernel function $K(\cdot)$ is continuously differentiable with bounded derivative. The bandwidth sequence is chosen such that $h_n \to 0, n h_n \to \infty$.
        \item \textbf{ML Rate and Realizability:} Let $\mathcal{F}_n$ be the machine learning hypothesis class used for a sample of size $n$. We assume $\mathcal{F}_n$ is sufficiently rich such that $g_1(u,\cdot) \in \mathcal{F}, \forall u \in[0,1]$
        and the statistical estimation error converges at a rate $r_{X,n}$ for an effective sample of size $n$.
        \item \textbf{$W_1$ rate for the conditional distribution of $Y$:} Given $m$ i.i.d.\ samples $\{(Y_i, X_i)\}_{i=1}^m$ on $[y_{\min}, y_{\max}] \times \mathcal{X}$, the conditional distribution estimator satisfies $\sup_{x \in \mathcal{X}} W_1\bigl(\widehat{F}_{Y \mid X=x},\, F_{Y \mid X=x}\bigr) = O_P(r_Q(m))$ for some $r_Q(m) \to 0$.
        \item \textbf{$W_1$-Lipschitz boundary:} Letting $F_{Y,\delta \mid x}$ denote the conditional distribution of $Y$ given $\{p(Z,X) \leqslant \underline{p}(x) + \delta, X = x\}$, this distribution is $W_1$-Lipschitz in $\delta$ at $\delta = 0$, uniformly in $x$: there exists $C_{\mathrm{Lip}} > 0$ such that $\sup_{x \in \mathcal{X}} W_1\!\bigl(F_{Y,\delta \mid x},\, F_{Y,\underline{p} \mid x}\bigr) \leqslant C_{\mathrm{Lip}}\,\delta$ for all sufficiently small $\delta > 0$.
    \end{enumerate}
\end{assumption}

Conditions~1--3 are standard regularity requirements: Condition~1 is satisfied by kernel, sieve, or $\ell_1$-penalized propensity estimators; Condition~2 ensures the piecewise-linear interpolation $\hat{g}_1$ incurs only an $h_n^2$ second-order bias; Condition~3 imposes standard kernel and bandwidth restrictions. Condition~4 is deliberately flexible: any ML procedure (random forests, neural networks, penalized linear methods) achieving kernel-weighted error rate $r_{X,n}$ is admissible. Conditions~5 and~6 ask only for conditional $W_1$-type rates---Condition~5 is strictly weaker than the uniform sup-norm rates established for standard conditional distribution and quantile regression estimators \citep{hall1999methods, guerre2012uniform, belloni2019conditional}, while Condition~6 controls the bias from localizing to observations within $\delta_n$ of $\underline{p}(x)$.

The following theorem formally establishes the nonparametric convergence rate of the lower bound estimator.

\begin{theorem}[Nonparametric Convergence Rate]\label{thm:continuous_rate}
    Suppose \cref{asp:structure}, \cref{asp:bounded_y}, \cref{asp:continuity_p}, \cref{asp:phi_smooth}, \cref{asp:continuous_regularity}, and \cref{asp:continuous_rates} hold, the grid resolution satisfies $M \gtrsim \sqrt{n h_n}$, the localization bandwidth satisfies $\delta_n \to 0$ with $n\delta_n \to \infty$, and $r_{p,n} = o(\delta_n)$. Then, the sample-split estimator $\hat{\underline{\theta}}_{\omega, 1}$ converges to the true lower bound $\underline{\theta}_{\omega, 1}$ at the following rate:
    \begin{equation*}
        |\hat{\underline{\theta}}_{\omega, 1} - \underline{\theta}_{\omega, 1}| = O_P\!\left(\frac{r_{p,n}}{h_n} + r_{X,n} + h_n^2 + r_Q(n\delta_n) + \delta_n\right).
    \end{equation*}
\end{theorem}

%% file: sections/experiments.tex
We validate our theoretical results through synthetic experiments and a real-data empirical application. All experiments compare our method (hereafter IVOT) bounds against the moment-relaxation bounds of \citet{magne2018using} (hereafter IVMTE), which represent the state-of-the-art alternative. Detailed numerical tables and additional results are deferred to \cref{app:simulation}.

\subsection{Synthetic Experiment}\label{sec:sim_synthetic}

We consider two synthetic settings, one with a continuous instrument and one with a discrete instrument, each using a distinct data-generating process (DGP). In both cases the target is $\theta_\alpha := \mathbb{E}[Y^{q_\alpha} - Y]$ under the policy $q_\alpha(Z) = \operatorname{clip}(p(Z)+\alpha, 0, 1)$ for $\alpha \in [-0.12, 0.12]$. For IVMTE, we specify both MTR functions $m_0,m_1$ as degree-$9$ $u$-splines on $[0,1]$ with nine interior knots at $\{0.1,0.2,\ldots,0.9\}$; this is a deliberately flexible, near-nonparametric sieve, chosen so that IVMTE is not disadvantaged by an overly restrictive parametric choice. Full DGP specifications are given in \cref{app:simulation}.

\paragraph{Continuous instrument.}
The continuous-instrument DGP uses a logistic propensity $p(Z) = \operatorname{logistic}(-1+2Z)$ with limited support $p(Z) \in [0.27, 0.73]$ and a \emph{constant} marginal treatment effect $\text{MTE}(u) = 0.5$, so the ground truth is $\theta_\alpha \approx 0.5\alpha$ for small $|\alpha|$. This setting deliberately isolates the identification difficulty: even though the MTE is constant, the limited propensity support prevents moment-relaxation methods from recovering tight bounds. \cref{fig:synthetic_bounds}(a) displays the identified sets at $n = 5{,}000$. IVOT yields tighter bounds than IVMTE; at $\alpha = 0.05$, the IVOT interval width is $0.011$ versus the IVMTE width of $0.096$, a roughly $8\times$ reduction. The IVMTE bounds are notably asymmetric---the upper bound extends substantially above the truth while the lower bound crosses zero---illustrating how discarding distributional information leads to spuriously wide identified sets. Across the 25 reported $\alpha$ grid points in this design, the true $\theta_\alpha$ falls inside the IVOT interval at 22 grid points. The three near-misses occur at $|\alpha| \leqslant 0.02$ where the IVOT interval is extremely tight (width ${\approx}\,0.001$) and finite-sample error in propensity estimation pushes the truth just outside the bounds. With $n = 10{,}000$, this pointwise inclusion check holds at all 25 grid points (\cref{app:simulation}), confirming this is a finite-sample phenomenon.

\paragraph{Discrete instrument.}
The discrete-instrument DGP uses $Z \sim \text{Bernoulli}(0.5)$ with two-point propensity score $\{0.25, 0.75\}$ and a heterogeneous $\text{MTE}(u) = 0.48 + 0.18u$, designed to test the closed-form bounds (\cref{thm:discrete_lower_bound}) when the propensity takes finitely many values. \cref{fig:synthetic_bounds}(b) shows the results at $n = 5{,}000$. IVOT again yields tighter bounds, and the true $\theta_\alpha$ lies inside the IVOT interval at all reported $\alpha$ grid points. At $\alpha = 0.05$, the IVOT interval is approximately $2.4\times$ narrower than IVMTE (width $0.041$ versus $0.097$), and the IVOT 95\% delta-method confidence interval $[-0.007,\; 0.046]$ is tighter than the IVMTE 95\% backward CI $[-0.050,\; 0.050]$. Because the propensity support consists of only two points, the MTE is identified only on $(0.25, 0.75)$; IVOT tightens bounds by exploiting the full distributional information in the identified region when bounding the contribution from the unidentified regions $(0, 0.25)$ and $(0.75, 1)$.

\paragraph{Discussion.}
The disparity between the two methods is consistent with the theoretical mechanism identified in \cref{sec:pid}: IVMTE reduces the observed conditional distributions to first-moment constraints, whereas IVOT retains the full quantile structure. The advantage is largest when the propensity score support is contained to $[0,1]$ (so that the unidentified region is large) and when the outcome distribution is informative about its quantile structure in the identified region. Numerical details for selected $\alpha$ values are tabulated in \cref{app:simulation}.

\begin{figure}[htbp]
    \centering
    \subfigure[Continuous instrument ($n = 5{,}000$)]{%
        \includegraphics[width=0.46\textwidth]{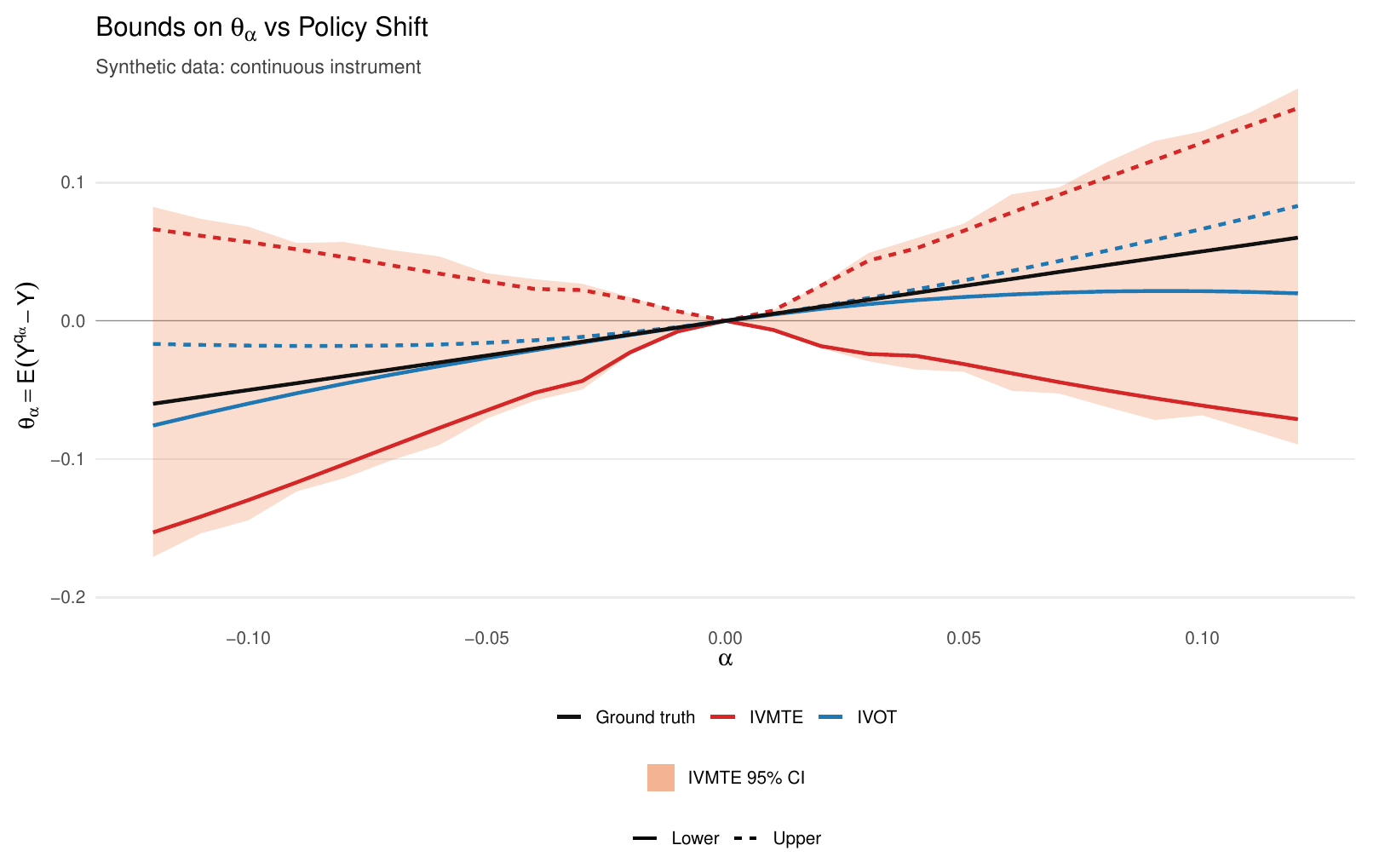}%
    }
    \quad
    \subfigure[Discrete instrument ($n = 5{,}000$)]{%
        \includegraphics[width=0.46\textwidth]{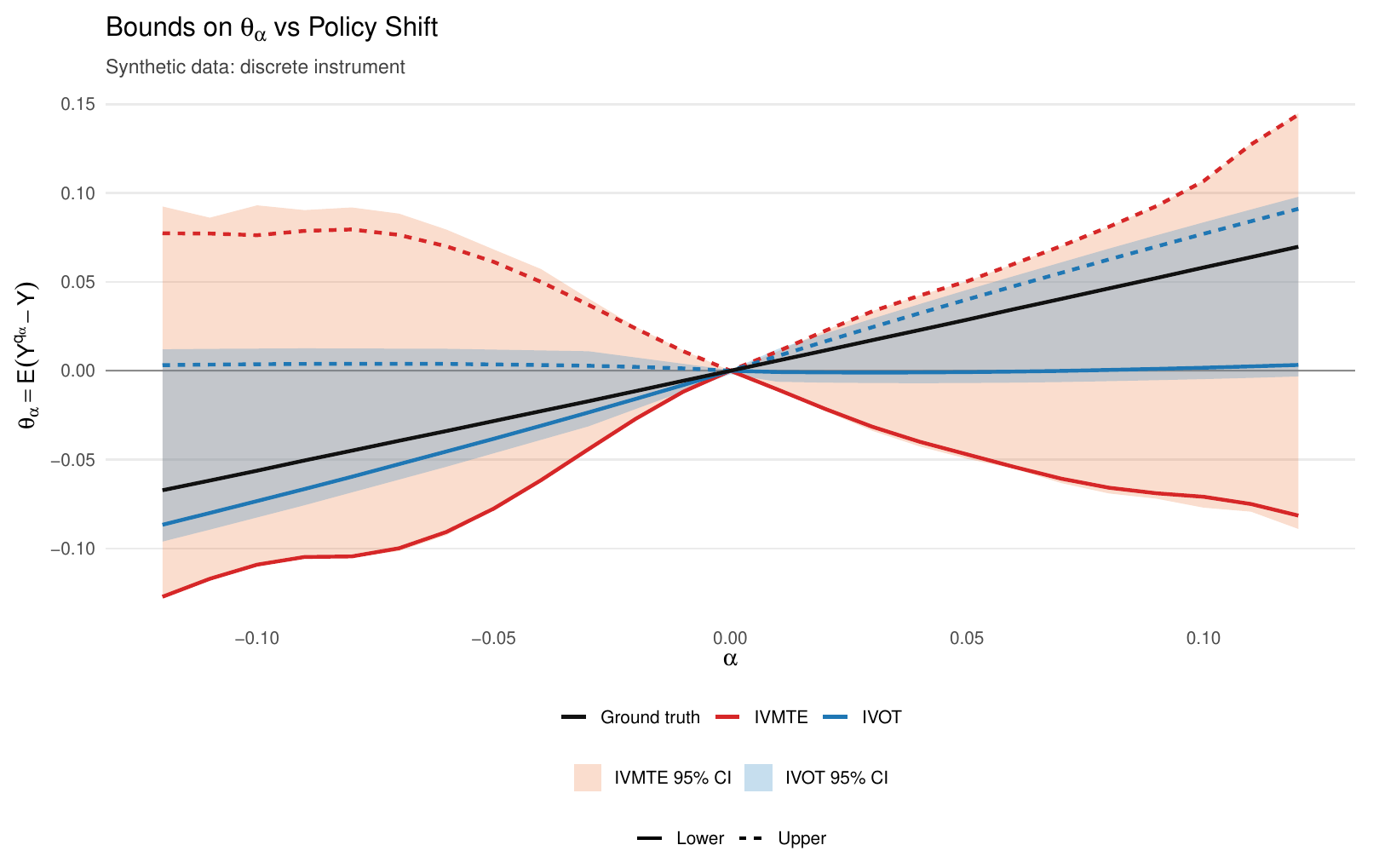}%
    }
    \caption{IVOT versus IVMTE identified sets for $\theta_\alpha = \mathbb{E}[Y^{q_\alpha}-Y]$ across policy shifts $\alpha \in [-0.12, 0.12]$. Solid line: ground truth. Shaded ribbons: pointwise 95\% confidence intervals for the estimated bound endpoints (orange = IVMTE, blue = IVOT). IVOT bounds are tighter in both settings.}
    \label{fig:synthetic_bounds}
\end{figure}

\subsection{Effect of Price Subsidies for Bed Nets}\label{sec:sim_bednet}

\paragraph{Background and setup.}
We apply our method to the bed net subsidy experiment of \citet{dupas2014short}. A longstanding debate in development economics concerns the optimal pricing strategy for health-protective goods: while high prices screen out low-valuation users, subsidies expand take-up but may attract individuals who ultimately do not use the product. Understanding the \emph{causal} effect of price reductions on product \emph{usage}---not just on take-up---is therefore essential for welfare analysis. The dataset records insecticide-treated bed net (ITN) transactions and follow-up usage for $n = 1{,}078$ Kenyan households; the instrument $Z$ is the offered price (17 distinct values spanning 0--250 KSh, randomized across distribution points), the treatment $W \in \{0,1\}$ indicates purchase, and the outcome $Y \in \{0,1\}$ indicates one-year follow-up usage. The baseline policy is the reference price $z_0 = 150$ KSh. For each $\alpha \in [0.05, 0.62]$, we study the counterfactual policy target that raises the purchase propensity at that baseline price from $\hat p(z_0)$ to $q_\alpha := \min(\hat{p}(z_0) + \alpha, 1)$; thus the denominator of $\text{PRTE}_\alpha = \mathbb{E}[Y^{q_\alpha} - Y]/\alpha$ is the known shift size $\alpha$. The maximum $\alpha_{\max}\approx 0.621$ equals the propensity at zero price minus the baseline propensity. Because the instrument is discrete, the closed-form bounds of \cref{thm:discrete_lower_bound} apply directly to the numerator, after which we divide by $\alpha$. For IVMTE, we report two sieve sizes (degree-$10$ and degree-$20$ $u$-splines on $[0,1]$) so that the shape-restriction trade-off is explicit: degree~$20$ enlarges the MTR class and imposes weaker shape restrictions, so its identified set should widen relative to degree~$10$. Full data, propensity-estimation, and policy-construction details are in \cref{app:simulation}.

\paragraph{Results.}
\cref{fig:bednet_bounds} plots the IVOT and IVMTE identified sets for $\text{PRTE}_\alpha$ across subsidy levels. IVOT consistently yields a tighter identified set than IVMTE across the reported range of $\alpha$ and under both spline specifications. At the smallest subsidy ($\alpha = 0.05$), the IVOT interval is $[0.740, 1.000]$ (width $0.260$), while the IVMTE interval is $[0.129, 0.990]$ (width $0.861$) at degree $10$ and widens to $[-0.372, 1.002]$ (width $1.374$) at degree $20$---roughly a $3.3\times$ and $5.3\times$ reduction in width relative to IVOT, respectively. The IVMTE bound visibly widens as the spline degree increases from $10$ to $20$, empirically confirming that the higher-degree sieve is more nonparametric and imposes weaker shape restrictions on the MTR. IVOT's advantage is therefore robust across these two IVMTE specifications.

Both methods indicate a positive policy effect across most of the subsidy range---the IVOT lower bound is positive throughout, and the degree-$10$ IVMTE lower bound is positive for all $\alpha$, with the more conservative degree-$20$ lower bound turning positive once $\alpha \geqslant 0.14$. Moreover, IVOT achieves near-point identification (width $< 0.05$) for the majority of $\alpha$ values, exhibiting a periodic pattern in which the bounds tighten to near-zero width before widening slightly as additional complier groups enter the policy margin: for instance, the identified set is essentially a point for $\alpha \in \{0.08, 0.09, 0.11\}$ and again for $\alpha \in \{0.23, 0.24, 0.26\}$. At the maximum feasible subsidy ($\alpha \approx 0.621$), IVOT point-identifies $\text{PRTE}_\alpha \approx 0.597$. Full numerical results, including IVOT 95\% delta-method confidence intervals, are reported in \cref{app:simulation}.

\begin{figure}[!htbp]
    \centering
    \includegraphics[width=0.58\textwidth]{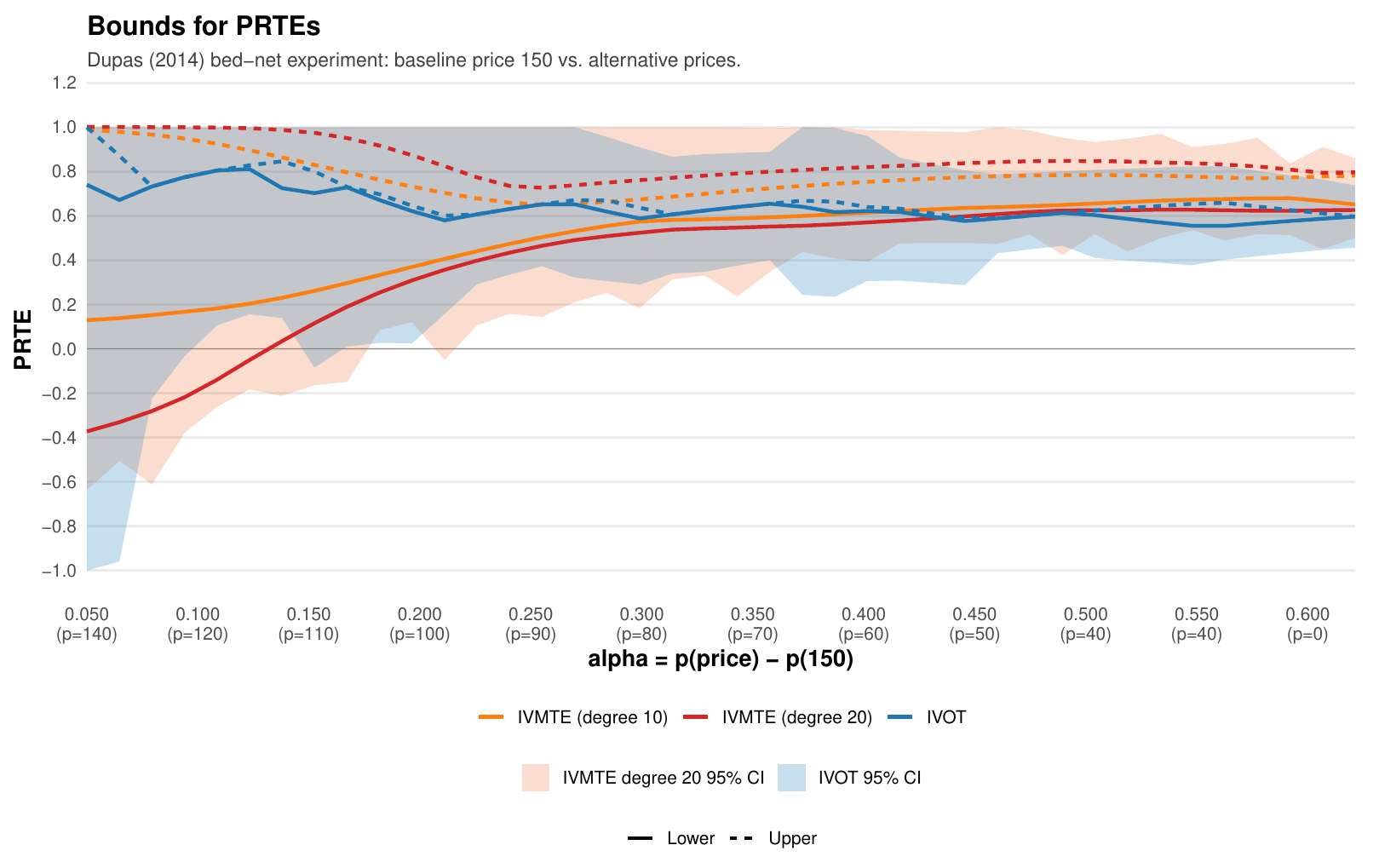}
    \caption{IVOT versus IVMTE identified sets for $\text{PRTE}_\alpha = \mathbb{E}[Y^{q_\alpha}-Y]/\alpha$ on the bed net data of \citet{dupas2014short}. IVMTE shown for degree-$10$ (orange) and degree-$20$ (red) $u$-spline sieves; degree~$20$ widens the bound, reflecting the more nonparametric MTR class. Shaded ribbons: pointwise 95\% confidence intervals for the estimated bound endpoints (orange = IVMTE deg.~$20$, blue = IVOT). IVOT substantially tightens the identified set and achieves near-point identification at several subsidy levels.}
    \label{fig:bednet_bounds}
\end{figure}

%% file: sections/conclusion.tex
We established a novel connection between the partial identification of PRTEs in IV models and OT. By formulating the problem as a CCOT problem over joint distributions compatible with the observed data, we showed that the multidimensional optimization analytically reduces to one-dimensional OT problems with product costs, yielding explicit closed-form expressions for the sharp bounds and bypassing the computationally intensive optimization required by existing approaches. The framework extends to general treatment settings (continuous and multi-valued), and we develop corresponding estimation results: Neyman-orthogonal DML scores deliver $\sqrt{n}$-consistency and asymptotic normality for discrete instruments, while for continuous instruments we characterize nonparametric convergence rates. A central insight, demonstrated theoretically and empirically, is that preserving full distributional information is both tractable and valuable for tight partial identification: moment-relaxation approaches that reduce the observed data to conditional expectations discard higher-order distributional features and can produce substantially wider bounds, whereas the CCOT formulation enforces compatibility with the entire observed distribution and makes the underlying optimization analytically transparent. We believe this insight extends beyond the IV setting studied here.

\paragraph{Acknowledgements. } Jose Blanchet gratefully acknowledges support from DoD through the grants Air Force Office of Scientific Research under award number FA9550-20-1-0397 and ONR 1398311, also support from NSF via grants 2229012, 2312204, 2403007 is gratefully acknowledged. 
Vasilis Syrgkanis gratefully acknowledges support from NSF Award IIS-2337916, an Amazon Research Award and a Google Research Award.

%% file: sections/notation.tex
\section{Notation}
{
\begin{longtable}{lp{12cm}}
\caption{Table of Important Variables and Mathematical Symbols}
\label{tab:symbols} \\

% --- HEADER FOR THE VERY FIRST PAGE ---
\hline
\textbf{Symbol} & \textbf{Description} \\
\hline
\endfirsthead

% --- HEADER FOR ALL SUBSEQUENT PAGES ---
\multicolumn{2}{c}%
{{\bfseries \tablename\ \thetable{} -- continued from previous page}} \\
\hline
\textbf{Symbol} & \textbf{Description} \\
\hline
\endhead

% --- FOOTER FOR ALL PAGES EXCEPT THE LAST ---
\hline \multicolumn{2}{r}{{Continued on next page}} \\
\endfoot

% --- FOOTER FOR THE VERY LAST PAGE ---
\hline
\endlastfoot

% --- TABLE CONTENTS ---
$W$ & Treatment variable (binary, discrete, or continuous), taking values in $\mathcal{W}$ \\
$Z$ & Instrumental variable, taking values in $\mathcal{Z}$ \\
$X$ & Covariates, taking values in $\mathcal{X}$ \\
$Y$ & Observed outcome variable, taking values in $\mathcal{Y}$ \\
$Y(w)$ & Potential outcome under counterfactual treatment $W=w$ \\
$U$ & Unobserved latent heterogeneity (resistance type), normalized to $U\mid X \sim \text{Unif}(0,1)$ \\
$p(Z,X)$ & Baseline conditional propensity score or threshold function, $\mathbb{P}(W=1\mid Z,X)$ \\
$q(Z,X)$ & Alternative policy conditional propensity score \\
$\omega(x,u)$ & Policy-specific weight function for the target parameter \\
$\theta_{\omega}$ & Target causal estimand (e.g., Policy-Relevant Treatment Effect) \\
$\underline{\theta}_{\omega,1}, \overline{\theta}_{\omega,1}$ & Sharp lower and upper bounds for the $W=1$ component of the target parameter \\
$y_{\min}, y_{\max}$ & Lower and upper logical/support bounds of the potential outcome \\
$\pi_w$ & Structural joint distribution (probability measure) of $(Y(w), U, X)$ \\
$\mathbb{P}_{\text{obs}}$ & The true observed joint distribution of the data $(Y, W, Z, X)$ \\
$p_k$ & $k$-th smallest unique value in the range of the propensity score $p(z)$ or $p(z,x)$ \\
$I_i, I_i(x)$ & A disjoint sub-interval $(p_i, p_{i+1}]$ of the latent variable space for $U$ \\
$\mu_{1,i}$ & Identified conditional probability measure of $Y(1)$ on the interval $I_i$ \\
$Q_{Y,i}(t)$ & Quantile function of the identified conditional distribution $\mu_{1,i}$ \\
$Q_{\omega,i}(t)$ & Quantile function of the random variable $\omega(U)$ where $U \sim \text{Unif}(p_i, p_{i+1})$ \\
$\underline{p}, \overline{p}$ & Minimum and maximum values of the propensity score in the connected continuous-IV setting \\
$\underline{p}_k(x), \overline{p}_k(x)$ & Endpoints of the $k$-th connected component of the conditional propensity-score range \\
$L_k(x), G_k(x)$ & $k$-th identified LIV interval and adjacent OT gap in the general continuous-IV decomposition \\
$\mu_{1,k\mid x}$ & Identified conditional probability measure of $Y(1)$ on the gap $G_k(x)$ \\
$q_{j,i}(x)$ & $j$-th smallest value of alternative policy $q(z,x)$ falling strictly inside $[p_i(x), p_{i+1}(x))$ \\
$S_k$ & Baseline instrument level set, $\{z \in \mathcal{Z} : p(z,x) = p_k(x)\}$ \\
$T_{j,k}$ & Alternative instrument level set, $\{z \in \mathcal{Z} : q(z,x) = q_{j,k}(x)\}$ \\
$c_{\text{gap}}$ & Minimal gap separating true propensity score levels \\
$\kappa_{j,i}(x)$ & Relative quantile position within the $i$-th baseline interval, $\frac{q_{j,i}(x)-p_i(x)}{p_{i+1}(x)-p_i(x)}$ \\
$\nu_{j,i}(x)$ & Conditional quantile function of the outcome evaluated at $\kappa_{j,i}(x)$ \\
$\gamma_{\text{full},i}(X), \gamma_{j,i}(X), \gamma_K(X)$ & Weighting probabilities for instruments spanning or terminating within specific intervals \\
$J_{\text{full},i}(X), J_{j,i}(X)$ & Conditional expectations/integrals of the outcome within specific quantile bounds \\
$\psi$ & Neyman-orthogonal score function used for Double Machine Learning estimation \\
$I_1, I_2$ & Auxiliary and main estimation samples used for sample splitting \\
$K(\cdot), h_n$ & 1D kernel function and bandwidth sequence used for localized machine learning \\
$\mathcal{F}_n$ & Machine learning hypothesis class \\
$r_{p,n}, r_{X,n}$ & Convergence rates for the estimated propensity score and the machine learning algorithm \\
$F_{W|Z,X}^{-1}$ & Conditional quantile function of the treatment $W$ (structural equation for continuous $W$) \\
$\mathcal{U}_{\text{id}}(x,w)$ & Identifiable region (union of supports of $U$) for the continuous/multi-valued treatment \\
$I_{x,w}(z)$ & Interval support of the latent variable for a multi-valued treatment \\
$J_{x,w}(z)$ & Isolated, disjoint sub-region of $I_{x,w}(z)$ derived from the DAG \\
$P_{1,k}(X)$ & Conditional expectation $\mathbb{E}_{\text{obs}}[YW \mid Z \in S_k, X]$ in the discrete outcome case; in the binary case, $J_{\text{full},k}(X) = P_{1,k+1}(X) - P_{1,k}(X)$ \\
$h_{j,k}^{-}(X),\, h_{j,k}^{+}(X)$ & Binary-specialized integrals for the lower and upper bounds in the discrete outcome case \\
\end{longtable}}

%% file: sections/appendix.tex
\section{Proofs in \cref{sec:pid}}

\subsection{Proof of Sharpness}

\begin{proof}[Proof of \cref{prop:sharpness}]
    \noindent\textbf{Direct direction.}
    Let $\tilde{\pi}$ be a probability measure on $\mathcal{Y}^2 \times [0,1] \times \mathcal{X}$ whose $(Y(w), U, X)$-marginals satisfy $\tilde{\pi}_w \in \Gamma(\mathbb{P}_{\text{obs}})$ for $w \in \{0, 1\}$. By the constraints in~\eqref{eq:ot_formulation}, each $\tilde{\pi}_w$ is a probability measure on $\mathcal{Y} \times [0,1] \times \mathcal{X}$ satisfying:
    \begin{enumerate}
        \item[(C1)] $\displaystyle \int_0^{p(z,x)} \tilde{\pi}_1(\mathrm{d}y \mid u, x) \, \mathrm{d}u = \mathbb{P}_{\text{obs}}(\mathrm{d}y, W\!=\!1 \mid Z\!=\!z, X\!=\!x)$, for $\mathbb{P}_{\text{obs}}$-a.e.\ $(z,x)$;
        \item[(C2)] $\displaystyle \int_{p(z,x)}^{1} \tilde{\pi}_0(\mathrm{d}y \mid u, x) \, \mathrm{d}u = \mathbb{P}_{\text{obs}}(\mathrm{d}y, W\!=\!0 \mid Z\!=\!z, X\!=\!x)$, for $\mathbb{P}_{\text{obs}}$-a.e.\ $(z,x)$;
        \item[(C3)] $\tilde{\pi}_w(\mathrm{d}u, \mathrm{d}x) = \mathrm{d}u \, \mathbb{P}_{\text{obs}}(\mathrm{d}x)$ for $w \in \{0,1\}$.
    \end{enumerate}
    Since $\mathcal{Y}^2 \subseteq \mathbb{R}^2$ is a Polish space and the $(U, X)$-marginal of $\tilde{\pi}$ is $\mathrm{d}u \, \mathbb{P}_{\text{obs}}(\mathrm{d}x)$ by (C3) applied to either marginal, $\tilde{\pi}$ admits a regular conditional kernel $\tilde{\pi}(\mathrm{d}y_0, \mathrm{d}y_1 \mid u, x)$ for $(\mathrm{d}u \otimes \mathbb{P}_{\text{obs}}(\mathrm{d}x))$-a.e.\ $(u, x)$, whose $Y(0)$- and $Y(1)$-marginals recover $\tilde{\pi}_0(\mathrm{d}y_0 \mid u, x)$ and $\tilde{\pi}_1(\mathrm{d}y_1 \mid u, x)$ respectively.

    \medskip
    \noindent\textbf{Construction.}
    We construct the joint distribution $\mathbb{P}$ of $(Y(0), Y(1), U, X, Z, W)$ on the product space $\mathcal{Y}^2 \times [0,1] \times \mathcal{X} \times \mathcal{Z} \times \{0,1\}$ via the following disintegration:
    \begin{equation}\label{eq:sharp_construct}
        \mathbb{P}(\mathrm{d}y_0, \mathrm{d}y_1, \mathrm{d}u, \mathrm{d}x, \mathrm{d}z, \{w\})
        = \tilde{\pi}(\mathrm{d}y_0, \mathrm{d}y_1 \mid u, x) \, \mathrm{d}u \, \mathbb{P}_{\text{obs}}(\mathrm{d}z \mid x) \, \mathbb{P}_{\text{obs}}(\mathrm{d}x) \, \delta_{\indicator(u \leqslant p(z,x))}(\{w\}),
    \end{equation}
    where $\delta_{\indicator(u \leqslant p(z,x))}(\{w\})$ is the Dirac mass placing all weight on $w = \indicator(u \leqslant p(z,x))$, enforcing the threshold-crossing selection rule deterministically. Crucially, $(Y(0), Y(1))$ are drawn jointly from $\tilde{\pi}(\cdot \mid u, x)$ rather than from an independent product, so the coupling prescribed by $\tilde{\pi}$ is preserved conditional on $(U, X)$; the independent coupling $\tilde{\pi}_0(\mathrm{d}y_0 \mid u, x) \, \tilde{\pi}_1(\mathrm{d}y_1 \mid u, x)$ is one feasible choice among many. The measure in~\eqref{eq:sharp_construct} is well-defined since each factor is a regular conditional probability measure.

    \medskip
    \noindent\textbf{Verification of (i): joint marginal distribution.}
    We show that the $(Y(0), Y(1), U, X)$-marginal of $\mathbb{P}$ equals $\tilde{\pi}$. Marginalizing~\eqref{eq:sharp_construct} over $(Z, W)$:
    \begin{align*}
        \mathbb{P}(\mathrm{d}y_0, \mathrm{d}y_1, \mathrm{d}u, \mathrm{d}x)
        &= \tilde{\pi}(\mathrm{d}y_0, \mathrm{d}y_1 \mid u, x) \, \mathrm{d}u \, \mathbb{P}_{\text{obs}}(\mathrm{d}x) \underbrace{\sum_{w \in \{0, 1\}} \int_{\mathcal{Z}} \delta_{\indicator(u \leqslant p(z,x))}(\{w\}) \, \mathbb{P}_{\text{obs}}(\mathrm{d}z \mid x)}_{=\, 1} \\
        &= \tilde{\pi}(\mathrm{d}y_0, \mathrm{d}y_1 \mid u, x) \, \mathrm{d}u \, \mathbb{P}_{\text{obs}}(\mathrm{d}x).
    \end{align*}
    By (C3) the $(U, X)$-marginal of $\tilde{\pi}$ is $\mathrm{d}u \, \mathbb{P}_{\text{obs}}(\mathrm{d}x)$, so the right-hand side equals $\tilde{\pi}(\mathrm{d}y_0, \mathrm{d}y_1, \mathrm{d}u, \mathrm{d}x)$, verifying (i).

    \medskip
    \noindent\textbf{Verification of (ii): structural assumptions.}
    We verify each part of \cref{asp:structure}.
    \begin{itemize}
        \item \emph{Consistency.} Consistency holds trivially by setting the observed outcome to be $Y = W Y(1) + (1-W) Y(0)$.
        \item \emph{Conditional instrumental exogeneity.} By~\eqref{eq:sharp_construct}, conditional on $X = x$, the variables $(Y(0), Y(1), U)$ are drawn from $\tilde{\pi}(\mathrm{d}y_0, \mathrm{d}y_1 \mid u, x) \, \mathrm{d}u$ and $Z$ is drawn independently from $\mathbb{P}_{\text{obs}}(\mathrm{d}z \mid x)$. These appear as independent factors in~\eqref{eq:sharp_construct}, so $Z \perp (Y(0), Y(1), U) \mid X$ under $\mathbb{P}$.
        \item \emph{Threshold crossing.} The selection equation $W = \indicator(U \leqslant p(Z,X))$ holds $\mathbb{P}$-a.s.\ by the Dirac mass in~\eqref{eq:sharp_construct}. Moreover, (C3) gives $U \mid X \sim \mathrm{Unif}(0,1)$, which is a continuous distribution.
    \end{itemize}

    \medskip
    \noindent\textbf{Verification of (iii): observational equivalence.}
    We must show that the distribution of the observables $(Y, Z, X, W)$ induced by $\mathbb{P}$ coincides with $\mathbb{P}_{\text{obs}}$. It suffices to verify the conditional distribution of $(Y, W)$ given $(Z, X)$, since the marginal of $(Z, X)$ under $\mathbb{P}$ equals $\mathbb{P}_{\text{obs}}(Z, X)$ by construction.

    Fix $(z, x)$ and consider the event $W = 1$. Under $\mathbb{P}$, marginalizing~\eqref{eq:sharp_construct} over $(Y(0), Y(1))$ reduces the joint kernel to its $Y(1)$-marginal $\tilde{\pi}_1(\cdot \mid u, x)$; using $W = 1 \Leftrightarrow U \leqslant p(z,x)$ and $Y = Y(1)$ when $W = 1$,
    \begin{align*}
        \mathbb{P}(Y \in A, \, W = 1 \mid Z = z, X = x)
        &= \int_0^{p(z,x)} \tilde{\pi}_1(A \mid u, x) \, \mathrm{d}u.
    \end{align*}
    By constraint (C1), this equals $\mathbb{P}_{\text{obs}}(Y \in A, \, W = 1 \mid Z = z, X = x)$ for every measurable $A \subseteq \mathcal{Y}$. Similarly, for $W = 0$ we have $Y = Y(0)$ and $U > p(z, x)$, so marginalizing over $(Y(0), Y(1))$ leaves the $Y(0)$-marginal $\tilde{\pi}_0$ and
    \begin{align*}
        \mathbb{P}(Y \in A, \, W = 0 \mid Z = z, X = x)
        &= \int_{p(z,x)}^{1} \tilde{\pi}_0(A \mid u, x) \, \mathrm{d}u,
    \end{align*}
    which equals $\mathbb{P}_{\text{obs}}(Y \in A, \, W = 0 \mid Z = z, X = x)$ by constraint (C2). Thus the conditional laws of $(Y, W) \mid (Z, X)$ under $\mathbb{P}$ and $\mathbb{P}_{\text{obs}}$ agree for both values of $W$, completing the verification of (iii).

    \medskip
    \noindent\textbf{Converse direction.}
    Suppose $(Y(0), Y(1), U, X, Z, W)$ is a joint distribution satisfying (ii)--(iii), and let $\tilde{\pi}$ denote its $(Y(0), Y(1), U, X)$-marginal. Fix $w \in \{0, 1\}$. The $(U, X)$-marginal of $\tilde{\pi}_w$ equals the $(U, X)$-marginal of $\tilde{\pi}$; by \cref{asp:structure} (threshold crossing), $U \mid X \sim \mathrm{Unif}(0, 1)$ and $X$ has marginal $\mathbb{P}_{\text{obs}}(\mathrm{d}x)$ under (iii), so $\tilde{\pi}_w(\mathrm{d}u, \mathrm{d}x) = \mathrm{d}u \, \mathbb{P}_{\text{obs}}(\mathrm{d}x)$, establishing (C3). For (C1), using consistency and the selection rule $W = \indicator(U \leqslant p(Z, X))$ together with conditional instrumental exogeneity,
    \begin{align*}
        \mathbb{P}_{\text{obs}}(Y \in A, \, W = 1 \mid Z = z, X = x)
        &= \mathbb{P}(Y(1) \in A, \, U \leqslant p(z, x) \mid Z = z, X = x) \\
        &= \int_0^{p(z, x)} \tilde{\pi}_1(A \mid u, x) \, \mathrm{d}u,
    \end{align*}
    where the second equality uses $Z \perp (Y(1), U) \mid X$ and $U \mid X \sim \mathrm{Unif}(0, 1)$. Property (iii) identifies the left-hand side with the observational conditional distribution, yielding (C1). The argument for (C2) is symmetric. Therefore $\tilde{\pi}_w \in \Gamma(\mathbb{P}_{\text{obs}})$ for each $w$, completing the converse.
\end{proof}

\subsection{Proof of \cref{thm:covariate_continuous_lower_bound}} \label{sec:cov_proof}

\begin{proof}[Proof of \cref{thm:covariate_continuous_lower_bound}]
We prove the sharp lower bound; the upper bound follows by replacing the countermonotonic coupling with the comonotonic one and swapping $y_{\min}$ and $y_{\max}$ in the trivial bound term.

\medskip
\noindent\textbf{Disintegration over $X$.}
Recall that the $W=1$ side minimization problem~\eqref{eq:ot_formulation_1} reads $\min_{\pi_1 \in \Gamma_1(\mathbb{P}_{\text{obs}})} \mathbb{E}_{\pi_1}[Y(1)\,\omega(X,U)]$, where $\Gamma_1(\mathbb{P}_{\text{obs}})$ is the set of probability measures on $\mathcal{Y} \times [0,1] \times \mathcal{X}$ satisfying
\begin{align}
    &\int_0^{p(z,x)} \pi_1(\mathrm{d}y \mid u, x)\,\mathrm{d}u = \mathbb{P}_{\text{obs}}(\mathrm{d}y, W\!=\!1 \mid p(Z,X)\!=\!p(z,x), X\!=\!x),\ \mathbb{P}_{\text{obs}}\text{-a.e.\ } (z,x), \label{eq:pf_cts_constraint1} \\
    &\pi_1(\mathrm{d}u, \mathrm{d}x) = \mathrm{d}u\;\mathbb{P}_{\text{obs}}(\mathrm{d}x). \label{eq:pf_cts_constraint2}
\end{align}
The observational constraint~\eqref{eq:pf_cts_constraint1} couples $\pi_1(\mathrm{d}y \mid u, x)$ across different $u$ only through the integral $\int_0^{p(z,x)} \pi_1(\mathrm{d}y \mid u, x)\,\mathrm{d}u$, which for fixed $x$ involves only the kernel values at that same $x$. Hence the constraints decouple across $x$. Disintegrating $\pi_1$ with respect to the $(U,X)$-marginal~\eqref{eq:pf_cts_constraint2} and using the compactness of $\mathcal{Y}$ together with a measurable selection argument \citep[Proposition 7.50]{bertsekas1996stochastic}, we may exchange the minimization with the outer expectation:
\begin{equation}\label{eq:pf_cts_exchange}
    \min_{\pi_1 \in \Gamma_1(\mathbb{P}_{\text{obs}})} \mathbb{E}_{\pi_1}[Y(1)\,\omega(X,U)] = \mathbb{E}_{X}\left[\min_{\pi_{1,X} \in \Gamma_{1,X}} \int_0^1 \left(\int_{\mathcal{Y}} y\;\pi_{1,X}(\mathrm{d}y \mid u)\right) \omega(X,u)\,\mathrm{d}u\right],
\end{equation}
where $\Gamma_{1,X}$ is the set of probability kernels $\pi_{1,X}(\mathrm{d}y \mid u)$ on $\mathcal{Y} \times [0,1]$ satisfying~\eqref{eq:pf_cts_constraint1} with $x = X$ held fixed.

\medskip
\noindent\textbf{Maximal-interval decomposition for fixed $x$.}
Fix $x \in \mathcal{X}$. Under \cref{asp:continuity_p}, the map $z \mapsto p(z,x)$ is continuous on each connected component of $\mathcal{Z}$, so the image of each component is a closed interval (possibly degenerate). Merging overlaps and ordering from left to right, the conditional range admits the maximal-interval decomposition~\eqref{eq:range_decomposition_cov},
\[
    p(\mathcal{Z}, x) \;=\; \bigcup_{k=1}^{K_x}[\underline{p}_k(x), \overline{p}_k(x)], \qquad 0 \leqslant \underline{p}_1(x) \leqslant \overline{p}_1(x) < \cdots < \underline{p}_{K_x}(x) \leqslant \overline{p}_{K_x}(x) \leqslant 1,
\]
with boundary conventions $\overline{p}_0(x) := 0$ and $\underline{p}_{K_x+1}(x) := 1$. Adopting the notation of \cref{subsubsec:continuous_iv}, this partitions $[0,1]$ into the conditional LIV intervals $L_k(x) = [\underline{p}_k(x), \overline{p}_k(x)]$ ($k=1,\dots,K_x$), the OT gaps $G_k(x) = (\overline{p}_k(x), \underline{p}_{k+1}(x))$ ($k=0,\dots,K_x-1$), and the trivial tail $(\overline{p}_{K_x}(x), 1]$. The observational constraint at the $x$-slice reads: for every $p \in p(\mathcal{Z}, x)$,
\begin{equation}\label{eq:pf_cts_obs}
    \int_0^{p} \pi_{1,x}(\mathrm{d}y \mid u)\,\mathrm{d}u = \mathbb{P}_{\text{obs}}(\mathrm{d}y,\, W\!=\!1 \mid p(Z,X) = p,\, X = x),
\end{equation}
with the convention $\mathbb{P}_{\text{obs}}(\mathrm{d}y, W\!=\!1 \mid p(Z,X)=0, X=x) \equiv 0$.

\medskip
\noindent\textbf{The OT gaps $G_k(x)$.}
For each $k = 0,\dots,K_x-1$, differencing~\eqref{eq:pf_cts_obs} at $p = \underline{p}_{k+1}(x)$ and at $p = \overline{p}_k(x)$ isolates the integral over $G_k(x) = (\overline{p}_k(x), \underline{p}_{k+1}(x))$:
\begin{equation}\label{eq:pf_cts_gap_constraint}
    \int_{G_k(x)} \pi_{1,x}(\mathrm{d}y \mid u)\,\mathrm{d}u = (\underline{p}_{k+1}(x)-\overline{p}_k(x))\,\mu_{1,k\mid x}(\mathrm{d}y),
\end{equation}
where $\mu_{1,k\mid x}$ is the conditional gap-identified measure in~\eqref{eq:gap_identified_measure_cov}. Dividing by $|G_k(x)| = \underline{p}_{k+1}(x)-\overline{p}_k(x) > 0$, the marginal distribution of $Y(1)$ conditional on $U \in G_k(x)$ is fixed to $\mu_{1,k\mid x}$, while the $U$-marginal is $\nu_{k\mid x} := \mathrm{Unif}(G_k(x))$. The joint distribution of $(Y(1), U)$ over $G_k(x)$ is not otherwise pinned down, so the contribution of this gap to the inner minimization is a standard 1D optimal transport problem:
\begin{equation}\label{eq:pf_cts_ot_sub}
    \min_{\gamma_k \in \Pi(\mu_{1,k\mid x},\, \nu_{k\mid x})} \int_{\mathcal{Y} \times G_k(x)} y \cdot \omega(x, u)\,\mathrm{d}\gamma_k(y, u).
\end{equation}

\medskip
\noindent\textbf{The LIV intervals $L_k(x)$.}
Fix a non-degenerate $L_k(x)$ (with $\underline{p}_k(x) < \overline{p}_k(x)$). On the continuum $L_k(x)$ the constraint~\eqref{eq:pf_cts_obs} holds for every $p \in L_k(x)$; subtracting the constraint at $p$ from that at $p + \mathrm{d}p$ and passing to the limit yields, for Lebesgue-a.e.\ $u \in L_k(x)$,
\begin{equation}\label{eq:pf_cts_deriv}
    \pi_{1,x}(\mathrm{d}y \mid u) = \frac{\partial}{\partial p}\bigg|_{p=u} \mathbb{P}_{\text{obs}}(\mathrm{d}y,\, W\!=\!1 \mid p(Z,X) = p,\, X = x).
\end{equation}
This is the LIV identification result \citep{heckman1999local}: the conditional distribution of $Y(1)$ given $U = u$ and $X = x$ is point-identified for $u \in L_k(x)$ a.e. Every feasible kernel must agree on $L_k(x)$, so the contribution to the objective is uniquely determined:
\begin{equation}\label{eq:pf_cts_liv}
    \int_{\underline{p}_k(x)}^{\overline{p}_k(x)} \left(\int_{\mathcal{Y}} y\;\pi_{1,x}(\mathrm{d}y \mid u)\right)\omega(x, u)\,\mathrm{d}u
    = \int_{\underline{p}_k(x)}^{\overline{p}_k(x)} \left(\frac{\partial}{\partial u}\,\mathbb{E}_{\text{obs}}[Y\,W \mid p(Z,X) = u,\, X = x]\right)\omega(x, u)\,\mathrm{d}u.
\end{equation}
On each non-degenerate LIV interval, the constraint~\eqref{eq:pf_cts_obs} has the integral form
\[
p \mapsto \mathbb{P}_{\text{obs}}(\mathrm{d}y, W\!=\!1 \mid p(Z,X) = p, X = x)
= \int_{\underline{p}_k(x)}^{p} \pi_{1,x}(\mathrm{d}y \mid u)\,\mathrm{d}u + C_k(\mathrm{d}y),
\]
for a measure-valued constant $C_k(\mathrm{d}y)$ depending only on the left endpoint of the interval. Hence this map is absolutely continuous on $L_k(x)$ as an indefinite integral of the kernel $u \mapsto \pi_{1,x}(\mathrm{d}y \mid u)$, so the Radon--Nikodym derivative in~\eqref{eq:pf_cts_deriv} exists for Lebesgue-a.e.\ $u$. Moreover, $|y\,\omega(x,u)| \leqslant y_{\max}\|\omega\|_\infty$ by compactness of $\mathcal{Y}$ and boundedness of $\omega$, so the integral in~\eqref{eq:pf_cts_liv} is well-defined. Degenerate $L_k(x)$ ($\underline{p}_k(x) = \overline{p}_k(x)$) contribute a Lebesgue-null integral and drop out.

\medskip
\noindent\textbf{Compatibility of the region-wise constraints.}
We verify that the region-wise constraints~\eqref{eq:pf_cts_gap_constraint} and~\eqref{eq:pf_cts_deriv} are jointly equivalent to the original constraint family~\eqref{eq:pf_cts_obs}. Fix $p \in p(\mathcal{Z}, x)$ and let $k^\star$ be the unique index with $p \in L_{k^\star}(x)$. Then
\[
    \int_0^p \pi_{1,x}(\mathrm{d}y\mid u)\,\mathrm{d}u \;=\; \sum_{k=0}^{k^\star-1}\underbrace{\int_{G_k(x)}\pi_{1,x}(\mathrm{d}y\mid u)\,\mathrm{d}u}_{\text{fixed by }\eqref{eq:pf_cts_gap_constraint}} \;+\; \sum_{k=1}^{k^\star-1}\underbrace{\int_{L_k(x)}\pi_{1,x}(\mathrm{d}y\mid u)\,\mathrm{d}u}_{\text{determined by }\eqref{eq:pf_cts_deriv}} \;+\; \underbrace{\int_{\underline{p}_{k^\star}(x)}^p\pi_{1,x}(\mathrm{d}y\mid u)\,\mathrm{d}u}_{\text{determined by }\eqref{eq:pf_cts_deriv}}.
\]
A direct computation using the telescoping definition of $\mu_{1,k\mid x}$ in~\eqref{eq:gap_identified_measure_cov} and the fundamental theorem of calculus applied to~\eqref{eq:pf_cts_deriv} shows that this sum equals $\mathbb{P}_{\text{obs}}(\mathrm{d}y, W\!=\!1 \mid p(Z,X)=p, X=x)$, so~\eqref{eq:pf_cts_obs} is automatically satisfied. Conversely, the gap constraints~\eqref{eq:pf_cts_gap_constraint} and the LIV identity~\eqref{eq:pf_cts_deriv} are consequences of~\eqref{eq:pf_cts_obs} by construction. Hence the constraints on the gaps, LIV intervals, and trivial tail are decoupled, and the inner minimization decomposes across regions.

\medskip
\noindent\textbf{The trivial tail $(\overline{p}_{K_x}(x), 1]$.}
For $u > \overline{p}_{K_x}(x)$, no value of the instrument produces a propensity score exceeding $\overline{p}_{K_x}(x)$, so $Y(1)$ is never observed. The constraint family~\eqref{eq:pf_cts_obs} imposes no restriction on $\pi_{1,x}(\mathrm{d}y \mid u)$ there. By \cref{asp:bounded_y}, the pointwise minimum of $\int_{\mathcal{Y}} y\;\pi_{1,x}(\mathrm{d}y \mid u) \cdot \omega(x,u)$ is attained by $\pi_{1,x}(\cdot \mid u) = \delta_{y_{\min}}$ when $\omega(x,u) \geqslant 0$ and by $\delta_{y_{\max}}$ when $\omega(x,u) < 0$:
\begin{equation}\label{eq:pf_cts_trivial}
    \min \int_{\overline{p}_{K_x}(x)}^{1} \left(\int_{\mathcal{Y}} y\;\pi_{1,x}(\mathrm{d}y \mid u)\right)\omega(x, u)\,\mathrm{d}u = \int_{\overline{p}_{K_x}(x)}^{1} \Big( y_{\min}\max\{0,\omega(x,u)\} + y_{\max}\min\{0,\omega(x,u)\} \Big)\mathrm{d}u.
\end{equation}

\medskip
\noindent\textbf{Solving each OT sub-problem.}
For each $k = 0,\dots,K_x-1$, the sub-problem~\eqref{eq:pf_cts_ot_sub} is a one-dimensional optimal transport problem with product cost $c(y,u) = y \cdot \omega(x,u)$ between the marginals $\mu_{1,k\mid x}$ and $\nu_{k\mid x} = \mathrm{Unif}(G_k(x))$. Invoking \cref{thm:1d_ot} (countermonotonic coupling for the minimum):
\begin{equation}\label{eq:pf_cts_ot_sol}
    \min_{\gamma_k \in \Pi(\mu_{1,k\mid x},\, \nu_{k\mid x})} \int y \cdot \omega(x,u)\,\mathrm{d}\gamma_k = (\underline{p}_{k+1}(x)-\overline{p}_k(x))\int_0^1 Q_{Y,k\mid x}(t)\,Q_{\omega,k\mid x}(1-t)\,\mathrm{d}t,
\end{equation}
where $Q_{Y,k\mid x}$ is the quantile function of $\mu_{1,k\mid x}$ and $Q_{\omega,k\mid x}$ is the quantile function of $\omega(x, U)$ for $U \sim \mathrm{Unif}(G_k(x))$. The prefactor $|G_k(x)| = \underline{p}_{k+1}(x)-\overline{p}_k(x)$ arises from the Jacobian of the rescaling.

\medskip
\noindent\textbf{Aggregation over $X$.}
Summing~\eqref{eq:pf_cts_ot_sol} over $k = 0,\dots,K_x-1$, adding the LIV contribution~\eqref{eq:pf_cts_liv} summed over $k = 1,\dots,K_x$, and adding the trivial tail~\eqref{eq:pf_cts_trivial} into~\eqref{eq:pf_cts_exchange}, the global minimum is
\begin{align*}
    \underline{\theta}_{\omega,1} = \mathbb{E}_{X}\Bigg[\;&\sum_{k=0}^{K_X-1}(\underline{p}_{k+1}(X)-\overline{p}_k(X))\int_0^1 Q_{Y,k\mid X}(t)\,Q_{\omega,k\mid X}(1-t)\,\mathrm{d}t \\
    &+ \sum_{k=1}^{K_X}\int_{\underline{p}_k(X)}^{\overline{p}_k(X)} \left(\frac{\partial}{\partial u}\,\mathbb{E}_{\text{obs}}[Y\,W \mid p(Z,X) = u,\, X]\right)\omega(X, u)\,\mathrm{d}u \\
    &+ \int_{\overline{p}_{K_X}(X)}^{1}\Big(y_{\min}\max\{0,\omega(X,u)\} + y_{\max}\min\{0,\omega(X,u)\}\Big)\mathrm{d}u\Bigg].
\end{align*}
The outer expectation is well-defined: each term inside the brackets is bounded uniformly in $x$ by $y_{\max}\|\omega\|_\infty$ (using compactness of $\mathcal{Y}$ and boundedness of $\omega$), so integrability with respect to $\mathbb{P}_{\text{obs}}(\mathrm{d}x)$ is immediate.

\medskip
\noindent\textbf{Sharpness.}
The bound is sharp because: (i) by \cref{prop:sharpness}, every feasible pair $(\pi_0, \pi_1) \in \Gamma(\mathbb{P}_{\text{obs}})$ corresponds to a valid data-generating process satisfying \cref{asp:structure}, so the infimum of $\theta_\omega$ over $\Gamma(\mathbb{P}_{\text{obs}})$ equals the infimum over all observationally equivalent structural models; and (ii) the infimum in each conditional gap sub-problem~\eqref{eq:pf_cts_ot_sub} is attained by the countermonotonic coupling (which exists since $\mu_{1,k\mid x}$ and $\nu_{k\mid x}$ are Borel probability measures on $\mathbb{R}$), the contribution on each LIV interval~\eqref{eq:pf_cts_liv} is uniquely determined by every feasible kernel, and the infimum in the trivial tail~\eqref{eq:pf_cts_trivial} is attained by the Dirac mass at $y_{\min}$ where $\omega(x,u) \geqslant 0$ and at $y_{\max}$ where $\omega(x,u) < 0$. By a measurable selection argument, these conditional optimizers can be assembled into a globally feasible $\pi_1^\star \in \Gamma_1(\mathbb{P}_{\text{obs}})$ that attains the global infimum.
\end{proof}

\subsection{Illustrative Example for \cref{thm:discrete_lower_bound}}

\begin{example}\label{ex:discrete_bd}
    Consider a binary instrument with $\mathbb{P}(Z=1)=1/2$ and a constant propensity score $p(0) = p(1) = 1/2$. Let the target weight function be $\omega(u) = \mathbb{E}[\indicator(u \in (0, q_1(Z)))]$, where $q_1(0) = 1/4$ and $q_1(1) = 3/4$.

    Simple calculation gives the explicit piecewise form of the weight function:
    \begin{align*}
        \omega(u) =
        \begin{cases}
            1, & \text{if } u \in (0, 1/4] \\
            1/2, & \text{if } u \in (1/4, 3/4] \\
            0, & \text{if } u \in (3/4, 1]
        \end{cases}
    \end{align*}
    Because the propensity score is constant $1/2$, the data identifies a single constrained interval $I_0 = (0, 1/2]$ where $Y(1)$ is observed (corresponding to $W=1$), and an unconstrained interval $I_1 = (1/2, 1]$ where $Y(1)$ is entirely unobserved.

    Plugging this into our closed-form optimal transport solution, the sharp lower bound elegantly decomposes into three distinct economic components. On $I_0$, the optimal coupling pairs the highest weights in $\omega(u)$ with the lowest quantiles of the observed outcome distribution $Y \mid W=1$. For the unconstrained interval $I_1$, because $\omega(u) \geqslant 0$ here, the optimal solution trivially assigns the logical minimum $y_{\min}$.

    Evaluating the integral exactly yields:
    \begin{align}\label{eq:example_bound}
        \underline{\theta}_{\omega, 1}
        &= \int_0^{1/2} Q_{Y \mid W=1}(2u) \omega(u) \, \mathrm{d}u + \int_{1/2}^1 y_{\min} \omega(u) \, \mathrm{d}u \nonumber \\
        &= \underbrace{\frac{1}{4} \mathbb{E}[Y \mid W=1]}_{\text{Identifiable Mean}} \ + \ \underbrace{\frac{1}{8} \text{CVaR}_{0.5}(Y \mid W=1)}_{\text{Optimal Transport (CVaR)}} \ + \ \underbrace{\frac{1}{8} y_{\min}}_{\text{Trivial Lower Bound}}
    \end{align}
    where $\text{CVaR}_{\alpha}(Y) = \frac{1}{\alpha} \int_0^\alpha Q_Y(s) \, \mathrm{d}s$ is the Conditional Value at Risk, representing the mathematical expectation of $Y$ conditional on it falling within its bottom $\alpha$-quantile.

    This explicit decomposition reveals exactly why moment relaxation yields unnecessarily loose bounds. Moment relaxation only enforces the first moment over the entire observed interval: $\mathbb{E}[m_1(U) \mid U \in (0, 1/2)] = \mathbb{E}[Y \mid W=1]$.

    To minimize the objective, their linear program evaluates:
    \begin{align*}
        \min_{m_1} \left( \frac{1}{4} \mathbb{E}[Y \mid W=1] + \frac{1}{2} \int_0^{1/4} m_1(u) \, \mathrm{d}u \right) + \frac{1}{8} y_{\min}
    \end{align*}
    subject to the above-mentioned moment constraints. Lacking the strict distributional constraint that bounds the tail behavior, the MTR optimal solution pushes the unconstrained function $m_1(u)$ to its worst-case lower bound, replacing the true exact tail expectation $\text{CVaR}_{0.5}(Y \mid W=1)$ with the global constant $y_{\min}$. Because $\text{CVaR}_{0.5}(Y \mid W=1) \geqslant y_{\min}$, the moment relaxation yields extremely conservative lower bounds. Our CCOT formulation explicitly prevents this by leveraging all distributional information in the observed data.

    We further remark that a decomposition similar to (\ref{eq:example_bound}) exists in general (see \cref{sec:est_inf} for details), and \cref{fig:decomposition} illustrates the three identification regions that compose the bound.
\end{example}

\subsection{Comparison with Moment Relaxation Approach}

\begin{example}[Comparison with Moment Relaxation Approach] \label{exp:loose_bd}
    Suppose there are no covariates and the weight function is $\omega(u) = \indicator(u \in (1/4, 1/2))$. Our target parameter simplifies to:
    \begin{equation*}
        \mathbb{E}[(Y(1) - Y(0)) \indicator(U \in (1/4, 1/2))].
    \end{equation*}
    Consider a data-generating process with a binary instrument $Z \in \{0,1\}$, propensity scores $p(0) = 1/4$ and $p(1) = 3/4$, and the observed outcome is constantly zero, $Y(w) = 0$ almost surely for $w \in \{0,1\}$.

    Under our CCOT formulation, this parameter is point-identified and its value is exactly $0$. Because the conditional distribution of the observed outcome is a Dirac measure at zero, $\mathbb{P}(Y \in \cdot \mid W=1, Z=z) = \delta_0(\cdot)$ for each $z$, the observational constraint forces $\pi_{1}(\cdot \mid u) = \delta_0$ for almost all $u \in (0, 3/4)$. Similarly, $\mathbb{P}(Y \in \cdot \mid W=0, Z=z) = \delta_0(\cdot)$ forces $\pi_0(\cdot \mid u) = \delta_0$ for almost all $u \in (1/4, 1)$. Since $\omega$ is supported on $(1/4, 1/2) \subset (0, 3/4) \cap (1/4, 1)$, both $Y(1)$ and $Y(0)$ are point-identified on the support of $\omega$, and the target is sharply $0$.

    In contrast, \citet{magne2018using} relax the distributional constraints into moment constraints on the MTR, defined as $m_w(u) = \mathbb{E}[Y(w) \mid U=u]$. Under their framework, the data only restricts the MTR through its conditional expectation:
    \begin{align*}
        \frac{1}{p(z)} \int_{0}^{p(z)} m_1(u) \, \mathrm{d}u &= \mathbb{E}[Y \mid W=1, Z=z] = 0, \\
        \frac{1}{1-p(z)} \int_{p(z)}^{1} m_0(u) \, \mathrm{d}u &= \mathbb{E}[Y \mid W=0, Z=z] = 0.
    \end{align*}
    Evaluating at $z=0$ and $z=1$ and taking differences yields:
    \begin{equation*}
        \int_{1/4}^{3/4} m_1(u) \, \mathrm{d}u = 0, \qquad \int_{1/4}^{3/4} m_0(u) \, \mathrm{d}u = 0.
    \end{equation*}
    Their method bounds the target parameter by solving:
    \begin{equation*}
        \begin{split}
            \max_{m_0,m_1 \in \mathcal{M}} / \min_{m_0,m_1 \in \mathcal{M}} \quad & \int_{1/4}^{1/2} (m_1(u) - m_0(u)) \, \mathrm{d}u \\
            \text{s.t.} \quad & \int_{1/4}^{3/4} m_1(u) \, \mathrm{d}u = 0, \quad \int_{1/4}^{3/4} m_0(u) \, \mathrm{d}u = 0.
        \end{split}
    \end{equation*}

    If the function class $\mathcal{M}$ is sufficiently rich, this relaxation discards the point-mass nature of the outcome and yields unnecessarily wide bounds. The moment constraints only pin the \emph{integral} of each MTR over $(1/4, 3/4)$ to zero, but leave the \emph{shape} within this interval unrestricted. An adversarial optimizer exploits this by setting $m_1(u) = M$ on $(1/4, 1/2)$ and $m_1(u) = -M$ on $(1/2, 3/4)$, which satisfies $\int_{1/4}^{3/4} m_1 = 0$ yet contributes $M/4$ to the objective. Applying the symmetric construction to $m_0$ yields a total upper bound of $M/2$ instead of $0$. The resulting bounds $[-M/2, \, M/2]$ are arbitrarily wider than the sharp value of $0$. By strictly enforcing the full distributional constraint, the CCOT approach eliminates these mathematically feasible but structurally impossible counterfactuals. \cref{fig:example3_1} illustrates this gap numerically.
    \begin{figure}
        \centering
        \includegraphics[width=0.5\linewidth]{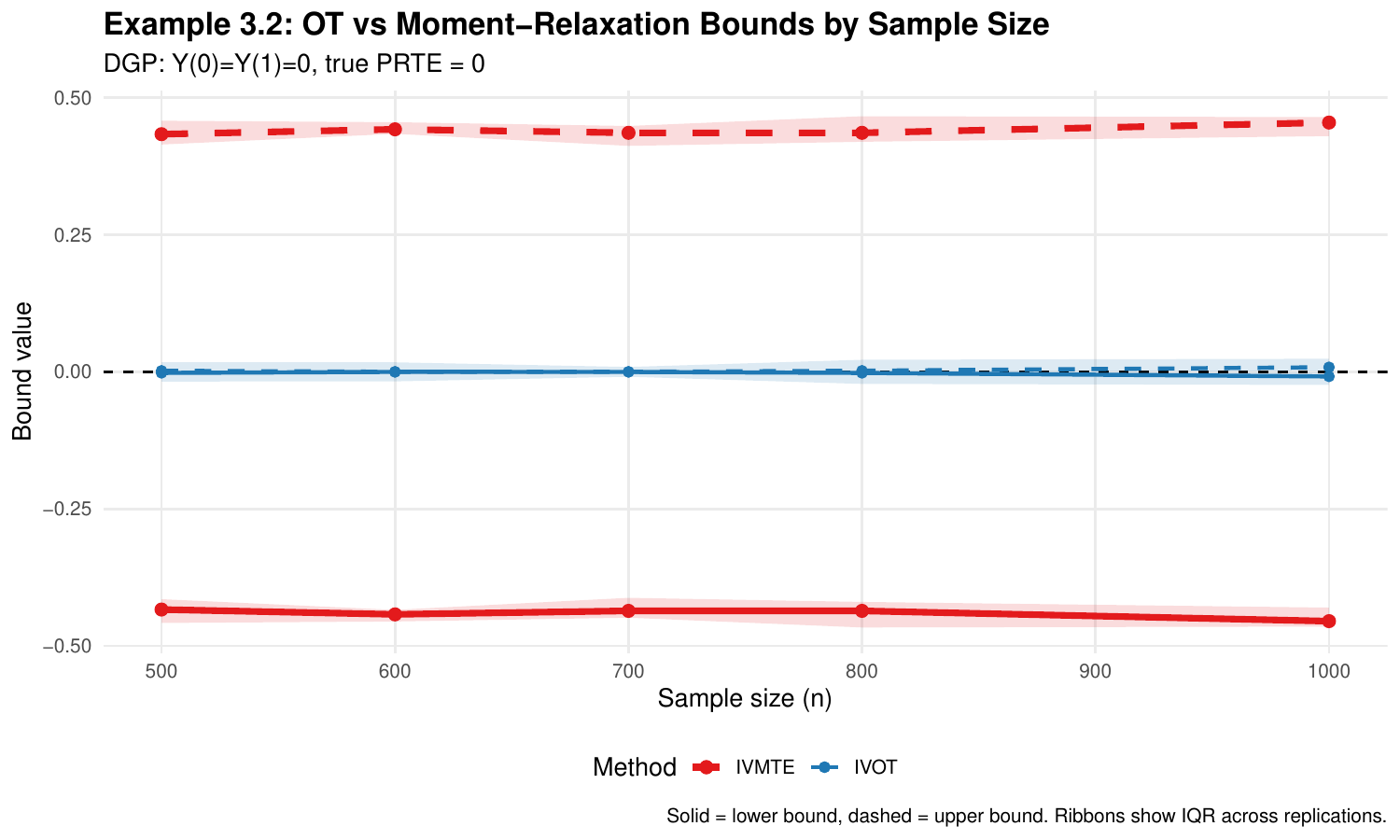}
        \caption{Bounds of the moment relaxation approach and the CCOT approach by different sample sizes. Each data point is an average of 20 runs. In the experiment $M=1$, i.e., $y_{\max}=1, y_{\min}=-1$. }
        \label{fig:example3_1}
    \end{figure}
\end{example}

\section{Proofs in \cref{sec:extension}}\label{sec:proofs_extension}

\begin{proof}[Proof of \cref{prop:sharpness_generalized}]
    Let $\{\tilde{\pi}_{w,x}\}_{w \in \mathcal{W}, x \in \mathcal{X}} \in \Gamma_{\text{gen}}(\mathbb{P}_{\text{obs}})$ be an arbitrary feasible family. By the constraints in~\eqref{eq:ot_generalized_separate}, each $\tilde{\pi}_{w,x}$ is a probability measure on $\mathcal{Y} \times [0,1]$ satisfying:
    \begin{enumerate}
        \item[(C1)] $\displaystyle \int_0^1 \tilde{\pi}_{w,x}(\mathrm{d}y \mid u) \, \mathrm{d}F_{U \mid Z, X, W}(u \mid z, x, w) = \mathbb{P}_{\text{obs}}(\mathrm{d}y \mid W\!=\!w, Z\!=\!z, X\!=\!x)$, for all $z \in \mathcal{Z}(x,w)$, for $\lambda$-a.e.\ $w \in \mathcal{W}$, and for $\mathbb{P}_{\text{obs}}$-a.e.\ $x \in \mathcal{X}$;
        \item[(C2)] $\tilde{\pi}_{w,x}(\mathrm{d}u) = \mathrm{d}u$ for $\lambda$-a.e.\ $w \in \mathcal{W}$ and $\mathbb{P}_{\text{obs}}$-a.e.\ $x \in \mathcal{X}$.
    \end{enumerate}
    Here $\tilde{\pi}_{w,x}(\mathrm{d}y \mid u)$ denotes the regular conditional distribution obtained by disintegrating $\tilde{\pi}_{w,x}$ with respect to the $U$-marginal, which exists since $\mathcal{Y} \subseteq \mathbb{R}$ is a Polish space.

    \medskip
    \noindent\textbf{Construction.}
    We construct the joint distribution $\mathbb{P}$ of $(\{Y(w)\}_{w \in \mathcal{W}}, U, X, Z, W)$ in three stages.

    \emph{Stage 1: Base variables.} Draw $(X, Z, U)$ from:
    \begin{equation}\label{eq:sharp_base_gen}
        X \sim \mathbb{P}_{\text{obs}}(\mathrm{d}x), \quad Z \mid X \sim \mathbb{P}_{\text{obs}}(\mathrm{d}z \mid x), \quad U \sim \mathrm{Unif}(0,1) \text{ independently of } (Z, X).
    \end{equation}
    Note that $U$ is drawn independently of $Z$ given $X$; this will ensure conditional instrumental exogeneity. The treatment is then deterministically assigned as $W = F_{W \mid Z,X}^{-1}(U \mid Z, X)$.

    \emph{Stage 2: Potential outcome process.} Conditional on $(U, X) = (u, x)$, we construct the stochastic process $\{Y(w)\}_{w \in \mathcal{W}}$ as follows. For each finite subset $\{w_1, \dots, w_m\} \subset \mathcal{W}$, define the finite-dimensional distribution:
    \begin{equation}\label{eq:sharp_fdd}
        \mathbb{P}\big(Y(w_1) \in A_1, \dots, Y(w_m) \in A_m \mid U = u, X = x\big) = \prod_{j=1}^{m} \tilde{\pi}_{w_j, x}(A_j \mid u).
    \end{equation}
    That is, conditional on $(U, X)$, the potential outcomes at any finite collection of treatment levels are mutually independent, with each $Y(w)$ drawn from $\tilde{\pi}_{w,x}(\cdot \mid u)$. These finite-dimensional distributions are trivially consistent under permutations and marginalization (since each factor is a probability measure). Because $\mathcal{Y}$ is Polish and $\mathcal{W}$ is a Borel subset of a compact interval, the Kolmogorov extension theorem yields a probability measure on the product $\sigma$-algebra of $\mathcal{Y}^{\mathcal{W}}$ whose finite-dimensional marginals are given by~\eqref{eq:sharp_fdd}; this is sufficient for the integral functionals studied here, and we do not require additional path regularity. This defines the conditional distribution of $\{Y(w)\}_{w \in \mathcal{W}}$ given $(U, X)$.

    \emph{Stage 3: Full joint distribution.} The joint law $\mathbb{P}$ of $(\{Y(w)\}_{w \in \mathcal{W}}, U, X, Z, W)$ is obtained by combining Stages~1 and~2. Crucially, the process $\{Y(w)\}_{w \in \mathcal{W}}$ depends on $(Z, W)$ only through $(U, X)$, since Stage~2 conditions solely on $(U, X)$.

    \medskip
    \noindent\textbf{Verification of (i): conditional distributions.}
    Fix $w \in \mathcal{W}$ in the full-$\lambda$-measure set on which (C2) holds, and fix $x \in \mathcal{X}$ in the corresponding full-$\mathbb{P}_{\mathrm{obs}}$-measure set. By~\eqref{eq:sharp_fdd} with $m = 1$, the conditional distribution of $Y(w)$ given $(U, X) = (u, x)$ under $\mathbb{P}$ is exactly $\tilde{\pi}_{w,x}(\cdot \mid u)$. Since $U$ has Lebesgue density on $[0,1]$ independently of $X$ by~\eqref{eq:sharp_base_gen}, the joint distribution of $(Y(w), U)$ given $X = x$ is:
    \begin{align*}
        \mathbb{P}(Y(w) \in A, \, U \in B \mid X = x)
        = \int_{B} \tilde{\pi}_{w,x}(A \mid u) \, \mathrm{d}u
        = \tilde{\pi}_{w,x}(Y(w) \in A, \, U \in B),
    \end{align*}
    where the last equality uses (C2). This holds for all measurable $A \subseteq \mathcal{Y}$ and $B \subseteq [0,1]$, so $(Y(w), U) \mid X = x$ has distribution $\tilde{\pi}_{w,x}$ under $\mathbb{P}$. Therefore part~(i) of the proposition holds for $\lambda$-a.e.\ $w$ and $\mathbb{P}_{\mathrm{obs}}$-a.e.\ $x$.

    \medskip
    \noindent\textbf{Verification of (ii): structural assumptions.}
    We verify each part of \cref{asp:structure_generalized}.
    \begin{itemize}
        \item \emph{Consistency.} Under $\mathbb{P}$, the treatment is $W = F_{W \mid Z,X}^{-1}(U \mid Z, X)$, and the observed outcome is $Y = Y(W)$ by definition.
        \item \emph{Conditional instrumental exogeneity.} By construction, conditional on $X = x$, the instrument $Z$ is drawn from $\mathbb{P}_{\text{obs}}(\mathrm{d}z \mid x)$ in Stage~1, while the process $(U, \{Y(w)\}_{w \in \mathcal{W}})$ is constructed in Stages~1--2 depending only on $(U, X)$ and not on $Z$. Since $U$ is drawn independently of $Z$ given $X$, and the conditional law of $\{Y(w)\}_{w \in \mathcal{W}}$ given $(U, X)$ does not involve $Z$, we have $Z \perp (U, \{Y(w)\}_{w \in \mathcal{W}}) \mid X$ under $\mathbb{P}$.
        \item \emph{Selection mechanism.} The equation $W = F_{W \mid Z,X}^{-1}(U \mid Z, X)$ holds $\mathbb{P}$-a.s.\ by Stage~1. Moreover, $U \mid X \sim \mathrm{Unif}(0,1)$ by~\eqref{eq:sharp_base_gen}.
    \end{itemize}

    \medskip
    \noindent\textbf{Verification of (iii): observational equivalence.}
    We must show that the distribution of the observables $(Y, Z, X, W)$ induced by $\mathbb{P}$ coincides with $\mathbb{P}_{\text{obs}}$. The marginal of $(Z, X)$ under $\mathbb{P}$ equals $\mathbb{P}_{\text{obs}}(Z, X)$ by~\eqref{eq:sharp_base_gen}. We verify the remaining two components: (1) the conditional distribution of $W$ given $(Z, X)$, and (2) the conditional distribution of $Y$ given $(W, Z, X)$.

    \emph{Conditional distribution of $W$ given $(Z, X)$.}
    Under $\mathbb{P}$, we have $W = F_{W \mid Z,X}^{-1}(U \mid Z, X)$ with $U \sim \mathrm{Unif}(0,1)$ independently of $(Z, X)$. By the quantile transformation, $W \mid (Z=z, X=x)$ has cumulative distribution function $F_{W \mid Z,X}(\cdot \mid z, x)$, which is exactly the observed conditional distribution of $W$ under $\mathbb{P}_{\text{obs}}$.

    \emph{Conditional distribution of $Y$ given $(W, Z, X)$.}
    Fix $(z, x)$ and $\lambda$-a.e.\ $w \in \mathcal{W}$. By the selection mechanism, the event $\{W = w\}$ corresponds to $U \in \mathcal{U}_{w}(z,x) \coloneqq \{u \in [0,1] : F_{W \mid Z,X}^{-1}(u \mid z, x) = w\}$. By the structure of the quantile function, $\mathcal{U}_{w}(z,x)$ is always an interval (possibly a singleton). On this interval, the Lebesgue measure restricted to $\mathcal{U}_{w}(z,x)$ and renormalized to a probability measure is precisely $F_{U \mid Z,X,W}(\cdot \mid z, x, w)$.

    By consistency, $Y = Y(W) = Y(w)$ on $\{W = w\}$. Since $\{Y(w')\}_{w' \in \mathcal{W}}$ is conditionally independent of $Z$ given $(U, X)$ by Stage~2, the regular conditional distribution of $Y$ given $(W = w, Z = z, X = x)$ under $\mathbb{P}$ is:
    \begin{align*}
        \mathbb{P}(Y \in A \mid W = w, Z = z, X = x)
        &= \frac{1}{|\mathcal{U}_{w}(z,x)|}\int_{\mathcal{U}_{w}(z,x)} \tilde{\pi}_{w,x}(A \mid u) \, \mathrm{d}u \\
        &= \int_0^1 \tilde{\pi}_{w,x}(A \mid u) \, \mathrm{d}F_{U \mid Z,X,W}(u \mid z, x, w).
    \end{align*}
    By constraint (C1), the right-hand side equals $\mathbb{P}_{\text{obs}}(Y \in A \mid W = w, Z = z, X = x)$ for $\lambda$-a.e.\ $w$.

    Since the conditional distribution of $W \mid (Z, X)$ and the conditional distribution of $Y \mid (W, Z, X)$ agree with $\mathbb{P}_{\text{obs}}$ up to the usual null sets in the conditioning variable, the induced distribution of $(Y, Z, X, W)$ under $\mathbb{P}$ is identical to $\mathbb{P}_{\text{obs}}$.
\end{proof}

\begin{proof}[Proof of \cref{thm:tsem_bounds}]
We prove the sharp lower bound; the upper bound follows by swapping $y_{\min}$ and $y_{\max}$ in the trivial bound term.

By \cref{prop:sharpness_generalized}, the infimum of $\theta_\omega$ over all structural models satisfying \cref{asp:structure_generalized} that are observationally equivalent to $\mathbb{P}_{\text{obs}}$ equals the infimum over the feasible set $\Gamma_{\text{gen}}(\mathbb{P}_{\text{obs}})$. The global minimization problem is:
\[
    \underline{\theta} = \min_{\{\pi_{w,x}\} \in \Gamma_{\text{gen}}(\mathbb{P}_{\text{obs}})} \int_{\mathcal{W}} \mathbb{E}_{X}\!\left[\int_0^1 \mathbb{E}_{\pi_{w,X}}[Y(w) \mid U\!=\!u]\,\omega(w,X,u)\,\mathrm{d}u\right]\mathrm{d}w,
\]
where the feasible set $\Gamma_{\text{gen}}(\mathbb{P}_{\text{obs}})$ consists of all measure families $\{\pi_{w,x}\}_{w \in \mathcal{W}, x \in \mathcal{X}}$ satisfying, for each $(x,w)$:
\begin{enumerate}
    \item[(C1)] $\displaystyle \int_0^1 \pi_{w,x}(\mathrm{d}y \mid u)\,\mathrm{d}F_{U \mid Z,X,W}(u \mid z,x,w) = \mathbb{P}_{\text{obs}}(\mathrm{d}y \mid W\!=\!w, Z\!=\!z, X\!=\!x)$, for all $z \in \mathcal{Z}(x,w)$;
    \item[(C2)] $\pi_{w,x}(\mathrm{d}u) = \mathrm{d}u$.
\end{enumerate}
For a given identifiable weight $\omega(w,x,u)$, the objective for each $(x,w)$-stratum is:
\[
    \int_0^1 \mathbb{E}_{\pi_{w,x}}[Y(w) \mid U\!=\!u]\,\omega(w,x,u)\,\mathrm{d}u.
\]

\medskip
\noindent\textbf{Separation over $(x,w)$ strata.}
Both the objective and the constraints (C1)--(C2) are separable across $(x,w)$: the kernel $\pi_{w,x}$ at one stratum is constrained independently of all other strata, and the objective decomposes additively. Since $\mathcal{Y}$ is compact (\cref{asp:bounded_y}), the integrand $\int_{\mathcal{Y}} y\,\pi_{w,x}(\mathrm{d}y \mid u) \cdot \omega(w,x,u)$ is bounded uniformly in $(x,w)$. By a measurable selection argument \citep[Proposition~7.50]{bertsekas1996stochastic}, the minimization and the outer integration can be exchanged:
\[
    \underline{\theta} = \int_{\mathcal{W}} \mathbb{E}_{X}\!\left[\underline{\theta}(w, X)\right]\mathrm{d}w,
\]
where $\underline{\theta}(x,w)$ denotes the value of the $(x,w)$-sub-problem in~\eqref{eq:ot_generalized_separate}. It therefore suffices to solve each sub-problem individually.

\medskip
\noindent\textbf{Decomposition of the objective for fixed $(x,w)$.}
Fix $(x,w)$ and drop them from the notation where unambiguous. Under \cref{asp:strict_monotonicity}, the structural treatment function $h(z,x,\cdot)$ is strictly increasing in $U$, so for each instrument value $z$, the conditional distribution $F_{U \mid Z=z, X=x, W=w}$ degenerates to a point mass at the unique $u_z$ satisfying $w = F_{W \mid Z,X}^{-1}(u_z \mid z, x)$. The identifiable support is therefore the measurable subset
\[
    \mathcal{U}_{\text{id}}(x,w) = \left\{ F_{W \mid Z,X}(w \mid z, x) : z \in \mathcal{Z}(x,w) \right\} \subseteq [0,1].
\]

We decompose the unit interval into the identified region $\mathcal{U}_{\text{id}}(x,w)$ and its complement $[0,1] \setminus \mathcal{U}_{\text{id}}(x,w)$:
\begin{align*}
    \int_0^1 \mathbb{E}_{\pi}[Y(w) \mid U\!=\!u]\,\omega(w,x,u)\,\mathrm{d}u
    &= \underbrace{\int_0^1 \indicator\!\big(u \in \mathcal{U}_{\text{id}}(x,w)\big)\,\mathbb{E}_{\pi}[Y(w) \mid U\!=\!u]\,\omega(w,x,u)\,\mathrm{d}u}_{\text{identified contribution}} \\
    &\quad + \underbrace{\int_0^1 \indicator\!\big(u \notin \mathcal{U}_{\text{id}}(x,w)\big)\,\mathbb{E}_{\pi}[Y(w) \mid U\!=\!u]\,\omega(w,x,u)\,\mathrm{d}u}_{\text{unidentified contribution}}.
\end{align*}
Because $F_{U \mid Z=z, X=x, W=w}$ is degenerate at $u_z$ for each $z$, constraint~(C1) pins down $\pi(\mathrm{d}y \mid u)$ pointwise on the identified region $u \in \mathcal{U}_{\text{id}}(x,w)$ and imposes no restriction on the complement. We therefore optimize the identified and unidentified contributions independently.

\medskip
\noindent\textbf{Point identification on $\mathcal{U}_{\text{id}}(x,w)$.}
Since $F_{U \mid Z=z, X=x, W=w}$ is a point mass at $u_z$, constraint~(C1) collapses to:
\[
    \pi_{w,x}(\mathrm{d}y \mid u_z) = \mathbb{P}_{\text{obs}}(\mathrm{d}y \mid W\!=\!w, Z\!=\!z, X\!=\!x),
    \quad \text{for each } z \in \mathcal{Z}(x,w).
\]
Hence for every $u \in \mathcal{U}_{\text{id}}(x,w)$, selecting any $z_u \in \mathcal{Z}(x,w)$ such that $F_{W \mid Z,X}(w \mid z_u, x) = u$ yields
\[
    \pi_{w,x}(\mathrm{d}y \mid u) = \mathbb{P}_{\text{obs}}(\mathrm{d}y \mid W\!=\!w, Z\!=\!z_u, X\!=\!x).
\]
This uniquely pins down $\pi_{w,x}(\mathrm{d}y \mid u)$ on $\mathcal{U}_{\text{id}}(x,w)$; if multiple instrument values map to the same $u$, the observational equivalence constraints imply the same conditional law, so the choice of $z_u$ is immaterial. In particular, there is no freedom in the coupling on this region, and the identified contribution is uniquely determined:
\begin{align*}
    &\int_0^1 \indicator\!\big(u \in \mathcal{U}_{\text{id}}(x,w)\big)\,
        \mathbb{E}_{\pi_{w,x}}[Y(w) \mid U\!=\!u]\,
        \omega(w,x,u)\,\mathrm{d}u \\
    &\quad= \int_0^1 \indicator\!\big(u \in \mathcal{U}_{\text{id}}(x,w)\big)\,
        \mathbb{E}_{\text{obs}}[Y \mid W\!=\!w, Z\!=\!z_u, X\!=\!x]\,
        \omega(w,x,u)\,\mathrm{d}u,
\end{align*}
where $z_u$ is the instrument value satisfying $F_{W \mid Z,X}(w \mid z_u, x) = u$.

\medskip
\noindent\textbf{Trivial bound on $[0,1] \setminus \mathcal{U}_{\text{id}}(x,w)$.}
For $u \notin \mathcal{U}_{\text{id}}(x,w)$, no instrument value generates treatment level $w$ at latent rank $u$. The constraints impose no restriction on $\pi_{w,x}(\mathrm{d}y \mid u)$, and the minimization is unconstrained. By \cref{asp:bounded_y}, the pointwise minimum is achieved by $\pi_{w,x}(\cdot \mid u) = \delta_{y_{\min}}$ when $\omega(w,x,u) \geqslant 0$, and by $\pi_{w,x}(\cdot \mid u) = \delta_{y_{\max}}$ when $\omega(w,x,u) < 0$, yielding the contribution:
\[
    \int_0^1 \indicator\!\big(u \notin \mathcal{U}_{\text{id}}(x,w)\big)\,
    \Big( y_{\min}\max\{0,\omega(w,x,u)\} + y_{\max}\min\{0,\omega(w,x,u)\} \Big)\mathrm{d}u.
\]

\medskip
\noindent\textbf{Aggregation.}
Combining the two regions and integrating over all $(x,w)$ strata:
\begin{align*}
    \underline{\theta}
    &= \int_{\mathcal{W}} \mathbb{E}_{X}\!\Bigg[
        \int_0^1 \indicator\!\big(u \in \mathcal{U}_{\text{id}}(x,w)\big)\,
        \mathbb{E}_{\text{obs}}[Y \mid w, Z\!=\!z_u, X]\,
        \omega(w,X,u)\,\mathrm{d}u
    \Bigg]\mathrm{d}w \\
    &\quad + \int_{\mathcal{W}} \mathbb{E}_{X}\!\Bigg[
        \int_0^1 \indicator\!\big(u \notin \mathcal{U}_{\text{id}}(x,w)\big)\,
        \Big( y_{\min}\max\{0,\omega(w,X,u)\} + y_{\max}\min\{0,\omega(w,X,u)\} \Big)\mathrm{d}u
    \Bigg]\mathrm{d}w.
\end{align*}
Integrability is immediate from the compactness of $\mathcal{Y}$.

\medskip
\noindent\textbf{Sharpness.}
By \cref{prop:sharpness_generalized}, the infimum over $\Gamma_{\text{gen}}(\mathbb{P}_{\text{obs}})$ equals the infimum over all observationally equivalent structural models. On $\mathcal{U}_{\text{id}}(x,w)$ the value is uniquely pinned down by the data; on the complement the infimum is attained by the Dirac mass $\delta_{y_{\min}}$ where $\omega(w,x,u) \geqslant 0$ and $\delta_{y_{\max}}$ where $\omega(w,x,u) < 0$. For each $(x,w)$, this defines a conditional kernel $\pi_{w,x}^\star(\mathrm{d}y \mid u)$ that is measurable in $u$ (as a piecewise composition of the data-determined kernel on the finite set $\mathcal{U}_{\text{id}}(x,w)$ and the Dirac masses on the complement). By a standard measurable selection argument \citep[Proposition~7.50]{bertsekas1996stochastic}, these conditional optimizers can be assembled into a jointly measurable map $(x,w) \mapsto \pi_{w,x}^\star$ that belongs to $\Gamma_{\text{gen}}(\mathbb{P}_{\text{obs}})$ and attains the global infimum.
By \cref{prop:sharpness_generalized}, the infimum over $\Gamma_{\text{gen}}(\mathbb{P}_{\text{obs}})$ equals the infimum over all observationally equivalent structural models. On $\mathcal{U}_{\text{id}}(x,w)$ the value is uniquely pinned down by the data; on the complement the infimum is attained by the Dirac mass $\delta_{y_{\min}}$ where $\omega(w,x,u) \geqslant 0$ and $\delta_{y_{\max}}$ where $\omega(w,x,u) < 0$. For each $(x,w)$, this defines a conditional kernel $\pi_{w,x}^\star(\mathrm{d}y \mid u)$ that is measurable in $u$ (as a piecewise composition of the data-determined kernel on the measurable set $\mathcal{U}_{\text{id}}(x,w)$ and the Dirac masses on the complement). By a standard measurable selection argument \citep[Proposition~7.50]{bertsekas1996stochastic}, these conditional optimizers can be assembled into a jointly measurable map $(x,w) \mapsto \pi_{w,x}^\star$ that belongs to $\Gamma_{\text{gen}}(\mathbb{P}_{\text{obs}})$ and attains the global infimum.
\end{proof}

\begin{proof}[Proof of \cref{thm:multivalued_bounds}]
We prove the sharp lower bound; the upper bound follows by replacing the countermonotonic coupling with the comonotonic one and swapping $y_{\min}$ and $y_{\max}$ in the trivial bound term.

For each $(x,w)$-stratum, the sub-problem from~\eqref{eq:ot_generalized_separate} reads:
\[
    \underline{\theta}(x,w) = \min_{\pi_{w,x}} \int_0^1 \mathbb{E}_{\pi_{w,x}}[Y(w) \mid U\!=\!u]\,\omega(w,x,u)\,\mathrm{d}u
\]
subject to the observational constraints: for every $z \in \mathcal{Z}(x,w)$,
\[
    \frac{1}{|I_{x,w}(z)|}\int_{I_{x,w}(z)} \pi_{w,x}(\mathrm{d}y \mid u)\,\mathrm{d}u = \mathbb{P}_{\text{obs}}(\mathrm{d}y \mid W\!=\!w, Z\!=\!z, X\!=\!x),
\]
and the uniform marginal constraint $\pi_{w,x}(\mathrm{d}u) = \mathrm{d}u$, where $I_{x,w}(z) = \{u \in [0,1] : F_{W \mid Z,X}^{-1}(u \mid z,x) = w\}$ is the latent support interval and $\omega(w,x,u)$ is the identifiable weight function.

\medskip
\noindent\textbf{Separation over $(x,w)$ strata.}
By the same separability argument as in the proof of \cref{thm:tsem_bounds}, the minimization and the outer summation can be exchanged, yielding $\underline{\theta}_\omega = \sum_{w \in \mathcal{W}} \mathbb{E}_{X}[\underline{\theta}(w, X)]$. It therefore suffices to solve each sub-problem individually.

\medskip
\noindent\textbf{Decomposition of the objective for fixed $(x,w)$.}
Fix $(x,w)$ and drop them from the notation where unambiguous. Since $\mathcal{Z}$ is discrete, the constraint~\eqref{eq:contraints_multitreatment} is a finite system. The identified support is $\mathcal{U}_{\text{id}}(x,w) = \bigcup_{z \in \mathcal{Z}(x,w)} I(z)$, and the unconstrained region is $J(\emptyset) = [0,1] \setminus \mathcal{U}_{\text{id}}(x,w)$. We decompose the objective accordingly:
\begin{align*}
    \int_0^1 \mathbb{E}_{\pi}[Y(w) \mid U\!=\!u]\,\omega(w,x,u)\,\mathrm{d}u
    &= \underbrace{\int_{\mathcal{U}_{\text{id}}(x,w)} \mathbb{E}_{\pi}[Y(w) \mid U\!=\!u]\,\omega(w,x,u)\,\mathrm{d}u}_{\text{constrained contribution}} \\
    &\quad + \underbrace{\int_{J(\emptyset)} \mathbb{E}_{\pi}[Y(w) \mid U\!=\!u]\,\omega(w,x,u)\,\mathrm{d}u}_{\text{unconstrained contribution}}.
\end{align*}
On $J(\emptyset)$, the data imposes no restriction on $\pi(\mathrm{d}y \mid u)$, so this region is optimized separately (see below). On $\mathcal{U}_{\text{id}}(x,w)$, the intervals $\{I(z)\}_{z \in \mathcal{Z}(x,w)}$ may overlap, coupling the constraints across different instruments. We disentangle these overlapping constraints using the DAG structure from \cref{asp:algebra}.

\medskip
\noindent\textbf{Disentangling the constraints via the DAG structure.}
Recall that the constraint~\eqref{eq:contraints_multitreatment} reads: for every $z \in \mathcal{Z}(x,w)$,
\[
    \frac{1}{|I(z)|}\int_{I(z)} \pi(\mathrm{d}y \mid u)\,\mathrm{d}u
    = \mathbb{P}_{\text{obs}}(\mathrm{d}y \mid W\!=\!w, Z\!=\!z, X\!=\!x),
\]
where $I(z) = I_{x,w}(z)$ and $\pi = \pi_{w,x}$.

Under \cref{asp:algebra}, the intersection-closed family $\{I(z)\}_{z \in \mathcal{Z}(x,w)}$ admits a DAG ordering by strict inclusion. Recall from~\eqref{eq:mu_isolated} the isolated disjoint sub-regions $J(z) = I(z) \setminus \bigcup_{z' \in \mathrm{children}(z)} I(z')$ and the recursively defined isolated measures $\mu_{w,z,x}$. We now verify that the original constraints are equivalent to the simplified marginal constraints~\eqref{eq:simplified_marginal_constraint}.

We proceed by induction on the DAG. For a leaf node $z$ (i.e., $\mathrm{children}(z) = \emptyset$), we have $J(z) = I(z)$, so~\eqref{eq:contraints_multitreatment} directly gives
\[
    \frac{1}{|J(z)|}\int_{J(z)} \pi(\mathrm{d}y \mid u)\,\mathrm{d}u
    = \mathbb{P}_{\text{obs}}(\mathrm{d}y \mid W\!=\!w, Z\!=\!z, X\!=\!x)
    = \mu_{w,z,x}(\mathrm{d}y).
\]
For an internal node $z$, suppose~\eqref{eq:simplified_marginal_constraint} holds for all descendants. Then from~\eqref{eq:contraints_multitreatment}:
\begin{align*}
    |I(z)|\,\mathbb{P}_{\text{obs}}(\mathrm{d}y \mid w, z, x)
    &= \int_{I(z)} \pi(\mathrm{d}y \mid u)\,\mathrm{d}u \\
    &= \int_{J(z)} \pi(\mathrm{d}y \mid u)\,\mathrm{d}u
        \;+\; \sum_{z' \in \mathrm{children}(z)}
            \int_{I(z')} \pi(\mathrm{d}y \mid u)\,\mathrm{d}u.
\end{align*}
By the inductive hypothesis applied to each child $z'$ and its descendants, the second term equals $\sum_{z' \in \mathrm{Dec}(z)} |J(z')|\,\mu_{w,z',x}(\mathrm{d}y)$. Rearranging and dividing by $|J(z)|$ recovers exactly the definition~\eqref{eq:mu_isolated}:
\[
    \frac{1}{|J(z)|}\int_{J(z)} \pi(\mathrm{d}y \mid u)\,\mathrm{d}u = \mu_{w,z,x}(\mathrm{d}y).
\]
Hence the original system~\eqref{eq:contraints_multitreatment} is equivalent to~\eqref{eq:simplified_marginal_constraint}. 

\medskip
\noindent\textbf{Independent OT sub-problems on the constrained regions.}
Since the regions $\{J(z)\}_{z \in \mathcal{Z}(x,w)}$ are pairwise disjoint and partition the identified support $\mathcal{U}_{\text{id}}(x,w)$, and the simplified constraint~\eqref{eq:simplified_marginal_constraint} couples $\pi(\mathrm{d}y \mid u)$ only within each $J(z)$, the constrained contribution from the decomposition above further decomposes:
\[
    \int_{\mathcal{U}_{\text{id}}(x,w)} \mathbb{E}_{\pi}[Y(w) \mid U\!=\!u]\,\omega(w,x,u)\,\mathrm{d}u
    = \sum_{z \in \mathcal{Z}(x,w)} \int_{J(z)} \mathbb{E}_{\pi}[Y(w) \mid U\!=\!u]\,\omega(w,x,u)\,\mathrm{d}u.
\]
No cross-region constraints exist, so the optimization on the identified support reduces to independent 1D problems.

For each $z \in \mathcal{Z}(x,w)$, the constraint~\eqref{eq:simplified_marginal_constraint} fixes the marginal distribution of $Y(w)$ on $J(z)$ to be $\mu_{w,z,x}$, while the $U$-marginal is $\mathrm{Unif}(J(z))$. The sub-problem on $J(z)$ is therefore a standard 1D optimal transport problem:
\[
    \min_{\gamma_z \in \Pi(\mu_{w,z,x},\, \mathrm{Unif}(J(z)))}
    \int_{\mathcal{Y} \times J(z)} y \cdot \omega(w,x,u)\,\mathrm{d}\gamma_z(y,u).
\]

\medskip
\noindent\textbf{Closed-form solution via countermonotonic coupling.}
Since the cost is a product of a function of $y$ and a function of $u$, we invoke \cref{thm:1d_ot} directly (countermonotonic coupling for the minimum):
\[
    \min_{\gamma_z \in \Pi(\mu_{w,z,x},\, \mathrm{Unif}(J(z)))}
    \int y \cdot \omega(w,x,u)\,\mathrm{d}\gamma_z
    = |J(z)|\int_0^1 Q_{Y,J(z)}(t)\,Q_{\omega,J(z)}(1-t)\,\mathrm{d}t,
\]
where $Q_{Y,J(z)}$ is the quantile function of $\mu_{w,z,x}$ and $Q_{\omega,J(z)}$ is the quantile function of $\omega(w,x,U)$ for $U \sim \mathrm{Unif}(J(z))$. The prefactor $|J(z)|$ arises from the Jacobian of the rescaling.

\medskip
\noindent\textbf{Trivial bound on the unconstrained region.}
On $J(\emptyset) = [0,1] \setminus \bigcup_{z \in \mathcal{Z}(x,w)} I(z)$, the data imposes no restriction on $\pi(\mathrm{d}y \mid u)$. By \cref{asp:bounded_y}, the pointwise minimum is $\pi(\cdot \mid u) = \delta_{y_{\min}}$ when $\omega(w,x,u) \geqslant 0$ and $\pi(\cdot \mid u) = \delta_{y_{\max}}$ when $\omega(w,x,u) < 0$, contributing:
\[
    \int_{J(\emptyset)} \Big( y_{\min}\max\{0,\omega(w,x,u)\} + y_{\max}\min\{0,\omega(w,x,u)\} \Big)\mathrm{d}u.
\]

\medskip
\noindent\textbf{Aggregation.}
Summing over all constrained regions and the unconstrained remainder:
\begin{align*}
        \underline{\theta}(x,w)
   & = \sum_{z \in \mathcal{Z}(x,w)} |J_{x,w}(z)|
        \int_0^1 Q_{Y,J_{x,w}(z)}(t)\,Q_{\omega,J_{x,w}(z)}(1-t)\,\mathrm{d}t
    \;\\ &\quad \quad \quad \quad+\; \int_{J_{x,w}(\emptyset)} \Big(y_{\min}\max\{0,\omega(w,x,u)\} + y_{\max}\min\{0,\omega(w,x,u)\}\Big)\mathrm{d}u.
\end{align*}
The global bound follows by summing over treatment levels: $\underline{\theta}_\omega = \sum_{w \in \mathcal{W}} \mathbb{E}_{X}[\underline{\theta}(w, X)]$. Integrability is immediate from compactness of $\mathcal{Y}$ and boundedness of $\omega$.

\medskip
\noindent\textbf{Sharpness.}
Sharpness follows by the same argument as in the proof of \cref{thm:tsem_bounds}: on each $J(z)$ the infimum is attained by the countermonotonic coupling, on $J(\emptyset)$ by the appropriate Dirac masses, and a measurable selection argument assembles these into a globally feasible optimizer.
\end{proof}

\section{Estimation Procedure Details}\label{app:procedure_details}

This appendix collects the full procedural details abbreviated in \cref{sec:est_inf} by Algorithms~\ref{alg:dml_continuous_outcome}--\ref{alg:continuous_instrument}, plus the formal definitions and derivations underlying the closed-form rewritings~(\ref{eq:target_discrete_iv}) and~(\ref{eq:continuous_simplified_target}).

\subsection{Proof of \cref{prop:discrete_reduction}}\label{app:closed_form_rewriting}

The nuisance components $\gamma_{\text{full},k}, \gamma_{j,k}, \gamma_K, J_{\text{full},k}, J_{j,k}, \Delta_K$ and the relative threshold $\kappa_{j,k}$ are defined in \cref{sec:est_inf} immediately preceding \cref{prop:discrete_reduction}; we refer to those formulas throughout this proof.

\begin{proof}[Proof of \cref{prop:discrete_reduction}]
    By \cref{thm:discrete_lower_bound}, it suffices to rewrite the conditional integral $\int_0^1 Q_{Y,k\mid X}(u)\,Q_{\omega,k\mid X}(1-u)\,\mathrm{d}u$ as an observable finite weighted sum. Fix $X$. On each baseline interval $(p_k(X), p_{k+1}(X))$ the weight function $\omega(X,u)$ is a non-increasing step function, constant between consecutive alternative-policy values $\{q_{j,k}(X)\}_{j=0}^{l_k}$. On the sub-interval $(q_{j,k}(X), q_{j+1,k}(X))$ (with the convention $q_{0,k}(X)=p_k(X)$ and $q_{l_k+1,k}(X) = p_{k+1}(X)$), exactly those instruments with $q(Z,X) > q_{j,k}(X)$ within the $k$-th interval and all instruments in strictly higher intervals contribute; hence
    \begin{equation*}
        \omega(X, q_{j,k}(X)) = \mathbb{P}_{\text{obs}}\!\left(Z \in \bigcup_{l=j}^{l_k} T_{l,k} \mathrel{\big|} X\right) + \mathbb{P}_{\text{obs}}\!\left(Z \in \bigcup_{i=k+1}^K S_i \mathrel{\big|} X\right) - \mathbb{P}_{\text{obs}}\!\left(Z \in \bigcup_{i=k+1}^K \bigcup_{l=0}^{l_i} T_{l,i} \mathrel{\big|} X\right).
    \end{equation*}
    Regrouping, this step-function value decomposes as $\omega(X, q_{j,k}(X)) = \gamma_{j,k}(X) - \gamma_{\text{full},k}(X) \cdot \indicator(j \geqslant 1) - \gamma_{\text{full},k}(X)\indicator(j=0)$, where the $\gamma_{\text{full},k}(X)$ term collects the "full" contribution acting on the entire interval and $\gamma_{j,k}(X)$ collects the additional sub-interval mass. Multiplying by $Q_{Y,k\mid X}$ and integrating over each sub-interval $[\kappa_{j,k}(X), \kappa_{j+1,k}(X)]$ in the normalized scale $u = (v - p_k(X))/(p_{k+1}(X) - p_k(X))$ yields
    \begin{equation*}
        (p_{k+1}(X) - p_k(X)) \int_0^1 Q_{Y,k\mid X}(u)\,Q_{\omega,k\mid X}(1-u)\,\mathrm{d}u = \gamma_{\text{full},k}(X)\, J_{\text{full},k}(X) + \sum_{j=1}^{l_k}\gamma_{j,k}(X)\, J_{j-1,k}(X),
    \end{equation*}
    where we used $J_{\text{full},k}(X) = (p_{k+1}(X)-p_k(X))\,\mathbb{E}_{\mu_{1,k\mid X}}[Y] = (p_{k+1}(X)-p_k(X))\int_0^1 Q_{Y,k\mid X}(u)\,\mathrm{d}u$ and the definition of $J_{j-1,k}$.

    The trivial-bound term on the tail $(p_K(X), 1]$ simplifies for the PRTE weight: since $\omega(X,u) = \mathbb{P}_{\text{obs}}(u < q(Z,X) \mid X) \geqslant 0$ for $u > p_K(X)$, we have $\int_{p_K(X)}^1 (y_{\min}\max\{0,\omega\} + y_{\max}\min\{0,\omega\})\,\mathrm{d}u = y_{\min}\int_{p_K(X)}^1 \omega(X,u)\,\mathrm{d}u$, which equals $\Delta_K(X)$ by Fubini. Taking the outer expectation over $X$ and summing over $k$ yields~(\ref{eq:target_discrete_iv}).

    The corresponding representation for the upper bound is obtained by the symmetric argument with $Q_{\omega,k\mid X}(1-t)$ replaced by $Q_{\omega,k\mid X}(t)$ in the OT term and $y_{\min}, y_{\max}$ swapped in the tail.
\end{proof}

\subsection{Proof of \cref{prop:J_specialization}}\label{app:J_specialization}

\begin{proof}[Proof of part (i).] Fix $k$ and $j \in \{0,\dots,l_k-1\}$. By the change of variables $u = F_{\mu_{1,k\mid X}}(y)$, continuity of $\mu_{1,k\mid X}$ at $\nu_{j,k}(X)$ and $\nu_{j+1,k}(X)$ gives $F_{\mu_{1,k\mid X}}(\nu_{j,k}(X)) = \kappa_{j,k}(X)$ and $F_{\mu_{1,k\mid X}}(\nu_{j+1,k}(X)) = \kappa_{j+1,k}(X)$ almost surely. Therefore,
\begin{align*}
    \int_{\kappa_{j,k}(X)}^{\kappa_{j+1,k}(X)} Q_{Y,k\mid X}(u)\,\mathrm{d}u &= \int_{\nu_{j,k}(X)}^{\nu_{j+1,k}(X)} y\, \mathrm{d}F_{\mu_{1,k\mid X}}(y) = \mathbb{E}_{\mu_{1,k\mid X}}\!\big[Y\,\indicator\!\big(\nu_{j,k}(X) < Y \leqslant \nu_{j+1,k}(X)\big)\big],
\end{align*}
where the second equality uses the continuity of $\mu_{1,k\mid X}$ at the two boundary points to include/exclude boundary atoms without ambiguity. Multiplying by $p_{k+1}(X)-p_k(X)$ yields~\eqref{eq:J_cexpectation}.
\end{proof}

\begin{proof}[Proof of part (ii).] Suppose $Y \in \{0,1\}$. Let $\theta_k(X) \coloneqq \mathbb{E}_{\mu_{1,k\mid X}}[Y]$ denote the complier success probability on the $k$-th interval. By the definition of the identified complier measure \citep[cf.][]{angrist1995identification}, $(p_{k+1}(X)-p_k(X))\,\theta_k(X) = \mathbb{E}_{\text{obs}}[YW \mid Z \in S_{k+1}, X] - \mathbb{E}_{\text{obs}}[YW \mid Z \in S_k, X] = P_{1,k+1}(X) - P_{1,k}(X)$, so $J_{\text{full},k}(X) = P_{1,k+1}(X) - P_{1,k}(X)$.

Since $Y \in \{0,1\}$, the conditional quantile function is the step function $Q_{Y,k\mid X}(u) = \indicator(u \geqslant 1 - \theta_k(X))$. Equivalently, $Q_{Y,k\mid X}(u) = 1$ iff $u(p_{k+1}(X) - p_k(X)) + p_k(X) \geqslant p_{k+1}(X) - J_{\text{full},k}(X)$, i.e., the corresponding propensity value $v = p_k(X) + u(p_{k+1}(X)-p_k(X))$ lies in the interval $[p_{k+1}(X) - J_{\text{full},k}(X),\, p_{k+1}(X)]$. Writing $J_{j-1,k}(X)$ as a length in the original propensity scale, it equals the length of the intersection of $[q_{j-1,k}(X), q_{j,k}(X)]$ with $[p_{k+1}(X) - J_{\text{full},k}(X),\, p_{k+1}(X)]$:
\begin{equation*}
    J_{j-1,k}(X) = \max\!\Big(0,\; q_{j,k}(X) - \max\!\big(q_{j-1,k}(X),\; p_{k+1}(X) - J_{\text{full},k}(X)\big)\Big) = h_{j,k}^{-}(X).
\end{equation*}
Substituting into~(\ref{eq:target_discrete_iv}) yields~\eqref{eq:binary_lower_bound}.
\end{proof}

\subsection{Discrete Instrument, Continuous Outcome: Nuisance Estimation}\label{app:procedure_continuous_outcome}

We give the full nuisance-estimation steps summarized in \cref{alg:dml_continuous_outcome}. Randomly partition the observed dataset of size $n$ into two disjoint sets: an auxiliary sample $I_1$ and a main estimation sample $I_2$, each of size $n/2$. Using only $I_1$, we estimate the collection of nuisance parameters as follows.
\begin{enumerate}
    \item \textbf{Propensity Scores:} Estimate the baseline propensity score $p(z,x)$ by regressing the treatment $W$ on covariates $X$ conditional on each $z \in \mathcal{Z}$. Compute the alternative policy $\hat{q}(z,x) = \phi(z, x, \hat{p}(z,x))$ by plugging in $\hat{p}$.
    \item \textbf{Level Sets:} By \cref{asp:prop_gap}, the true score levels are separated by at least $c_{\text{gap}}$. We first compute the empirical marginal averages over the auxiliary sample: $\bar{p}(z) = \frac{1}{|I_1|} \sum_{i \in I_1} \hat{p}(z, X_i)$ and $\bar{q}(z) = \frac{1}{|I_1|} \sum_{i \in I_1} \hat{q}(z, X_i)$. Let $V = \{\bar{p}(z)\}_{z\in\mathcal{Z}} \cup \{\bar{q}(z)\}_{z\in\mathcal{Z}}$. We group the elements in $V$ to disjoint subsets $V_1,\cdots,V_E$ such that the radius of each subset is less than $c_\text{gap}/2$.  To construct the baseline sets $\widehat{S}_k$, we partition $\mathcal{Z}$ into equivalence classes where $z$ and $z'$ are grouped if $\bar{p}(z)$ and $\bar{p}(z')$ are in the same $V_e$. Let $\hat{p}_k$ be the mean of $\bar{p}(z)$ within $\widehat{S}_k$, indexed such that $\hat{p}_k < \hat{p}_{k+1}$. Similarly, to construct the alternative sets $\widehat{T}_{j,k}$, we partition this subset by grouping $z$ and $z'$ if $\bar{q}(z),\bar{q}(z')$ are in the same subset $V_e$. Ordering these resulting groups by their mean $\bar{q}$ values yields $\widehat{T}_{1,k}, \dots, \widehat{T}_{l_k,k}$. Let $\hat{q}_{j,k}$ be the mean of $\bar{q}(z)$ within $\widehat{T}_{j,k}$, indexed such that $\hat{p}_k\leqslant \hat{q}_{1,k}<\cdots<\hat{q}_{\hat{l}_k,k} < \hat{p}_{k+1}$.
    \item \textbf{Instrument Assignment: } Estimate $\pi_k(x) = \mathbb{P}_{\text{obs}}(Z \in S_k \mid X=x)$ by regressing $\indicator(Z \in \widehat{S}_k)$ on $X$. Similarly, estimate the weighting probabilities $\gamma_{j,k}(x)$ and $\gamma_K(x)$ by regressing their respective set indicators on $X$.
    \item \textbf{Conditional Quantiles:} By definition, $\nu_{j,k}(x) = Q_{Y,k \mid x}(\kappa_{j,k}(x))$ is the $\kappa_{j,k}(x)$-quantile of the conditional distribution $\mu_{1,k \mid x}$, where $\mu_{1,k \mid x}$ is identified in (\ref{eq:y_identified_measure}).
    Applying the quantile characterization $\mathbb{P}_{\mu_{1,k \mid x}}(Y \leqslant \nu_{j,k}(x)) = \kappa_{j,k}(x) = \frac{q_{j,k}(x) - p_k(x)}{p_{k+1}(x) - p_k(x)}$ and expanding $\mathbb{P}_{\mu_{1,k \mid x}}$ in terms of the observed conditional distributions yields the following moment equation:
    \begin{align*}
    p_{k+1} (X)&\mathbb{E} [\indicator (Y\leqslant \nu _{j,k} (X))\mid X,Z\in S_{k+1}, W=1 ]\\
    &-p_{k} (X)\mathbb{E} [ \indicator  (Y\leqslant \nu _{j,k} (X))\mid X,Z\in S_{k}, W=1 ]-(q_{j,k} (X)-p_{k} (X)) =0.
\end{align*}
We will estimate $\nu_{j,k}$ by
\begin{align*}
\nu_{j,k}( x) & =\arg\min_{\nu }\mathbb{E}\left[\Omega(Z,X, W)\max\{0,\nu ( X) -Y\} -( q_{j,k}( X) -p_{k}( X)) \nu ( X)\right],
\end{align*}
where $$\Omega(Z,X, W) = W \cdot \left(\frac{\indicator(Z\in S_{k+1} )}{\pi _{k+1}( X)} -\frac{ \indicator(Z\in S_{k} )}{\pi _{k}( X)}\right).$$
Then, we estimate the conditional expectations of the indicators within the moment constraints:
\begin{align*}
    M_{j,k}^{+}(X) &\coloneqq \mathbb{E}\left[\indicator(Y \leqslant \nu_{j,k}(X)) \mathrel{|} Z \in S_{k+1}, W=1,  X \right], \\
    M_{j,k}^{-}(X) &\coloneqq \mathbb{E}\left[\indicator(Y \leqslant \nu_{j,k}(X)) \mathrel{|} Z \in S_{k}, W=1,  X \right].
\end{align*}

    \item \textbf{Estimate} $J_{j,k}(X)$ \textbf{and} $J_{\text{full},i}(X)$: We estimate $J_{j,k}(x)$ by
\begin{align*}
J_{j,k}^{+}(X) & =\mathbb{E}[ Y\indicator( \nu _{j,k}( X) < Y \leqslant \nu _{j+1,k}( X)) \mid X,Z\in S_{k+1} ,W=1],\\
J_{j,k}^{-}(X) & =\mathbb{E}[ Y\indicator( \nu _{j,k}( X) < Y \leqslant \nu _{j+1,k}( X)) \mid X,Z\in S_{k} ,W=1] ,\\
J_{j,k}(X) & = p_{k+1}( X)J_{j,k}^{+}(X) - p_{k}( X)J_{j,k}^{-}(X),
\end{align*}
and estimate $ J_{\text{full} ,k} (X) $ by
\begin{align*}
J_{\text{full} ,k}^{+} (X) & =\mathbb{E}[ Y\mid X,Z\in S_{k+1} ,W=1]  ,\\
J_{\text{full} ,k}^{-} (X) & =\mathbb{E}[ Y\mid X,Z\in S_{k} ,W=1] ,\\
J_{\text{full} ,k} (X) & = p_{k+1}( X)J_{\text{full} ,k}^{+} (X)-  p_{k}( X)J_{\text{full} ,k}^{-} (X).
\end{align*}
\end{enumerate}

In summary, the nuisance functions of this problem are, for all interval indices $k \in \{0, \dots, K\}$ and sub-interval indices $j \in \{0, \dots, l_k\}$:
$$\eta(x) =  \left\{p( z,x) ,\, p_{k}( x) ,\, \pi _{k}( x),\, \gamma_K(x) ,\, \nu _{j,k}( x),\, M_{j,k}^{\pm}(x),\, J_{j,k}^{\pm}( x),\, J_{\text{full} ,k}^{\pm}(x)\right\}_{j,k}.$$

The Neyman-orthogonal score for the target~(\ref{eq:target_discrete_iv}) is constructed by combining empirical influence functions for the conditional expectations $J_{j,k}^{\pm}(X)$ and $J_{\mathrm{full},k}^{\pm}(X)$ with pathwise derivatives arising from the quantile thresholds $\nu_{j,k}(X)$ and propensity score boundaries $p_k(X)$. In addition, the score must correct for the estimation of (i)~the conditional quantile indicators $M_{j,k}^{\pm}(X)$ via their Riesz representers, (ii)~the propensity scores $p_k(X)$ through IPW residuals, and (iii)~the alternative policy values $q_{j,k}(X)$ through the derivative $\frac{\partial \phi}{\partial p}$. The weighting probabilities $\gamma(X)$ are handled by forming product scores using the chain rule for influence functions. The complete formulas for all score components---$\psi_{j,k}^J$, $\psi_{\mathrm{full},k}^J$, $\psi_{j,k}^{\mathrm{prod}}$, $\psi_{\mathrm{full},k}^{\mathrm{prod}}$, and $\psi_{\Delta_K}$---are given in \cref{app:scores_continuous}.

\subsection{Discrete Instrument, Binary Outcome: Nuisance Estimation}\label{app:procedure_binary_outcome}

The closed-form reduction~(\ref{eq:binary_lower_bound}) is proved as part~(ii) of \cref{prop:J_specialization} in \cref{app:J_specialization}. Here we give the full nuisance-estimation steps summarized in \cref{alg:dml_binary_outcome}. The procedure follows the same two-step sample-splitting strategy as the continuous outcome case, but with a simpler nuisance list: the conditional quantile estimation (Step~4 of the continuous-outcome procedure) and the indicator expectations $M_{j,k}^{\pm}$ are no longer needed. The full nuisance list, for all interval indices $k \in \{0, \dots, K\}$ and sub-interval indices $j \in \{0, \dots, l_k\}$, is
\begin{equation*}
    \eta(x) = \Big\{p(z,x),\; p_k(x),\; \pi_k(x),\; \gamma_{\text{full},k}(x),\; \gamma_{j,k}(x),\; \gamma_K(x),\; P_{1,k}(x)\Big\}_{j,k}.
\end{equation*}

\textbf{Step 1: Nuisance Estimation.} Using the auxiliary sample $I_1$:
\begin{enumerate}
    \item \textbf{Propensity Scores and Level Sets:} Estimate $p(z,x)$, compute $\hat{q}(z,x) = \phi(z,x,\hat{p}(z,x))$, and construct the level sets $\widehat{S}_k$, $\widehat{T}_{j,k}$ exactly as in the continuous outcome procedure (Steps~1--3 in \cref{app:procedure_continuous_outcome}).
    \item \textbf{Conditional Joint Probability:} Estimate $P_{1,k}(x) = \mathbb{E}[YW \mid Z \in S_k, X=x]$ by regressing $YW$ on $X$ conditional on $Z \in \widehat{S}_k$, using any flexible machine learning method. Similarly, estimate $\gamma_{\text{full},k}(x)$, $\gamma_{j,k}(x)$, and $\gamma_K(x)$ by regressing their respective set indicators on $X$.
    \item \textbf{Rearrangement:} In the population, $P_{1,k}(X)$ is monotone non-decreasing in $k$ and the gaps are bounded by $P_{1,k+1}(X) - P_{1,k}(X) \leqslant p_{k+1}(X) - p_k(X)$, since $J_{\text{full},k}(X) = (p_{k+1}(X) - p_k(X))\,\theta_k(X)$ with $\theta_k(X) \in [0,1]$. However, because each $\hat{P}_{1,k}$ is estimated on a disjoint subsample ($Z \in \widehat{S}_k$) by a separate regression, the finite-sample estimates need not preserve either constraint. We therefore project them back onto the structurally feasible set via monotone rearrangement \citep{chernozhukov2009improving}: for each observation $X_i$, sort the vector $(\hat{P}_{1,0}(X_i), \ldots, \hat{P}_{1,K}(X_i))$ in ascending order to obtain $(\hat{P}_{1,0}^*(X_i), \ldots, \hat{P}_{1,K}^*(X_i))$.
\end{enumerate}

\textbf{Step 2: Neyman-Orthogonal Score Construction.} The target functional for the lower bound is $\underline{\theta}_{\omega,1} = \mathbb{E}[m(X; \eta)]$, where
\begin{equation*}
    m(X; \eta) = \sum_{k=0}^{K-1}\bigg(\gamma_{\text{full},k}(X)\, J_{\text{full},k}(X) + \sum_{j=1}^{l_k} \gamma_{j,k}(X)\, h_{j,k}^{-}(X)\bigg) + \Delta_K(X),
\end{equation*}
and the Neyman-orthogonal score takes the aggregated form
\begin{align}\label{eq:binary_score}
    \psi(O; \eta) &= \sum_{k=0}^{K-1}\bigg(\psi_{\text{full},k}^{\text{prod}}(O; \eta) + \sum_{j=1}^{l_k} \psi_{j,k}^{\text{prod}}(O; \eta)\bigg) + \psi_{\Delta_K}(O; \eta).
\end{align}
The construction of each component score proceeds by computing the pathwise derivatives of $h_{j,k}^{-}(X)$ with respect to $P_{1,k}(X)$, $p_k(X)$, and $q_{j,k}(X)$, and combining them with standard IPW residuals. Because $h_{j,k}^{-}$ involves a $\max$ operator, the pathwise derivatives depend on active-set, boundary, and dominance indicators that track which constraint binds in the $\max$. The product scores account for the estimation of $\gamma(X)$ via the chain rule for influence functions, and the trivial-bound score $\psi_{\Delta_K}$ is identical to the continuous-outcome case. The complete formulas are given in \cref{app:scores_discrete}. Averaging the plug-in scores over $I_2$ yields the estimator previewed in~(\ref{eq:binary_estimator}); the upper-bound score is derived similarly.

\subsection{Proof of \cref{prop:continuous_reduction}}\label{app:procedure_fubini}

\begin{proof}[Proof of \cref{prop:continuous_reduction}]
    Because $\mathcal{Z}$ is connected and $z \mapsto p(z,X)$ is continuous by \cref{asp:continuity_p}, its image is a connected interval for almost every $X$, so \cref{thm:covariate_continuous_lower_bound} specializes with $K_X = 1$ and conditional LIV interval $[\underline{p}(X), \overline{p}(X)]$. We derive the three terms of~(\ref{eq:continuous_simplified_target}) in turn. Substituting the explicit form of the weight function,
    \begin{equation*}
        \omega(x,u) = \mathbb{E}_{Z \mid X}[\indicator (u \leqslant q(Z,X)) - \indicator (u \leqslant p(Z,X)) \mid X=x],
    \end{equation*}
    and applying Fubini's theorem to swap the order of integration, we handle each term separately.

    \emph{First term (quantile integral on the gap $(0, \underline{p}(X))$).} Substituting the weight function and swapping expectations, apply a change of variables $v = \underline{p}(X)u$ to avoid dividing by the baseline propensity score in the integration limits. Using $\indicator (A) - 1 = -\indicator (A^c)$, the indicator becomes $-\indicator (v > q(Z,X))$, which restricts the domain of integration to the upper tail:
    \begin{align*}
        \mathbb{E}_{X}\!\left[ \underline{p}(X)\int_{0}^{1} Q_{Y,\underline{p} \mid X}(u)\omega(X, \underline{p}(X)u)\, \mathrm{d}u \right]
        &= \mathbb{E}_{Z,X}\!\left[ \underline{p}(X)\int_{0}^{1} Q_{Y,\underline{p} \mid X}(u)\big(\indicator (\underline{p}(X)u \leqslant q(Z,X)) - 1\big)\, \mathrm{d}u \right] \\
        &= -\mathbb{E}_{Z,X}\!\left[ \int_{\min\{\underline{p}(X), q(Z,X)\}}^{\underline{p}(X)} Q_{Y,\underline{p} \mid X}\!\left({v}/{\underline{p}(X)}\right)\, \mathrm{d}v \right].
    \end{align*}

    \emph{Second term ($g_1$ difference on the LIV interval).} Integrating the derivative of $g_1$ yields $g_1$ evaluated at the truncated limits defined by the indicators. The evaluation at the lower endpoint $\underline{p}(X)$ cancels between the two indicator terms:
    \begin{align*}
       \mathbb{E}_{X}\!\left[ \int_{\underline{p}(X)}^{\overline{p}(X)} \frac{\partial g_1(u, X)}{\partial u} \omega(X,u)\, \mathrm{d}u \right]   &= \mathbb{E}_{Z,X}\!\left[ \int_{\underline{p}(X)}^{\overline{p}(X)} \frac{\partial g_1(u, X)}{\partial u} \big(\indicator (u \leqslant q(Z,X)) - \indicator (u \leqslant p(Z,X))\big)\, \mathrm{d}u \right] \\
        &= \mathbb{E}_{Z,X}\!\left[ g_1\big(\min\{\overline{p}(X), q(Z,X)\}, X\big) - g_1\big(p(Z,X), X\big) \right].
    \end{align*}

    \emph{Third term (trivial bound on $(\overline{p}(X),1]$).} Since $p(Z,X) \leqslant \overline{p}(X)$, we have $\indicator(u \leqslant p(Z,X)) = 0$ for $u \geqslant \overline{p}(X)$, so $\omega(X,u) = \mathbb{P}(u \leqslant q(Z,X) \mid X) \geqslant 0$ for the PRTE weight on $(\overline{p}(X), 1]$. Thus $\int_{\overline{p}(X)}^1 \big(y_{\min}\max\{0,\omega\} + y_{\max}\min\{0,\omega\}\big)\,\mathrm{d}u$ reduces to
    \begin{align*}
        y_{\min} \mathbb{E}_{X}\!\left[\int_{\overline{p}(X)}^{1} \omega(X,u)\, \mathrm{d}u\right] &= y_{\min} \mathbb{E}_{Z,X}\!\left[\int_{\overline{p}(X)}^{1} \indicator (u \leqslant q(Z,X))\, \mathrm{d}u\right] = y_{\min}\, \mathbb{E}_{Z,X}\!\left[\max\{0, q(Z,X) - \overline{p}(X)\}\right].
    \end{align*}
    Aggregating the three components yields~(\ref{eq:continuous_simplified_target}). The upper-bound representation follows from the symmetric argument, swapping $Q_{Y,\underline{p}\mid X}(u)$ for $Q_{Y,\underline{p}\mid X}(1-u)$ in the OT term and $y_{\min}$ for $y_{\max}$ in the trivial tail.
\end{proof}

\subsection{Continuous Instrument: Localized $g_1$ and Boundary Quantile Estimation}\label{app:procedure_continuous_instrument}

We give the full procedural details summarized in \cref{alg:continuous_instrument}. We randomly partition the observed dataset of size $n$ into three disjoint subsets: a propensity estimation sample $I_0$, a nuisance estimation sample $I_1$, and a main evaluation sample $I_2$, each of size $n/3$. The three-fold split is needed because the boundary quantile estimator is trained on a subset of observations selected using the estimated propensity score; by estimating $\hat{p}$ on the independent sample $I_0$, the selection rule becomes a fixed function conditional on $I_0$, so the selected observations in $I_1$ remain conditionally i.i.d.

\paragraph{Propensity Scores and Boundaries.} Using the propensity estimation sample $I_0$, estimate the baseline propensity score $p(z,x)$ by regressing the treatment $W$ on $(Z, X)$ to obtain $\hat{p}(z,x)$. Compute the alternative policy propensity by plugging in the estimated score: $\hat{q}(z,x) = \phi\big(z, x, \hat{p}(z,x)\big)$. Next, for each $x \in \mathcal{X}$, estimate the conditional boundary limits by taking the supremum and infimum over the instrument support:
    \begin{equation*}
        \widehat{\overline{p}}(x) = \sup_{z \in \mathcal{Z}} \hat{p}(z,x), \quad \text{and} \quad \widehat{\underline{p}}(x) = \inf_{z \in \mathcal{Z}} \hat{p}(z,x).
    \end{equation*}

\paragraph{Conditional Expectation ($g_1$) via Localized ML.} Estimating $g_1(u, x) = \mathbb{E}_{\text{obs}}[YW \mid p(Z,X)=u, X=x]$ poses two competing challenges. The covariate dimension $d_X$ precludes classical nonparametric estimators and calls for flexible machine learning, yet \eqref{eq:continuous_simplified_target} evaluates $g_1$ at the {estimated} propensity scores $\hat{p}(Z,X)$ and $\min\{\widehat{\overline{p}}(X), \hat{q}(Z,X)\}$, so the functional expansion underlying \cref{thm:continuous_rate} requires $\hat{g}_1(u,x)$ to be Lipschitz continuous in $u$ with high probability. An off-the-shelf ML regression of $YW$ on the joint feature vector $(u, X)$---using random forests, boosted trees, or ReLU networks---generally produces fits that are piecewise constant or non-smooth in $u$ and therefore does not deliver this guarantee.

    We resolve both issues by combining a kernel-weighted local $M$-estimator in the scalar direction $u$ with an unconstrained ML fit in the covariate direction $x$, and then interpolating between grid points. Structurally, this extends classical local polynomial regression \citep{fan1996local} to ML base learners and can be viewed as a local $M$-estimator in the sense of the generalized random forest framework \citep{athey2019generalized}, with our standard kernel in $u$ playing the role of the adaptive forest weights. The closest precedent in the causal-inference literature is the local-polynomial DR-learner of \citet{kennedy2023towards}, which pairs ML-estimated nuisances with a local polynomial second-stage smoother to obtain smoothness in a conditioning variable.

    Concretely, let $K(\cdot)$ be a continuously differentiable $1$D kernel with bounded derivative and let $h_n$ be a bandwidth. Fix an equidistant grid $u_m = m/M$ for $m=0, 1, \dots, M$ covering $[0,1]$. For each grid point $u_m$, we train a flexible ML model $\hat{f}_{u_m} \in \mathcal{F}$ (e.g., random forests, neural networks, or $\ell_1$-penalized regression) by minimizing the kernel-weighted squared-error loss on the nuisance estimation sample $I_1$, using the propensity scores $\hat{p}$ estimated from $I_0$:
    \begin{equation*}
        \hat{f}_{u_m} \;=\; \arg\min_{f \in \mathcal{F}} \sum_{i \in I_1} K\!\left( \frac{u_m - \hat{p}(Z_i, X_i)}{h_n} \right) \big(Y_i W_i - f(X_i)\big)^2.
    \end{equation*}
    Each $\hat{f}_{u_m}(x)$ targets $g_1(u_m, x)$ using observations whose estimated propensity score lies within an $O(h_n)$ neighborhood of $u_m$. For any $u \in [u_m, u_{m+1}]$, we define the final estimator by linearly interpolating between adjacent grid predictions:
    \begin{equation*}
        \hat{g}_1(u, x) \;=\; \hat{f}_{u_m}(x) + \frac{u - u_m}{u_{m+1} - u_m} \big(\hat{f}_{u_{m+1}}(x) - \hat{f}_{u_m}(x)\big).
    \end{equation*}
    Two features of this construction merit emphasis. First, $\hat{g}_1(u,x)$ is piecewise linear in $u$ {by construction} and is therefore deterministically Lipschitz in $u$ on every realization of the data, with Lipschitz constant $M \cdot \max_m \|\hat{f}_{u_{m+1}} - \hat{f}_{u_m}\|_\infty$---a property that no joint ML regression on $(u, X)$ is known to deliver. Second, because ML flexibility in $X$ and local smoothing in $u$ operate in orthogonal directions, the covariate dimension enters only through the ML rate $r_{X,n}$ in \cref{asp:continuous_rates}, while the smoothing in $u$ contributes only a univariate second-order bias of order $h_n^2$.

\paragraph{Boundary Quantile Estimation.} Estimating $Q_{Y,\underline{p} \mid X}(\cdot)$ requires two ingredients: first, isolating observations whose propensity score lies near the lower boundary $\underline{p}(X)$; second, estimating the full conditional quantile function on this localized subsample.

    For the localization step, choose a bandwidth sequence $\delta_n \to 0$. Using the propensity scores $\hat{p}$ and boundary estimates $\widehat{\underline{p}}$ from $I_0$, restrict the nuisance sample to
    \begin{equation*}
        \mathcal{I}_{\delta} = \{i \in I_1 : \hat{p}(Z_i, X_i) \leqslant \widehat{\underline{p}}(X_i) + \delta_n\}.
    \end{equation*}
    Because $\hat{p}$ and $\widehat{\underline{p}}$ are functions of $I_0$ alone, the selection rule is a fixed mapping conditional on $I_0$, so the observations in $\mathcal{I}_{\delta}$ are conditionally i.i.d.\ given $I_0$.

    For the conditional quantile step, we follow the standard grid-based approach to estimating a conditional quantile process: solve a finite set of conditional quantile regressions at prespecified quantile levels and interpolate in between, as employed in \citet{chernozhukov2013inference} and \citet{belloni2019conditional}. Fix a grid resolution $M$ and define the quantile levels $u_m = m/M$ for $m = 1, \dots, M-1$. At each $u_m$, run a conditional $u_m$-quantile regression of $Y$ on $X$ using the localized subsample $\mathcal{I}_{\delta}$ to obtain $\widehat{Q}_{Y,\underline{p} \mid X}(u_m)$; this inner regression may employ any consistent conditional quantile estimator (e.g., linear or series quantile regression, quantile forests, or neural quantile regression). The full continuous estimator $\widehat{Q}_{Y,\underline{p} \mid X}(u)$ on $[0,1]$ is then obtained by linear interpolation between adjacent grid estimates.

\paragraph{Target Evaluation.} Using the main evaluation sample $I_2$, we evaluate the empirical analog of the simplified target functional \eqref{eq:continuous_simplified_target} by plugging in the nuisance estimators from $I_0$ and $I_1$. For each observation $i \in I_2$, we compute the three components:

    The first term (quantile tail integral) is evaluated numerically over the estimated limits:
    \begin{equation*}
        \hat{\psi}_{1}(O_i) = - \int_{\min\{\widehat{\underline{p}}(X_i), \hat{q}(Z_i,X_i)\}}^{\widehat{\underline{p}}(X_i)} \widehat{Q}_{Y,\underline{p} \mid X_i}\left(\frac{v}{\widehat{\underline{p}}(X_i)}\right) \, \mathrm{d}v.
    \end{equation*}

    The second term evaluates the non-parametric regression $\hat{g}_1$ at the truncated limits:
    \begin{equation*}
        \hat{\psi}_{2}(O_i) = \hat{g}_1\big(\min\{\widehat{\overline{p}}(X_i), \hat{q}(Z_i,X_i)\}, X_i\big) - \hat{g}_1\big(\hat{p}(Z_i,X_i), X_i\big).
    \end{equation*}

    The third term (trivial bound) is the straightforward threshold excess:
    \begin{equation*}
        \hat{\psi}_{3}(O_i) = y_{\min} \max\big\{0, \hat{q}(Z_i,X_i) - \widehat{\overline{p}}(X_i)\big\}.
    \end{equation*}
    Averaging these three components over $I_2$ yields the estimator previewed in~(\ref{eq:estimator_continuous_instrument}).

\section{Orthogonal Score Constructions}\label{app:scores}

This section presents the detailed Neyman-orthogonal score formulas for the estimation procedures described in \cref{sec:est_inf}.

\subsection{Continuous Outcome Score Components}\label{app:scores_continuous}

We provide the complete score components for the continuous outcome setting of \cref{subsubsec:continuous_outcome}. Recall the target functional:
\begin{align*}
    \underline{\theta}_{\omega,1} = \mathbb{E}\Bigg[\sum_{k=0}^{K-1}\bigg(\gamma_{\text{full},k}(X) J_{\text{full},k}(X) + \sum_{j=1}^{l_k} \gamma_{j,k}(X) J_{j-1,k}(X)\bigg) + \Delta_K(X) \Bigg].
\end{align*}

\paragraph{Influence function residuals.} The empirical residuals for the conditional expectations, evaluated on observation $O = (X, Z, W, Y)$, are:
\begin{align*}
    IF_{J^+}(O) &= \frac{1}{\pi_{k+1}(X)} \indicator(Z \in S_{k+1}) W \Big( Y \indicator(\nu_{j,k}(X) \leqslant Y < \nu_{j+1,k}(X)) - J_{j,k}^+(X) \Big) \\
    IF_{M_{j,k}^+}(O) &= \frac{1}{\pi_{k+1}(X)} \indicator(Z \in S_{k+1}) W \Big( \indicator(Y \leqslant \nu_{j,k}(X)) - M_{j,k}^+(X) \Big),
\end{align*}
with $IF_{J^-}$ and $IF_{M_{j,k}^-}$ defined analogously over $S_k$.

\paragraph{Riesz representers.} Because $p_k(X)$ and $p_{k+1}(X)$ act as structural interval boundaries, the augmented functional yields non-zero derivatives with respect to these parameters. The exact Riesz representers for the marginal probabilities are:
\begin{align*}
    \alpha_{p_k}(X) &= -J_{j,k}^-(X) - \nu_{j,k}(X)M_{j,k}^-(X) + \nu_{j+1,k}(X)M_{j+1,k}^-(X) + \nu_{j,k}(X) - \nu_{j+1,k}(X), \\
    \alpha_{p_{k+1}}(X) &= J_{j,k}^+(X) + \nu_{j,k}(X)M_{j,k}^+(X) - \nu_{j+1,k}(X)M_{j+1,k}^+(X).
\end{align*}

\paragraph{Score for $J_{j,k}(X)$.} Aggregating the expectation residuals and the structural parameter influences, the complete uncentered orthogonal score for the integral component $J_{j,k}(X)$ is:
\begin{align*}
    \psi_{j,k}^J(O; \eta) &= J_{j,k}(X) + IF_{J^+}(O) - IF_{J^-}(O) \\
    &\quad + \nu_{j,k}(X) \Big( p_{k+1}(X)M_{j,k}^+(X) - p_k(X)M_{j,k}^-(X) - (q_{j,k}(X) - p_k(X)) \\
    &\quad \quad \quad \quad \quad \quad + IF_{M_{j,k}^+}(O) - IF_{M_{j,k}^-}(O) \Big) \\
    &\quad - \nu_{j+1,k}(X) \Big( p_{k+1}(X)M_{j+1,k}^+(X) - p_k(X)M_{j+1,k}^-(X) - (q_{j+1,k}(X) - p_k(X)) \\
    &\quad \quad \quad \quad \quad \quad + IF_{M_{j+1,k}^+}(O) - IF_{M_{j+1,k}^-}(O) \Big) \\
    &\quad + \alpha_{p_k}(X) \frac{1}{\pi_k(X)} \indicator(Z \in S_k) \big(W - p_k(X)\big) + \alpha_{p_{k+1}}(X) \frac{1}{\pi_{k+1}(X)} \indicator(Z \in S_{k+1}) \big(W - p_{k+1}(X)\big) \\
    &\quad - \nu_{j,k}(X) \sum_{z \in T_{j,k}} \frac{\partial \phi}{\partial p}(z,X, p(z,X)) \frac{\indicator(Z=z)}{\pi(z,X)} \big( W - p(z,X) \big) \\
    &\quad + \nu_{j+1,k}(X) \sum_{z \in T_{j+1,k}} \frac{\partial \phi}{\partial p}(z,X, p(z,X)) \frac{\indicator(Z=z)}{\pi(z,X)} \big( W - p(z,X) \big).
\end{align*}

\paragraph{Score for $J_{\mathrm{full},k}(X)$.} The score for the full interval component follows symmetrically but omits the quantile constraint terms:
\begin{align*}
    \psi_{\text{full},k}^J(O; \eta) &= J_{\text{full},k}(X) \\
    &\quad + \frac{1}{\pi_{k+1}(X)} \indicator(Z \in S_{k+1}) W \Big( Y - J_{\text{full},k}^+(X) \Big) - \frac{1}{\pi_k(X)} \indicator(Z \in S_k) W \Big( Y - J_{\text{full},k}^-(X) \Big) \\
    &\quad + J_{\text{full},k}^+(X) \frac{1}{\pi_{k+1}(X)} \indicator(Z \in S_{k+1}) \big(W - p_{k+1}(X)\big) - J_{\text{full},k}^-(X) \frac{1}{\pi_k(X)} \indicator(Z \in S_k) \big(W - p_k(X)\big).
\end{align*}

\paragraph{Product scores.} To account for the estimation of the weighting probabilities $\gamma(X)$, we construct the orthogonal product scores using the chain rule for influence functions. The uncentered product score for the fractional interval component is:
\begin{align*}
    \psi_{j,k}^{\text{prod}}(O; \eta) &= \gamma_{j,k}(X) \psi_{j-1,k}^J(O; \eta) + J_{j-1,k}(X) \left( \indicator\Big(Z \in \bigcup_{l=j}^{l_k} T_{l,k}\Big) - \gamma_{j,k}(X) \right),
\end{align*}
and similarly for the full interval component:
\begin{align*}
    \psi_{\text{full},k}^{\text{prod}}(O; \eta) &= \gamma_{\text{full},k}(X) \psi_{\text{full},k}^J(O; \eta) \\
    &\quad + J_{\text{full},k}(X) \left( \indicator\Big(Z \in \bigcup_{m=k+1}^K \bigcup_{l=0}^{l_m} T_{l,m}\Big) - \indicator\Big(Z \in \bigcup_{m=k+1}^K S_m\Big) - \gamma_{\text{full},k}(X) \right).
\end{align*}

\paragraph{Trivial bound score.} The uncentered orthogonal score for the trivial bound component on the final interval, $\Delta_K(X)$, is:
\begin{align*}
    \psi_{\Delta_K}(O; \eta) &= y_{\min} \indicator\Big(Z \in \bigcup_{j=1}^{l_K} T_{j,K}\Big) \big( q(Z,X) - p_K(X) \big) \\
    &\quad + y_{\min} \indicator\Big(Z \in \bigcup_{j=1}^{l_K} T_{j,K}\Big) \frac{\partial \phi}{\partial p}(Z,X, p(Z,X)) \big( W - p(Z,X) \big) \\
    &\quad - y_{\min} \gamma_K(X) \frac{1}{\pi_K(X)} \indicator(Z \in S_K) \big( W - p_K(X) \big).
\end{align*}

\subsection{Discrete Outcome Score Components}\label{app:scores_discrete}

We provide the complete score components for the binary outcome setting of \cref{subsubsec:discrete_outcome}. The standard IPW residuals for $P_{1,k}(X)$ and $p_k(X)$ are:
\begin{align*}
    R_{P,k}(O) &= \frac{\indicator(Z \in S_k)}{\pi_k(X)} \big(YW - P_{1,k}(X)\big), \\
    R_{p,k}(O) &= \frac{\indicator(Z \in S_k)}{\pi_k(X)} \big(W - p_k(X)\big).
\end{align*}
The key pathwise derivative of $m$ with respect to $J_{\text{full},k}(X)$ is:
\begin{equation*}
    \frac{\partial m}{\partial J_{\text{full},k}}(X) = \gamma_{\text{full},k}(X) + \sum_{j=1}^{l_k} \gamma_{j,k}(X)\, \indicator\!\big(q_{j-1,k}(X) < p_{k+1}(X) - J_{\text{full},k}(X) \leqslant q_{j,k}(X)\big),
\end{equation*}
which equals the weight function evaluated at the lower integration boundary. This has a natural interpretation: an infinitesimal increase in $J_{\text{full},k}$ extends the integration region $[p_{k+1} - J_{\text{full},k},\, p_{k+1}]$ to the left, and the marginal contribution is precisely $\omega$ at the boundary.

\paragraph{Auxiliary indicators.} To state the orthogonal score for $\mathbb{E}[h_{j,k}^{-}(X)]$, we introduce the following notation. Let
\begin{equation*}
    I_{j,k}^{-}(X) \coloneqq \indicator\!\big(q_{j,k}(X) > \max(q_{j-1,k}(X),\, p_{k+1}(X) - J_{\text{full},k}(X))\big)
\end{equation*}
be the \emph{active set indicator} for $h_{j,k}^{-}$, i.e., $I_{j,k}^{-} = \indicator(h_{j,k}^{-} > 0)$. Define the \emph{boundary indicator}
\begin{equation*}
    D_{j,k}^{-}(X) \coloneqq \indicator\!\big(q_{j-1,k}(X) < p_{k+1}(X) - J_{\text{full},k}(X) < q_{j,k}(X)\big),
\end{equation*}
which equals one precisely when the lower integration boundary $p_{k+1} - J_{\text{full},k}$ falls strictly inside the sub-interval $(q_{j-1,k}, q_{j,k})$, and the \emph{dominance indicator}
\begin{equation*}
    E_{j,k}^{-}(X) \coloneqq \indicator\!\big(q_{j-1,k}(X) > p_{k+1}(X) - J_{\text{full},k}(X)\big),
\end{equation*}
which equals one when the inner $\max$ in $h_{j,k}^{-}$ selects $q_{j-1,k}$ rather than $p_{k+1} - J_{\text{full},k}$. Note that $D_{j,k}^{-}$ and $E_{j,k}^{-}$ partition the active set: $I_{j,k}^{-} = D_{j,k}^{-} + E_{j,k}^{-} \cdot \indicator(q_{j,k} > q_{j-1,k})$ (the two indicators are mutually exclusive).

\paragraph{Pathwise derivatives.} The pathwise derivatives of $h_{j,k}^{-}$ with respect to the nuisance parameters are as follows. Since $h_{j,k}^{-}$ depends on $P_{1,m}$ only through $J_{\text{full},k} = P_{1,k+1} - P_{1,k}$:
\begin{equation*}
    \alpha_{P,k+1}^{(j,k)}(X) = D_{j,k}^{-}(X), \qquad \alpha_{P,k}^{(j,k)}(X) = -D_{j,k}^{-}(X), \qquad \alpha_{P,m}^{(j,k)} = 0 \;\text{for } m \notin \{k, k{+}1\}.
\end{equation*}
For the propensity scores, $p_{k+1}$ enters $h_{j,k}^{-}$ directly through the boundary $p_{k+1} - J_{\text{full},k}$, and $p_k$ enters only when $j = 1$ (since $q_{0,k} = p_k$):
\begin{equation*}
    \alpha_{p,k+1}^{(j,k)}(X) = -D_{j,k}^{-}(X), \qquad
    \alpha_{p,k}^{(j,k)}(X) = \begin{cases} -E_{1,k}^{-}(X)\, I_{1,k}^{-}(X) & \text{if } j = 1, \\ 0 & \text{if } j \geqslant 2. \end{cases}
\end{equation*}

\paragraph{Score for $h_{j,k}^{-}(X)$.} The complete orthogonal score for $\mathbb{E}[h_{j,k}^{-}(X)]$ is:
\begin{align*}
    \psi_{j,k}^{h}(O; \eta) &= h_{j,k}^{-}(X) + D_{j,k}^{-}(X)\,\big(R_{P,k+1}(O) - R_{P,k}(O)\big) \\
    &\quad + \alpha_{p,k+1}^{(j,k)}(X)\, R_{p,k+1}(O) + \alpha_{p,k}^{(j,k)}(X)\, R_{p,k}(O) \\
    &\quad + I_{j,k}^{-}(X) \sum_{z \in T_{j,k}} \frac{\partial \phi}{\partial p}(z,X,p(z,X))\, \frac{\indicator(Z = z)}{\pi(z, X)}\big(W - p(z,X)\big) \\
    &\quad - E_{j,k}^{-}(X)\, I_{j,k}^{-}(X) \sum_{z \in T_{j-1,k}} \frac{\partial \phi}{\partial p}(z,X,p(z,X))\, \frac{\indicator(Z = z)}{\pi(z, X)}\big(W - p(z,X)\big).
\end{align*}
The third line corrects for the estimation of the upper boundary $q_{j,k}$ through instruments in $T_{j,k}$, and the fourth line corrects for the lower boundary $q_{j-1,k}$ through instruments in $T_{j-1,k}$; the latter matters only when $q_{j-1,k}$ dominates $p_{k+1} - J_{\text{full},k}$ in the inner $\max$ (i.e., $E_{j,k}^{-} = 1$). When $j = 1$, the lower boundary $q_{0,k} = p_k$ does not depend on individual instruments, so the fourth line vanishes and its correction is absorbed into $\alpha_{p,k}^{(1,k)}$.

\paragraph{Product scores.} The product score for the fractional interval component is:
\begin{align*}
    \psi_{j,k}^{\text{prod}}(O; \eta) &= \gamma_{j,k}(X)\, \psi_{j,k}^{h}(O; \eta) + h_{j,k}^{-}(X) \left(\indicator\Big(Z \in \bigcup_{l=j}^{l_k} T_{l,k}\Big) - \gamma_{j,k}(X)\right).
\end{align*}
For the full-interval component, the orthogonal score for $J_{\text{full},k}(X) = P_{1,k+1}(X) - P_{1,k}(X)$ is simply $\psi_{\text{full},k}^{J}(O; \eta) = J_{\text{full},k}(X) + R_{P,k+1}(O) - R_{P,k}(O)$, and the full-interval product score is:
\begin{align*}
    \psi_{\text{full},k}^{\text{prod}}(O; \eta) &= \gamma_{\text{full},k}(X)\,\psi_{\text{full},k}^{J}(O; \eta) \\
    &\quad + J_{\text{full},k}(X)\left(\indicator\Big(Z \in \bigcup_{m=k+1}^{K} \bigcup_{l=0}^{l_m} T_{l,m}\Big) - \indicator\Big(Z \in \bigcup_{m=k+1}^{K} S_m\Big) - \gamma_{\text{full},k}(X)\right).
\end{align*}
The trivial bound score $\psi_{\Delta_K}(O; \eta)$ is identical to the continuous outcome case (see \cref{app:scores_continuous}), since $\Delta_K(X)$ does not involve the quantile structure.

\section{Proofs in \cref{sec:est_inf}}

\subsection{Discrete Instrument}
\begin{proof}[Proof of \cref{thm:dml_normality}]
Let $\eta_0$ denote the vector of true nuisance parameter values and $\hat{\eta}$ the estimates trained on the auxiliary sample~$I_1$. Write $\theta_0 = \underline{\theta}_{\omega,1}$ for the true parameter value. For a functional $\Phi(\eta)$ and a perturbation direction $\delta$ in the same function space as $\eta$, we write $\partial_\delta \Phi(\eta) \coloneqq \frac{d}{dt}\Phi(\eta + t\,\delta)\big|_{t=0}$ for the Gateaux (directional) derivative. When the perturbation is along a single nuisance component $\eta_\alpha$, we write $\partial_{\eta_\alpha}$ for brevity. Recall the target functional:
\begin{equation}\label{eq:pf_target_recall}
    \theta_0 = \mathbb{E}\Bigg[\sum_{k=0}^{K-1}\bigg(\gamma_{\mathrm{full},k}(X)\,J_{\mathrm{full},k}(X) + \sum_{j=1}^{l_k} \gamma_{j,k}(X)\,J_{j-1,k}(X)\bigg) + \Delta_K(X)\Bigg],
\end{equation}
and the total Neyman-orthogonal score:
\begin{equation}\label{eq:pf_total_score}
    \psi(O;\eta) = \sum_{k=0}^{K-1}\Bigl(\psi_{\mathrm{full},k}^{\mathrm{prod}}(O;\eta) + \sum_{j=1}^{l_k}\psi_{j,k}^{\mathrm{prod}}(O;\eta)\Bigr) + \psi_{\Delta_K}(O;\eta),
\end{equation}
where each product score debiases the corresponding term in the target. Specifically, $\psi_{j,k}^J(O;\eta)$ is the uncentered orthogonal score for the conditional integral component $\mathbb{E}[J_{j,k}(X)]$, and the product score for $\mathbb{E}[\gamma_{j,k}(X)\,J_{j-1,k}(X)]$ is:
\begin{equation}\label{eq:pf_prod_score_recall}
    \psi_{j,k}^{\mathrm{prod}}(O;\eta) = \gamma_{j,k}(X)\,\psi_{j-1,k}^J(O;\eta) + J_{j-1,k}(X)\Bigl(\indicator\bigl(Z \in {\textstyle\bigcup_{l=j}^{l_k}} T_{l,k}\bigr) - \gamma_{j,k}(X)\Bigr),
\end{equation}
and symmetrically for $\psi_{\mathrm{full},k}^{\mathrm{prod}}(O;\eta)$. The proof first verifies that the score satisfies Neyman orthogonality and then applies the standard DML convergence argument with sample splitting.

\medskip
\noindent\textbf{Exact recovery of the level sets.}
Under \cref{asp:prop_gap}, all distinct values in $\{p(z,x), q(z,x)\}_{z \in \mathcal{Z}}$ are separated by at least $c_{\mathrm{gap}} > 0$. Since $\mathcal{Z}$ is finite and $\hat{p}(z,x) \to p(z,x)$ in $L_2$ at rate $o_P(n^{-1/4})$, the marginal averages $\bar{p}(z)$ and $\bar{q}(z)$ converge uniformly to their population limits. For sufficiently large $n$, the clustering procedure therefore recovers $\widehat{S}_k = S_k$ and $\widehat{T}_{j,k} = T_{j,k}$ with probability approaching one. We condition on this event throughout.

\medskip
\noindent\textbf{Verification of the moment condition.}
We verify $\mathbb{E}[\psi(O;\eta_0)] = \theta_0$. All influence function correction terms in the score are inverse-probability-weighted residuals of the general form
\[
    \frac{\indicator(Z \in S_k)}{\pi_k(X)}\,W\bigl(g(Y,X) - \mathbb{E}[g(Y,X) \mid X, Z \in S_k, W\!=\!1]\bigr).
\]
By the tower property, conditioning first on $(X,Z)$: when $Z \in S_k$, we have $p(Z,X) = p_k(X)$, so $\mathbb{E}[W \mid X, Z \in S_k] = p_k(X)$. It follows that
\[
    \mathbb{E}\Big[\frac{\indicator(Z \in S_k)}{\pi_k(X)}\,W\bigl(g - \mathbb{E}[g \mid X, S_k, W\!=\!1]\bigr)\Big] = \mathbb{E}\Big[\frac{\indicator(Z \in S_k)}{\pi_k(X)}\,p_k(X)\cdot 0\Big] = 0.
\]
Applying this identity to each correction term: $\mathbb{E}[IF_{J^\pm}(O)] = 0$ and $\mathbb{E}[IF_{M_{j,k}^\pm}(O)] = 0$ at $\eta_0$. The quantile correction term in $\psi_{j,k}^J$ contains the factor $p_{k+1}(X)M_{j,k}^+(X) - p_k(X)M_{j,k}^-(X) - (q_{j,k}(X) - p_k(X))$, which vanishes at $\eta_0$ by the conditional moment equation
\begin{equation}\label{eq:pf_quantile_moment_dml}
    \begin{split}
        &p_{k+1}(X)\,\mathbb{E}[\indicator(Y \leqslant \nu_{j,k}(X)) \mid X, S_{k+1}, W\!=\!1]\\
        & \quad\quad \quad \quad - p_k(X)\,\mathbb{E}[\indicator(Y \leqslant \nu_{j,k}(X)) \mid X, S_k, W\!=\!1] - (q_{j,k}(X) - p_k(X)) = 0
    \end{split}
\end{equation}
that defines $\nu_{j,k}$. The propensity score correction terms $\frac{\indicator(Z \in S_k)}{\pi_k(X)}(W - p_k(X))$ and the policy correction terms $\frac{\indicator(Z=z)}{\pi(z,X)}(W - p(z,X))$ likewise have zero mean by the same tower-property argument. Consequently,
\begin{align*}
    \mathbb{E}[\psi_{j,k}^J(O;\eta_0)] &= \mathbb{E}[J_{j,k}(X)], & \mathbb{E}[\psi_{\mathrm{full},k}^J(O;\eta_0)] &= \mathbb{E}[J_{\mathrm{full},k}(X)].
\end{align*}
For the product scores, substituting~\eqref{eq:pf_prod_score_recall} and using $\mathbb{E}[\indicator(Z \in \bigcup_{l=j}^{l_k} T_{l,k}) - \gamma_{j,k}(X)] = 0$:
\begin{align*}
    \mathbb{E}[\psi_{j,k}^{\mathrm{prod}}(O;\eta_0)] &= \mathbb{E}\bigl[\gamma_{j,k}(X)\,\underbrace{\mathbb{E}[\psi_{j-1,k}^J(O;\eta_0) \mid X]}_{= J_{j-1,k}(X)}\bigr] + \mathbb{E}\bigl[J_{j-1,k}(X)\,\underbrace{\mathbb{E}[\indicator(\cdots) - \gamma_{j,k}(X) \mid X]}_{= 0}\bigr] \\
    &= \mathbb{E}[\gamma_{j,k}(X)\,J_{j-1,k}(X)],
\end{align*}
and symmetrically $\mathbb{E}[\psi_{\mathrm{full},k}^{\mathrm{prod}}(O;\eta_0)] = \mathbb{E}[\gamma_{\mathrm{full},k}(X)\,J_{\mathrm{full},k}(X)]$. The same reasoning yields $\mathbb{E}[\psi_{\Delta_K}(O;\eta_0)] = \mathbb{E}[\Delta_K(X)]$. Summing over all components and comparing with~\eqref{eq:pf_target_recall} gives $\mathbb{E}[\psi(O;\eta_0)] = \theta_0$.

\medskip
\noindent\textbf{Neyman orthogonality of $\psi_{j,k}^J$.}
We verify that the Gateaux derivative $\partial_t\,\mathbb{E}[\psi_{j,k}^J(O;\eta_0 + t\,\delta)]|_{t=0}$ vanishes for all perturbation directions~$\delta$. We group the nuisance parameters by type and check orthogonality with respect to each group.

\emph{Conditional expectations $(J_{j,k}^\pm, M_{j,k}^\pm)$.}
Consider a perturbation $J_{j,k}^+(X) \to J_{j,k}^+(X) + t\,\delta J^+(X)$. The plug-in term $J_{j,k} = p_{k+1}\,J_{j,k}^+ - p_k\,J_{j,k}^-$ contributes $+\mathbb{E}[p_{k+1}(X)\,\delta J^+(X)]$, while the IPW correction $IF_{J^+}$ contributes $-\mathbb{E}\bigl[\frac{\indicator(Z \in S_{k+1})}{\pi_{k+1}(X)}\,p_{k+1}(X)\,\delta J^+(X)\bigr]$. These cancel because
\[
    \mathbb{E}\!\left[\frac{\indicator(Z \in S_{k+1})}{\pi_{k+1}(X)} \mathrel{\bigg|} X\right] = \frac{\mathbb{P}(Z \in S_{k+1} \mid X)}{\pi_{k+1}(X)} = 1.
\]
The same cancellation applies to perturbations of $J_{j,k}^-$, $M_{j,k}^\pm$, and $M_{j+1,k}^\pm$ via their respective IPW corrections.

\emph{Instrument assignment probabilities $(\pi_k, \pi_{k+1})$.}
Perturbing $\pi_{k+1}(X) \to \pi_{k+1}(X) + t\,\delta\pi(X)$ introduces a factor $-\delta\pi(X)/\pi_{k+1}(X)^2$ in the IPW terms. For any such IPW term $\frac{\indicator(Z \in S_{k+1})}{\pi_{k+1}(X)}\,W\,(g(O) - \bar{g}(X))$, the first-order contribution is
\[
    -\mathbb{E}\!\left[\frac{\indicator(Z \in S_{k+1})}{\pi_{k+1}(X)^2}\,\delta\pi(X)\,W\,(g(O) - \bar{g}(X))\right].
\]
Conditioning on $(X,Z)$ with $Z \in S_{k+1}$, we have $\mathbb{E}[W(g(O) - \bar{g}(X)) \mid X, Z \in S_{k+1}] = p_{k+1}(X) \cdot 0 = 0$ at $\eta_0$, so the Gateaux derivative vanishes. The argument for $\delta\pi_k$ is identical.

\emph{Propensity scores $(p_k, p_{k+1})$.}
A perturbation $p_k(X) \to p_k(X) + t\,\delta p_k(X)$ generates first-order contributions from three sources (all multiplied by $\delta p_k(X)$ inside $\mathbb{E}[\cdot]$):
\begin{enumerate}
    \item[(i)] The plug-in term $J_{j,k} = p_{k+1}\,J_{j,k}^+ - p_k\,J_{j,k}^-$, contributing $-J_{j,k}^-(X)$;
    \item[(ii)] The quantile correction terms $\nu_{j,k}(\cdots) - \nu_{j+1,k}(\cdots)$, contributing $\nu_{j,k}(X)(1 - M_{j,k}^-(X)) - \nu_{j+1,k}(X)(1 - M_{j+1,k}^-(X))$;
    \item[(iii)] The Riesz representer correction $\alpha_{p_k}(X)\frac{\indicator(Z \in S_k)}{\pi_k(X)}(W - p_k(X))$, contributing $-\alpha_{p_k}(X)$.
\end{enumerate}
Substituting the definition
\[
    \alpha_{p_k} = -J_{j,k}^- - \nu_{j,k}\,M_{j,k}^- + \nu_{j+1,k}\,M_{j+1,k}^- + \nu_{j,k} - \nu_{j+1,k}
\]
and collecting terms:
\begin{align*}
    &\bigl(-J_{j,k}^-\bigr) + \bigl(\nu_{j,k} - \nu_{j,k}\,M_{j,k}^- - \nu_{j+1,k} + \nu_{j+1,k}\,M_{j+1,k}^-\bigr) \\
    &\quad - \bigl(-J_{j,k}^- - \nu_{j,k}\,M_{j,k}^- + \nu_{j+1,k}\,M_{j+1,k}^- + \nu_{j,k} - \nu_{j+1,k}\bigr) = 0.
\end{align*}
The verification for $\delta p_{k+1}$, using $\alpha_{p_{k+1}} = J_{j,k}^+ + \nu_{j,k}\,M_{j,k}^+ - \nu_{j+1,k}\,M_{j+1,k}^+$, is identical.

\emph{Quantile thresholds $(\nu_{j,k}, \nu_{j+1,k})$.}
A perturbation $\nu_{j,k}(X) \to \nu_{j,k}(X) + t\,\delta\nu(X)$ affects the score through two channels.

\emph{Channel~(a): indicator boundaries.} The indicator $\indicator(\nu_{j,k}(X) \leqslant Y)$ appears inside $IF_{J^\pm}$ and $IF_{M_{j,k}^\pm}$. Let $f^\pm(\nu) \coloneqq f_{Y \mid X, S_{k\pm1}, W=1}(\nu \mid X)$ denote the conditional densities. Differentiating the expected IPW corrections:
\begin{align*}
    \text{from } \mathbb{E}[IF_{J^+} - IF_{J^-}]&: \quad -\nu_{j,k}(X)\bigl(p_{k+1}(X)\,f^+(\nu_{j,k}) - p_k(X)\,f^-(\nu_{j,k})\bigr)\,\delta\nu(X), \\
    \text{from } \nu_{j,k}\cdot\mathbb{E}[IF_{M_{j,k}^+} - IF_{M_{j,k}^-}]&: \quad +\nu_{j,k}(X)\bigl(p_{k+1}(X)\,f^+(\nu_{j,k}) - p_k(X)\,f^-(\nu_{j,k})\bigr)\,\delta\nu(X),
\end{align*}
where the second line uses $\partial_\nu \mathbb{E}[\indicator(Y \leqslant \nu) \mid X, S_{k\pm1}, W\!=\!1] = f^\pm(\nu)$. These cancel exactly.

\emph{Channel~(b): explicit multiplicative factor.} The term $\nu_{j,k}(X)\bigl(p_{k+1} M_{j,k}^+ - p_k M_{j,k}^- - (q_{j,k} - p_k)\bigr)$ contributes $\delta\nu(X)$ times the conditional moment equation~\eqref{eq:pf_quantile_moment_dml}, which vanishes at $\eta_0$. The dependence of $\alpha_{p_k}$ and $\alpha_{p_{k+1}}$ on $\nu_{j,k}$ is absorbed by the zero-mean residuals: for any $L_2$ function $h(X)$,
\[
    \mathbb{E}\!\left[h(X)\,\frac{\indicator(Z \in S_k)}{\pi_k(X)}\bigl(W - p_k(X)\bigr)\right] = \mathbb{E}\!\left[h(X)\,\frac{\pi_k(X)}{\pi_k(X)}\cdot 0\right] = 0.
\]
The local conditional density assumption ensures all derivatives are well-defined. The treatment of $\delta\nu_{j+1,k}$ is symmetric.

\emph{Alternative policy propensity $(q_{j,k})$.}
Since $q_{j,k}(X) = \phi(z,X,p(z,X))$ for $z \in T_{j,k}$, perturbing $p(z,X) \to p(z,X) + t\,\delta p(z,X)$ for $z \in T_{j,k}$ induces $\delta q_{j,k}(X) = \frac{\partial\phi}{\partial p}(z,X,p(z,X))\,\delta p(z,X)$. The first-order contributions are:
\begin{align*}
    \text{from quantile correction}&: \quad -\mathbb{E}\bigl[\nu_{j,k}(X)\,\delta q_{j,k}(X)\bigr], \\
    \text{from policy IPW correction}&: \quad +\mathbb{E}\!\left[\nu_{j,k}(X)\sum_{z \in T_{j,k}}\frac{\partial\phi}{\partial p}\frac{\indicator(Z\!=\!z)}{\pi(z,X)}\bigl(W - p(z,X)\bigr)\right].
\end{align*}
The second term evaluates to $+\mathbb{E}[\nu_{j,k}(X)\,\delta q_{j,k}(X)]$ by the tower property: $\mathbb{E}[\frac{\indicator(Z=z)}{\pi(z,X)}\,(W - p(z,X)) \mid X] = \frac{\pi(z,X)}{\pi(z,X)}\cdot 0 = 0$, while $\mathbb{E}[\frac{\indicator(Z=z)}{\pi(z,X)}\,W \mid X] = p(z,X)$. The two contributions cancel. The treatment of $q_{j+1,k}$ is symmetric.

This completes the orthogonality verification for $\psi_{j,k}^J$.

\medskip
\noindent\textbf{Neyman orthogonality of the remaining components.}
The score $\psi_{\mathrm{full},k}^J$ is a simplified version of $\psi_{j,k}^J$ without quantile threshold terms; its orthogonality follows by restricting the above arguments to the nuisance parameters $(J_{\mathrm{full},k}^\pm, p_k, p_{k+1}, \pi_k, \pi_{k+1})$.

For the product scores, we invoke the chain rule for influence functions. Since $\psi_{j-1,k}^J$ is an orthogonal score for $\mathbb{E}[J_{j-1,k}(X)]$ and $\indicator(Z \in \bigcup_{l=j}^{l_k} T_{l,k}) - \gamma_{j,k}(X)$ is the influence function for $\gamma_{j,k}(X) = \mathbb{P}(Z \in \bigcup_{l=j}^{l_k} T_{l,k} \mid X)$, the product score~\eqref{eq:pf_prod_score_recall}
\[
    \psi_{j,k}^{\mathrm{prod}} = \gamma_{j,k}(X)\,\psi_{j-1,k}^J(O;\eta) + J_{j-1,k}(X)\bigl(\indicator(Z \in {\textstyle\bigcup_{l=j}^{l_k}} T_{l,k}) - \gamma_{j,k}(X)\bigr)
\]
is orthogonal for the product functional $\mathbb{E}[\gamma_{j,k}(X)\,J_{j-1,k}(X)]$. To see this, perturb any nuisance component $\eta_\alpha$ and apply the Leibniz rule:
\[
    \partial_{\eta_\alpha}\mathbb{E}[\psi_{j,k}^{\mathrm{prod}}] = \gamma_{j,k}\,\underbrace{\partial_{\eta_\alpha}\mathbb{E}[\psi_{j-1,k}^J]}_{=\,0\text{ (orthogonality of }\psi_{j-1,k}^J)} + J_{j-1,k}\,\underbrace{\partial_{\eta_\alpha}\mathbb{E}[\indicator(\cdots) - \gamma_{j,k}]}_{=\,0\text{ (trivial orthogonality)}} = 0.
\]
The same argument applies to $\psi_{\mathrm{full},k}^{\mathrm{prod}}$.

For $\psi_{\Delta_K}$: the target $\Delta_K = y_{\min}\,\mathbb{E}[\indicator(Z \in \bigcup T_{j,K})(q(Z,X) - p_K(X))]$ depends on the propensity scores through $q(z,x) = \phi(z,x,p(z,x))$ and $p_K(x)$. The IPW correction for $p_K$ uses $\frac{\indicator(Z \in S_K)}{\pi_K(X)}(W - p_K(X))$, and for $q$ uses $\frac{\partial\phi}{\partial p}\frac{\indicator(Z=z)}{\pi(z,X)}(W - p(z,X))$. Orthogonality with respect to these nuisance parameters follows by the same tower-property cancellations as above.

\medskip
\noindent\textbf{Second-order remainder bound.}
By Neyman orthogonality, the functional $\eta \mapsto \mathbb{E}[\psi(O;\eta)]$ has a vanishing gradient at $\eta_0$, so the Taylor expansion yields
\begin{equation}\label{eq:pf_dml_bias}
    \big|\mathbb{E}[\psi(O;\hat{\eta})] - \theta_0\big| \leqslant C\,\|\hat{\eta} - \eta_0\|_{L_2}^2.
\end{equation}
The smoothness of $\eta \mapsto \mathbb{E}[\psi(O;\eta)]$ is justified as follows. The non-smooth indicator functions $\indicator(\nu_{j,k}(X) \leqslant Y)$ in the score are smoothed upon taking expectations over $Y$, producing conditional CDFs $F_{Y \mid X,Z,W}(\nu_{j,k}(X))$ that are differentiable in $\nu_{j,k}$ with bounded derivatives by the local conditional density assumption. The IPW denominators $1/\pi_k(X)$ are smooth on the region $\pi_k \geqslant c_\pi > 0$ guaranteed by \cref{asp:overlap}. All remaining terms in the score are polynomial in the nuisance parameters, hence smooth. The constant $C$ in~\eqref{eq:pf_dml_bias} is uniformly controlled by: the boundedness of $Y$ (\cref{asp:bounded_y}), the overlap condition (\cref{asp:overlap}), the boundedness of $\frac{\partial\phi}{\partial p}$ (\cref{asp:phi_smooth}), and the conditional densities at the quantile thresholds. 

Since $\partial_\delta\,\mathbb{E}[\psi(O;\eta_0)] = 0$ for all directions $\delta$ (Neyman orthogonality), the Taylor expansion of $\eta \mapsto \mathbb{E}[\psi(O;\eta)]$ around $\eta_0$ starts at second order. The second-order remainder is a sum of bilinear forms in pairs of nuisance estimation errors $(\hat{\eta}_\alpha - \eta_{0,\alpha})(\hat{\eta}_\beta - \eta_{0,\beta})$, with coefficients uniformly bounded under the regularity conditions (bounded $Y$, overlap, conditional densities at the quantile thresholds, smooth $\phi$). Since $\|\hat{\eta}_\alpha - \eta_{0,\alpha}\|_{L_2} = o_P(n^{-1/4})$ for each component $\alpha$, the Cauchy--Schwarz inequality gives $|\mathbb{E}[\psi(O;\hat{\eta})] - \theta_0| = o_P(n^{-1/2})$.

\medskip
\noindent\textbf{Asymptotic normality.}
By the sample-splitting construction, $\hat{\eta}$ (trained on $I_1$) is independent of $\{O_i\}_{i \in I_2}$. Decompose:
\begin{equation}\label{eq:pf_dml_decomp}
    \sqrt{n/2}\,\bigl(\hat{\underline{\theta}}_{\omega,1} - \theta_0\bigr) = \underbrace{\frac{1}{\sqrt{n/2}}\sum_{i \in I_2}\bigl(\psi(O_i;\eta_0) - \theta_0\bigr)}_{\mathrm{(I)}} + \underbrace{\frac{1}{\sqrt{n/2}}\sum_{i \in I_2}\bigl(\psi(O_i;\hat{\eta}) - \psi(O_i;\eta_0)\bigr)}_{\mathrm{(II)}}.
\end{equation}
For term~(I): the summands are i.i.d.\ with mean zero and finite variance $\sigma^2 = \mathrm{Var}(\psi(O;\eta_0))$. Finiteness follows from \cref{asp:bounded_y} (bounded outcomes), \cref{asp:overlap} (bounded IPW weights), and boundedness of $\omega$ and $\frac{\partial\phi}{\partial p}$. By the Lindeberg--L\'{e}vy central limit theorem, $\mathrm{(I)} \xrightarrow{d} \mathcal{N}(0,\sigma^2)$.

For term~(II): condition on $I_1$ (which fixes $\hat{\eta}$) and write
\[
    \mathrm{(II)} = \sqrt{n/2}\cdot B_n + \frac{1}{\sqrt{n/2}}\sum_{i \in I_2}\epsilon_i,
\]
where $B_n = \mathbb{E}[\psi(O;\hat{\eta}) - \psi(O;\eta_0) \mid I_1] = \mathbb{E}[\psi(O;\hat{\eta})] - \theta_0$ and $\epsilon_i = \bigl(\psi(O_i;\hat{\eta}) - \psi(O_i;\eta_0)\bigr) - B_n$ are conditionally i.i.d.\ zero-mean random variables. By~\eqref{eq:pf_dml_bias}, $\sqrt{n/2}\,|B_n| = O_P\bigl(\sqrt{n}\,\|\hat{\eta} - \eta_0\|_{L_2}^2\bigr) = o_P(1)$. For the martingale term, the conditional variance satisfies
\[
    \mathrm{Var}(\epsilon_1 \mid I_1) \leqslant \mathbb{E}\bigl[(\psi(O;\hat{\eta}) - \psi(O;\eta_0))^2 \mid I_1\bigr].
\]
The score $\psi(O;\eta)$ is Lipschitz in the smooth nuisance components (conditional expectations, propensity scores, instrument probabilities) under \cref{asp:overlap}. For the non-smooth indicator terms, $|\indicator(\hat{\nu}_{j,k} \leqslant Y) - \indicator(\nu_{j,k}^0 \leqslant Y)|$ is nonzero only when $Y$ lies between $\hat{\nu}_{j,k}(X)$ and $\nu_{j,k}^0(X)$, which by the local conditional density assumption occurs with probability $O(|\hat{\nu}_{j,k} - \nu_{j,k}^0|)$. Consequently, $\mathrm{Var}(\epsilon_1 \mid I_1) = O_P(\|\hat{\eta} - \eta_0\|_{L_2}) = o_P(1)$, and $\frac{1}{\sqrt{n/2}}\sum_{i \in I_2}\epsilon_i = o_P(1)$ by Chebyshev's inequality. Combining, $\mathrm{(II)} = o_P(1)$.

By Slutsky's theorem applied to~\eqref{eq:pf_dml_decomp}, we conclude
\[
    \sqrt{n/2}\,\bigl(\hat{\underline{\theta}}_{\omega,1} - \underline{\theta}_{\omega,1}\bigr) \xrightarrow{d} \mathcal{N}(0,\,\sigma^2).
\]

\medskip
\noindent\textbf{Consistency of the variance estimator.}
Since $\hat{\eta} \xrightarrow{P} \eta_0$ and the score $\psi(O;\eta)$ is continuous in $\eta$ for $\mathbb{P}_{\mathrm{obs}}$-a.e.\ $O$ (the set of discontinuity, where $Y$ coincides with a quantile threshold, has measure zero since the conditional density exists at the quantile thresholds), the continuous mapping theorem together with the uniform law of large numbers gives $\hat{\sigma}^2 = \frac{1}{|I_2|}\sum_{i \in I_2}\bigl(\hat{\psi}(O_i;\hat{\eta}) - \hat{\underline{\theta}}_{\omega,1}\bigr)^2 \xrightarrow{P} \sigma^2$, enabling the construction of asymptotically valid confidence intervals.
\end{proof}

\begin{proof}[Proof of \cref{thm:dml_discrete_outcome}]
Let $\eta_0$ denote the true nuisance parameter values and $\hat{\eta}$ the estimates trained on $I_1$. Write $\theta_0 = \underline{\theta}_{\omega,1}$. Recall the target functional \eqref{eq:binary_lower_bound} and the orthogonal score \eqref{eq:binary_score}, which decomposes as $\psi = \sum_{k}(\psi_{\mathrm{full},k}^{\mathrm{prod}} + \sum_{j}\psi_{j,k}^{\mathrm{prod}}) + \psi_{\Delta_K}$. The proof verifies Neyman orthogonality at the component level and applies the standard DML convergence argument.

\medskip
\noindent\textbf{Exact recovery of the level sets.}
The argument is identical to the continuous outcome case. Under \cref{asp:prop_gap}, since $\hat{p}(z,x) \to p(z,x)$ in $L_2$ at rate $o_P(n^{-1/4})$ and $\mathcal{Z}$ is finite, the clustering procedure recovers $\widehat{S}_k = S_k$ and $\widehat{T}_{j,k} = T_{j,k}$ with probability approaching one. We condition on this event throughout.

\medskip
\noindent\textbf{Verification of the moment condition.}
We verify $\mathbb{E}[\psi(O;\eta_0)] = \theta_0$. All IPW correction terms in $\psi_{j,k}^h$ and $\psi_{\mathrm{full},k}^J$ have the form $\frac{\indicator(Z \in A)}{\mathbb{P}(Z \in A \mid X)}(g(O) - \mathbb{E}[g(O) \mid Z \in A, X])$ for some event $A$, which has conditional expectation zero by the tower property. Hence $\mathbb{E}[\psi_{j,k}^h(O;\eta_0)] = \mathbb{E}[h_{j,k}^-(X)]$ and $\mathbb{E}[\psi_{\mathrm{full},k}^J(O;\eta_0)] = \mathbb{E}[J_{\mathrm{full},k}(X)]$. Since the product scores satisfy $\mathbb{E}[\psi_{j,k}^{\mathrm{prod}}] = \gamma_{j,k}\,\mathbb{E}[\psi_{j,k}^h] + \mathbb{E}[h_{j,k}^-]\,\mathbb{E}[\indicator(\cdots) - \gamma_{j,k}] = \gamma_{j,k}\,\mathbb{E}[h_{j,k}^-]$ (and similarly $\mathbb{E}[\psi_{\mathrm{full},k}^{\mathrm{prod}}] = \gamma_{\mathrm{full},k}\,\mathbb{E}[J_{\mathrm{full},k}]$), summing gives $\mathbb{E}[\psi(O;\eta_0)] = \mathbb{E}[m(X;\eta_0)] = \theta_0$.

\medskip
\noindent\textbf{Neyman orthogonality.}
We must show $\partial_{\eta_\alpha}\mathbb{E}[\psi(O;\eta)]\big|_{\eta=\eta_0} = 0$ for each nuisance component $\eta_\alpha$. By the product score structure, it suffices to verify orthogonality of the component scores $\psi_{j,k}^h$ and $\psi_{\mathrm{full},k}^J$ with respect to $(P_{1,m}, p_m, p(z,X))$, and then show the product scores inherit orthogonality with respect to $\gamma$.

\emph{Orthogonality of $\psi_{j,k}^h$ with respect to $P_{1,m}$.}
Consider a perturbation $P_{1,m} \to P_{1,m} + t\,\delta P_{1,m}(X)$. Since $h_{j,k}^-$ depends on $P_{1,m}$ only through $J_{\mathrm{full},k} = P_{1,k+1} - P_{1,k}$, the pathwise derivative is
\[
    \partial_t \mathbb{E}[h_{j,k}^-] \big|_{t=0} = \mathbb{E}\!\left[\alpha_{P,m}^{(j,k)}(X)\,\delta P_{1,m}(X)\right],
\]
where $\alpha_{P,k+1}^{(j,k)} = D_{j,k}^-$, $\alpha_{P,k}^{(j,k)} = -D_{j,k}^-$, and $\alpha_{P,m}^{(j,k)} = 0$ for $m \notin \{k,k{+}1\}$. Under \cref{asp:discrete_outcome_gap}, $D_{j,k}^- = \indicator(q_{j-1,k} < p_{k+1} - J_{\mathrm{full},k} < q_{j,k})$ is well-defined. The IPW correction $D_{j,k}^-(R_{P,k+1} - R_{P,k})$ in $\psi_{j,k}^h$ contributes
\[
    \partial_t \mathbb{E}\!\left[D_{j,k}^-\!\left(\frac{\indicator(Z \in S_{k+1})}{\pi_{k+1}}\bigl(YW - P_{1,k+1}\bigr) - \frac{\indicator(Z \in S_k)}{\pi_k}\bigl(YW - P_{1,k}\bigr)\right)\right]\bigg|_{t=0} = -\mathbb{E}\!\left[\alpha_{P,m}^{(j,k)}\,\delta P_{1,m}\right],
\]
using $\mathbb{E}[\frac{\indicator(Z \in S_m)}{\pi_m} \mid X] = 1$, so the two contributions cancel. The same argument gives $\partial_{P_{1,m}}\mathbb{E}[\psi_{\mathrm{full},k}^J] = 0$ (with $\alpha_{P,k+1}^{(\mathrm{full},k)} = 1$ and $\alpha_{P,k}^{(\mathrm{full},k)} = -1$).

\emph{Orthogonality of $\psi_{j,k}^h$ with respect to $p_m$ and $p(z,X)$.}
The propensity scores enter $h_{j,k}^-$ through three channels:

\begin{enumerate}
    \item the integration boundary $p_{k+1} - J_{\mathrm{full},k}$, corrected by $\alpha_{p,k+1}^{(j,k)} R_{p,k+1}$
    \item the lower level-set boundary $q_{0,k} = p_k$ when $j = 1$, corrected by $\alpha_{p,k}^{(1,k)} R_{p,k}$
    \item the alternative policy values $q_{j,k}$ and $q_{j-1,k}$ through the mapping $\phi(z,X,p(z,X))$, corrected by the policy terms involving $T_{j,k}$ and $T_{j-1,k}$.
\end{enumerate}

For channel (a), a perturbation $p_{k+1} \to p_{k+1} + t\,\delta p_{k+1}(X)$ shifts the boundary by $\delta p_{k+1}$, giving $\partial_t h_{j,k}^- = -D_{j,k}^-\,\delta p_{k+1}$. The correction $\alpha_{p,k+1}^{(j,k)} R_{p,k+1} = -D_{j,k}^- \cdot \frac{\indicator(Z \in S_{k+1})}{\pi_{k+1}}(W - p_{k+1})$ contributes $\partial_t \mathbb{E}[\alpha_{p,k+1}^{(j,k)} R_{p,k+1}] = D_{j,k}^-\,\delta p_{k+1}$, canceling the plug-in derivative.

For channel (b), when $j = 1$ and $E_{1,k}^- = 1$ (i.e., $p_k > p_{k+1} - J_{\mathrm{full},k}$), a perturbation of $p_k$ shifts $q_{0,k} = p_k$, giving $\partial_t h_{1,k}^- = -E_{1,k}^- I_{1,k}^-\,\delta p_k$. The correction $\alpha_{p,k}^{(1,k)} R_{p,k}$ contributes the matching cancellation.

For channel (c), a perturbation of $p(z,X)$ for $z \in T_{j,k}$ shifts $q_{j,k}$ by $\frac{\partial \phi}{\partial p}\,\delta p(z,X)$, contributing $I_{j,k}^-\,\frac{\partial \phi}{\partial p}\,\delta p(z,X)$ to $\partial_t h_{j,k}^-$. The policy correction $I_{j,k}^- \sum_{z' \in T_{j,k}} \frac{\partial \phi}{\partial p} \frac{\indicator(Z = z')}{\pi(z',X)}(W - p(z',X))$ cancels this via $\mathbb{E}[\frac{\indicator(Z = z)}{\pi(z,X)} \mid X] = 1$. Symmetrically, when $E_{j,k}^- = 1$, the correction for $q_{j-1,k}$ through $T_{j-1,k}$ cancels the derivative $-E_{j,k}^- I_{j,k}^-\,\frac{\partial \phi}{\partial p}\,\delta p(z,X)$ for $z \in T_{j-1,k}$.

\emph{Inheritance by product scores.}
Since $\psi_{j,k}^{\mathrm{prod}} = \gamma_{j,k}\,\psi_{j,k}^h + h_{j,k}^-(\indicator(\cdots) - \gamma_{j,k})$ and $\mathbb{E}[\indicator(\cdots) - \gamma_{j,k} \mid X] = 0$ at $\eta_0$, any perturbation of $(P_{1,m}, p_m, p(z,X))$ yields
\[
    \partial_\eta \mathbb{E}[\psi_{j,k}^{\mathrm{prod}}] = \gamma_{j,k}\,\underbrace{\partial_\eta \mathbb{E}[\psi_{j,k}^h]}_{=\,0} + \partial_\eta \mathbb{E}[h_{j,k}^-] \cdot \underbrace{\mathbb{E}[\indicator(\cdots) - \gamma_{j,k} \mid X]}_{=\,0} = 0.
\]
The same holds for $\psi_{\mathrm{full},k}^{\mathrm{prod}}$ and $\psi_{\Delta_K}$.

\emph{Perturbation of $\gamma_{j,k}$ and $\gamma_{\mathrm{full},k}$.}
The product score structure ensures orthogonality. By the Leibniz rule:
\[
    \partial_{\gamma_{j,k}}\mathbb{E}[\psi_{j,k}^{\mathrm{prod}}] = \gamma_{j,k}\,\underbrace{\partial_{\gamma_{j,k}}\mathbb{E}[\psi_{j,k}^h]}_{=\,0} + \mathbb{E}[h_{j,k}^-]\,\underbrace{\partial_{\gamma_{j,k}}\mathbb{E}[\indicator(\cdots) - \gamma_{j,k}]}_{=\,0} = 0.
\]
The same argument applies to $\gamma_{\mathrm{full},k}$ and its corresponding product score.

\emph{Rearrangement.}
Under \cref{asp:discrete_outcome_gap}, the true values satisfy the monotonicity constraint $P_{1,0}(x) \leqslant \cdots \leqslant P_{1,K}(x)$ for all $x$ (since $J_{\mathrm{full},k}(x) \geqslant 0$). The rearrangement operator projects onto this monotonicity constraint pointwise in $x$. By the classical rearrangement inequality, for each $x$ the rearranged vector $(\hat{P}_{1,0}^*(x), \ldots, \hat{P}_{1,K}^*(x))$ is at least as close to the true monotone vector $(P_{1,0}(x), \ldots, P_{1,K}(x))$ as the original $(\hat{P}_{1,0}(x), \ldots, \hat{P}_{1,K}(x))$:
\[
    \sum_{k=0}^{K} \bigl|\hat{P}_{1,k}^*(x) - P_{1,k}(x)\bigr|^2 \leqslant \sum_{k=0}^{K} \bigl|\hat{P}_{1,k}(x) - P_{1,k}(x)\bigr|^2 \quad \text{for all } x.
\]
Integrating over $x$, rearrangement can only decrease the $L_2$ estimation error: $\|\hat{P}_{1,k}^* - P_{1,k}\|_{L_2} \leqslant \|\hat{P}_{1,k} - P_{1,k}\|_{L_2} = o_P(n^{-1/4})$. The subsequent clipping $\hat{J}_{\mathrm{full},k}^* \leftarrow \min(\hat{J}_{\mathrm{full},k}^*,\, \Delta p_k)$ is a projection onto $[0, \Delta p_k]$; since the true $J_{\mathrm{full},k} \in [0, \Delta p_k]$, this projection can only further reduce the estimation error by the same contraction argument. Therefore, all convergence rate assumptions on the nuisance estimators are preserved after rearrangement and clipping.

\medskip
\noindent\textbf{Second-order remainder bound.}
Let $\Delta\eta = \hat{\eta} - \eta_0$ and write the linearization error as
\[
    R_n = \frac{1}{|I_2|}\sum_{i \in I_2}\bigl[m(X_i;\hat{\eta}) - m(X_i;\eta_0) - \nabla_\eta m(X_i;\eta_0)^\top \Delta\eta(X_i)\bigr].
\]
The functional $m(X;\eta)$ is piecewise linear in $(J_{\mathrm{full},k}, p_k, q_{j,k})$ through the $\max$ and $\min$ operators, and bilinear in $(\gamma_{j,k}, h_{j,k}^-)$. Under \cref{asp:discrete_outcome_gap}, the $\max$ and $\min$ functions are locally linear at the true parameter values. The only second-order terms arise from the bilinear products $\gamma_{j,k}\,h_{j,k}^-$ and $\gamma_{\mathrm{full},k}\,J_{\mathrm{full},k}$, yielding a remainder of order
\[
    R_n = O_P\!\left(\|\Delta\gamma\|_{L_2}\,\|\Delta h\|_{L_2} + \|\Delta\gamma_{\mathrm{full}}\|_{L_2}\,\|\Delta J_{\mathrm{full}}\|_{L_2}\right) = O_P(n^{-1/4} \cdot n^{-1/4}) = o_P(n^{-1/2}).
\]

\medskip
\noindent\textbf{Asymptotic normality.}
By the standard DML argument \citep{chernozhukov2018double}, the moment condition, Neyman orthogonality, and the $o_P(n^{-1/2})$ remainder bound together imply:
\[
    \sqrt{n/2}\,(\hat{\underline{\theta}}_{\omega,1} - \theta_0) = \frac{1}{\sqrt{n/2}}\sum_{i \in I_2}\bigl(\psi(O_i;\eta_0) - \theta_0\bigr) + o_P(1) \xrightarrow{d} \mathcal{N}(0,\,\sigma^2),
\]
where $\sigma^2 = \mathrm{Var}(\psi(O;\eta_0))$.

\medskip
\noindent\textbf{Consistency of the variance estimator.}
Since $\hat{\eta} \xrightarrow{P} \eta_0$ and, under \cref{asp:discrete_outcome_gap}, the score $\psi(O;\eta)$ is continuous in $\eta$ for $\mathbb{P}_{\mathrm{obs}}$-a.e.\ $O$ (the set of discontinuity corresponds to $p_{k+1}(X) - J_{\mathrm{full},k}(X) = q_{j,k}(X)$, which has probability zero under the gap assumption), the continuous mapping theorem and the uniform law of large numbers give $\hat{\sigma}^2 \xrightarrow{P} \sigma^2$.
\end{proof}

\subsection{Continuous Instrument}

\begin{proof}[Proof of \cref{thm:continuous_rate}]

Recall the target functional from \eqref{eq:continuous_simplified_target}. Define $g_1(u,x) \coloneqq \mathbb{E}_{\mathrm{obs}}[YW \mid p(Z,X)=u, X=x]$ and decompose the target as $\underline{\theta}_{\omega, 1} = \mathbb{E}_{Z,X}[\psi_1(O) + \psi_2(O) + \psi_3(O)]$, where
\begin{align*}
    \psi_1(O) &\coloneqq - \int_{\min\{\underline{p}(X), q(Z,X)\}}^{\underline{p}(X)} Q_{Y,\underline{p} \mid X}\!\left(\frac{v}{\underline{p}(X)}\right) \mathrm{d}v, \\
    \psi_2(O) &\coloneqq g_1\big(\min\{\overline{p}(X), q(Z,X)\}, X\big) - g_1\big(p(Z,X), X\big), \\
    \psi_3(O) &\coloneqq y_{\min} \max\{0, q(Z,X) - \overline{p}(X)\}.
\end{align*}
The plug-in estimator is $\hat{\underline{\theta}}_{\omega, 1} = \frac{1}{|I_2|} \sum_{i \in I_2} [\hat{\psi}_1(O_i) + \hat{\psi}_2(O_i) + \hat{\psi}_3(O_i)]$, where each $\hat{\psi}_j$ substitutes $(\hat{p}, \hat{q}, \widehat{\underline{p}}, \widehat{\overline{p}}, \hat{g}_1, \widehat{Q}_{Y,\underline{p}\mid X})$ for the true quantities:
\begin{align*}
    \hat{\psi}_1(O) &= - \int_{\min\{\widehat{\underline{p}}(X), \hat{q}(Z,X)\}}^{\widehat{\underline{p}}(X)} \widehat{Q}_{Y,\underline{p} \mid X}\!\left(\frac{v}{\widehat{\underline{p}}(X)}\right) \mathrm{d}v, \\
    \hat{\psi}_2(O) &= \hat{g}_1\big(\min\{\widehat{\overline{p}}(X), \hat{q}(Z,X)\}, X\big) - \hat{g}_1\big(\hat{p}(Z,X), X\big), \\
    \hat{\psi}_3(O) &= y_{\min} \max\{0, \hat{q}(Z,X) - \widehat{\overline{p}}(X)\}.
\end{align*}

\medskip
\noindent\textbf{Main decomposition.}
By the sample-splitting construction, $\hat{\psi}$ is fixed conditional on $I_0 \cup I_1$ and the observations $\{O_i\}_{i \in I_2}$ are i.i.d.\ Decompose the estimation error as
\begin{align*}
    \hat{\underline{\theta}}_{\omega,1} - \underline{\theta}_{\omega,1} &= \underbrace{\frac{1}{|I_2|}\sum_{i \in I_2}\bigl(\hat{\psi}(O_i) - \psi(O_i)\bigr)}_{\text{nuisance plug-in error}} + \underbrace{\frac{1}{|I_2|}\sum_{i \in I_2}\psi(O_i) - \mathbb{E}[\psi(O)]}_{\text{sampling error}}.
\end{align*}
Since $Y$ is bounded (\cref{asp:bounded_y}) and $\omega$ is bounded, $\psi(O)$ has finite variance. By the central limit theorem, the sampling error is $O_P(n^{-1/2})$, which is dominated by the stated rate. By the law of large numbers conditional on $I_0 \cup I_1$, the first term satisfies
\[
    \frac{1}{|I_2|}\sum_{i \in I_2}\bigl(\hat{\psi}_j(O_i) - \psi_j(O_i)\bigr) = \mathbb{E}\bigl[\hat{\psi}_j(O) - \psi_j(O) \mid I_0, I_1\bigr] + O_P(n^{-1/2}),
\]
for each $j \in \{1,2,3\}$. It therefore suffices to bound $\mathbb{E}[|\hat{\psi}_j(O) - \psi_j(O)| \mid I_0, I_1]$ for each component.

\medskip
\noindent\textbf{Preliminary: propensity score plug-in errors.}
By \cref{asp:phi_smooth}, $|\hat{q}(z,x) - q(z,x)| = |\phi(z,x,\hat{p}) - \phi(z,x,p)| \leqslant C_\phi \|\hat{p} - p\|_\infty = O_P(r_{p,n})$, where $C_\phi$ is the uniform bound on $\frac{\partial\phi}{\partial p}$. Similarly, $|\widehat{\overline{p}}(x) - \overline{p}(x)| \leqslant \|\hat{p}-p\|_\infty$ and $|\widehat{\underline{p}}(x) - \underline{p}(x)| \leqslant \|\hat{p}-p\|_\infty$, so all estimated propensity-derived quantities are uniformly $O_P(r_{p,n})$-close to their population counterparts.

\medskip
\noindent\textbf{Analysis of the trivial bound $\hat{\psi}_3$.}
Since $|\max\{0,a\} - \max\{0,b\}| \leqslant |a-b|$, we have
\[
    |\hat{\psi}_3(O) - \psi_3(O)| \leqslant |y_{\min}|\bigl(|\hat{q}(Z,X) - q(Z,X)| + |\widehat{\overline{p}}(X) - \overline{p}(X)|\bigr) = O_P(r_{p,n}).
\]

\medskip
\noindent\textbf{Analysis of the conditional expectation term $\hat{\psi}_2$.}
Define $u_1 = \min\{\overline{p}(X), q(Z,X)\}$, $\hat{u}_1 = \min\{\widehat{\overline{p}}(X), \hat{q}(Z,X)\}$, $u_2 = p(Z,X)$, and $\hat{u}_2 = \hat{p}(Z,X)$. The error decomposes as
\begin{align*}
    |\hat{\psi}_2(O) - \psi_2(O)| &\leqslant \sum_{l=1}^2 |\hat{g}_1(\hat{u}_l, X) - g_1(u_l, X)| \\
    &\leqslant \sum_{l=1}^2 \Big(\underbrace{|\hat{g}_1(\hat{u}_l, X) - \hat{g}_1(u_l, X)|}_{\text{evaluation plug-in}} + \underbrace{|\hat{g}_1(u_l, X) - g_1(u_l, X)|}_{\text{estimator accuracy}}\Big).
\end{align*}

\emph{Stochastic Lipschitz property of $\hat{g}_1$.}
Since $\hat{g}_1(u,x)$ is piecewise linear in $u$ on $\{u_m = m/M\}_{m=0}^M$, its Lipschitz constant is $L_{\hat{g}}(x) = \max_{0 \leqslant m < M} M|\hat{f}_{u_{m+1}}(x) - \hat{f}_{u_m}(x)|$. By the triangle inequality,
\[
    M|\hat{f}_{u_{m+1}}(x) - \hat{f}_{u_m}(x)| \leqslant \underbrace{M|g_1(u_{m+1},x) - g_1(u_m,x)|}_{\leqslant\, \|\partial_u g_1\|_\infty \text{ by the mean value theorem}} + 2M \cdot \max_{m'} \|\hat{f}_{u_{m'}} - g_1(u_{m'}, \cdot)\|_\infty,
\]
where the first term is bounded by the regularity of $g_1$ (\cref{asp:continuous_rates}(2)). The grid-point error established below satisfies $\max_{m'}\|\hat{f}_{u_{m'}} - g_1(u_{m'}, \cdot)\|_\infty = O_P(r_{p,n}/h_n + r_{X,n} + h_n^2)$, so
\[
    L_{\hat{g}}(x) \leqslant \|\partial_u g_1\|_\infty + 2M \cdot O_P\!\left(\frac{r_{p,n}}{h_n} + r_{X,n} + h_n^2\right) = O_P(1),
\]
where the last step uses $M(r_{p,n}/h_n + r_{X,n} + h_n^2) = o(1)$ under $M \gtrsim \sqrt{nh_n}$ and the nuisance rates being $o(M^{-1})$. Consequently,
\[
    |\hat{g}_1(\hat{u}_l, X) - \hat{g}_1(u_l, X)| \leqslant L_{\hat{g}}(X) \cdot |\hat{u}_l - u_l| = O_P(r_{p,n}).
\]

\emph{Convergence of the localized ML estimator.}
We bound the pointwise error $|\hat{g}_1(u,x) - g_1(u,x)|$ at any fixed evaluation point $u \in [\underline{p}(x), \overline{p}(x)]$. Let $u_m$ be the nearest grid point. Define
\[
    \tilde{f}_{u_m} = \arg\min_{f \in \mathcal{F}} \sum_{i \in I_1} K\!\left(\frac{u_m - p(Z_i,X_i)}{h_n}\right)(Y_i\indicator(W_i\!=\!1) - f(X_i))^2
\]
as the oracle ML estimator using the true propensity scores. By the triangle inequality:
\[
    \|\hat{f}_{u_m} - g_1(u_m, \cdot)\|_2 \leqslant \underbrace{\|\hat{f}_{u_m} - \tilde{f}_{u_m}\|_2}_{\text{Generated Regressor Error}} + \underbrace{\|\tilde{f}_{u_m} - g_1(u_m, \cdot)\|_2}_{\text{Statistical ML Error $+$ Kernel Bias}}.
\]

For the \emph{statistical ML error and kernel bias}: let $R(f) = \mathbb{E}\bigl[K\bigl(\frac{u_m - p(Z,X)}{h_n}\bigr)(YW - f(X))^2\bigr]$ denote the kernel-weighted population risk, and let $f^*(x) = \mathbb{E}[K(\cdot)YW \mid X\!=\!x] / \mathbb{E}[K(\cdot) \mid X\!=\!x]$ be the unconstrained Bayes optimal predictor. Since the squared loss satisfies the Pythagorean identity $R(f) - R(f^*) = \|f - f^*\|_{\bar{w}}^2$, where $\bar{w}(x) = \mathbb{E}[K(\cdot) \mid X\!=\!x]$ defines a weighted $L_2$ norm equivalent to $\|\cdot\|_2$ (the density of $p(Z,X)$ is bounded above and below by \cref{asp:continuous_regularity}), the oracle inequality for empirical risk minimization \citep[Theorem~14.1]{wainwright2019high} applied to $\tilde{f}_{u_m}$ over $\mathcal{F}_n$ gives
\[
    \|\tilde{f}_{u_m} - f^*\|_{\bar{w}}^2 \leqslant \inf_{f \in \mathcal{F}_n} \|f - f^*\|_{\bar{w}}^2 + O_P(r_{X,n}^2).
\]
Since $g_1(u_m, \cdot) \in \mathcal{F}_n$ by realizability (\cref{asp:continuous_rates}), the infimum is bounded by $\|g_1(u_m, \cdot) - f^*\|_{\bar{w}}^2 = O(h_n^4)$, where the $O(h_n^2)$ bias of $f^*$ relative to $g_1(u_m, \cdot)$ follows from the twice continuous differentiability of $g_1$ in $u$ (\cref{asp:continuous_rates}) via the bias expansion of Nadaraya--Watson-type kernel regression \citep[Theorem~1.1]{tsybakov2009introduction}. The kernel weights localize the regression to an effective sample of size $\Theta(nh_n)$ (\cref{asp:continuous_regularity}), so the oracle inequality applies at this effective sample size. By the triangle inequality:
\[
    \|\tilde{f}_{u_m} - g_1(u_m, \cdot)\|_2 \leqslant \|\tilde{f}_{u_m} - f^*\|_{\bar{w}} + \|f^* - g_1(u_m, \cdot)\|_{\bar{w}} = O_P(r_{X,n} + h_n^2).
\]

For the \emph{generated regressor error}: the kernel weight perturbation satisfies
\[
    \left|K\!\left(\frac{u_m - \hat{p}(Z_i,X_i)}{h_n}\right) - K\!\left(\frac{u_m - p(Z_i,X_i)}{h_n}\right)\right| \leqslant \frac{\|K'\|_\infty}{h_n}\|\hat{p} - p\|_\infty = O\!\left(\frac{r_{p,n}}{h_n}\right),
\]
where $K'$ is bounded by \cref{asp:continuous_rates}. Since the kernel-weighted regression is a ratio estimator with denominator bounded away from zero (\cref{asp:continuous_regularity}), this $O(r_{p,n}/h_n)$ perturbation of each weight propagates linearly to the weighted average, yielding $\|\hat{f}_{u_m} - \tilde{f}_{u_m}\|_2 = O_P(r_{p,n}/h_n)$.

Combining, $\|\hat{f}_{u_m} - g_1(u_m, \cdot)\|_2 = O_P(r_{p,n}/h_n + r_{X,n} + h_n^2)$. The interpolation between grid points adds at most the grid resolution error, which is $O(1/M) = O((nh_n)^{-1/2})$ and is absorbed. Therefore,
\[
    |\hat{g}_1(u_l, X) - g_1(u_l, X)| = O_P\!\left(\frac{r_{p,n}}{h_n} + r_{X,n} + h_n^2\right).
\]
Aggregating the evaluation plug-in and estimator accuracy contributions, and noting $r_{p,n} \leqslant r_{p,n}/h_n$ for $h_n \leqslant 1$:
\[
    |\hat{\psi}_2(O) - \psi_2(O)| = O_P\!\left(\frac{r_{p,n}}{h_n} + r_{X,n} + h_n^2\right).
\]

\medskip
\noindent\textbf{Analysis of the boundary quantile integral $\hat{\psi}_1$.}
Define shorthand $a \coloneqq \underline{p}(X)$, $\hat{a} \coloneqq \widehat{\underline{p}}(X)$, $b \coloneqq \min\{\underline{p}(X), q(Z,X)\}$, and $\hat{b} \coloneqq \min\{\widehat{\underline{p}}(X), \hat{q}(Z,X)\}$. Adding and subtracting $\int_b^a \widehat{Q}_{Y,\underline{p}\mid X}(v/a)\,\mathrm{d}v$, the error decomposes as
\begin{align*}
    |\hat{\psi}_1(O) - \psi_1(O)| &\leqslant \underbrace{\left| \int_{\hat{b}}^{\hat{a}} \widehat{Q}_{Y,\underline{p}\mid X}\!\left(\frac{v}{\hat{a}}\right) \mathrm{d}v - \int_{b}^{a} \widehat{Q}_{Y,\underline{p}\mid X}\!\left(\frac{v}{a}\right) \mathrm{d}v \right|}_{\text{(I): limit and argument error}} + \underbrace{\int_{b}^{a} \left|\widehat{Q}_{Y,\underline{p}\mid X}\!\left(\frac{v}{a}\right) - Q_{Y,\underline{p}\mid X}\!\left(\frac{v}{a}\right)\right| \mathrm{d}v}_{\text{(II): quantile function error}}.
\end{align*}

\emph{Term~(I): Limit and argument error.}
If $\min\{a, \hat{a}\} = 0$, then $b \leqslant a$ and $\hat{b} \leqslant \hat{a}$ give $|\text{(I)}| \leqslant \max\{a, \hat{a}\}\,\|\widehat{Q}_{Y,\underline{p}\mid X}\|_\infty$. Since $\min\{a,\hat{a}\}=0$ implies $\max\{a,\hat{a}\} = |\hat{a}-a| = O_P(r_{p,n})$, we obtain $\text{(I)} = O_P(r_{p,n})$.

For $\min\{a,\hat{a}\} > 0$, the substitutions $t = v/\hat{a}$ and $t = v/a$ yield
\[
    \int_{\hat{b}}^{\hat{a}} \widehat{Q}_{Y,\underline{p}\mid X}\!\left(\frac{v}{\hat{a}}\right) \mathrm{d}v = \hat{a} \int_{\hat{b}/\hat{a}}^{1} \widehat{Q}_{Y,\underline{p}\mid X}(t)\,\mathrm{d}t, \quad \int_{b}^{a} \widehat{Q}_{Y,\underline{p}\mid X}\!\left(\frac{v}{a}\right) \mathrm{d}v = a \int_{b/a}^{1} \widehat{Q}_{Y,\underline{p}\mid X}(t)\,\mathrm{d}t.
\]
Define $G(c) \coloneqq \int_c^{1} \widehat{Q}_{Y,\underline{p}\mid X}(t)\,\mathrm{d}t$, so that $\text{(I)} = |\hat{a}\,G(\hat{b}/\hat{a}) - a\,G(b/a)|$. Since $\widehat{Q}_{Y,\underline{p}\mid X}$ is uniformly bounded by $[y_{\min}, y_{\max}]$ (\cref{asp:bounded_y}), $G$ is $\|\widehat{Q}_{Y,\underline{p}\mid X}\|_\infty$-Lipschitz and bounded by $\|\widehat{Q}_{Y,\underline{p}\mid X}\|_\infty$. We decompose
\[
    \hat{a}\,G(\hat{b}/\hat{a}) - a\,G(b/a) = \hat{a}\bigl[G(\hat{b}/\hat{a}) - G(b/a)\bigr] + (\hat{a} - a)\,G(b/a),
\]
so that $\text{(I)} \leqslant \hat{a}\,\|\widehat{Q}_{Y,\underline{p}\mid X}\|_\infty \cdot |\hat{b}/\hat{a} - b/a| + |\hat{a} - a|\,\|\widehat{Q}_{Y,\underline{p}\mid X}\|_\infty$. The second term is $O_P(r_{p,n})$. For the first term, note that
\[
    \hat{a}\left|\frac{\hat{b}}{\hat{a}} - \frac{b}{a}\right| = \left|\hat{b} - \frac{\hat{a}\,b}{a}\right| = \left|(\hat{b} - b) - \frac{b}{a}(\hat{a} - a)\right| \leqslant |\hat{b} - b| + \frac{b}{a}\,|\hat{a} - a| \leqslant |\hat{b} - b| + |\hat{a} - a| = O_P(r_{p,n}),
\]
where the penultimate inequality uses $b \leqslant a$. Hence $\text{(I)} = O_P(r_{p,n})$.

\emph{Term~(II): Quantile function error.}
The substitution $t = v/a$ yields
\[
    \text{(II)} = a \int_{b/a}^{1} \big|\widehat{Q}_{Y,\underline{p}\mid X}(t) - Q_{Y,\underline{p}\mid X}(t)\big|\,\mathrm{d}t \leqslant \int_0^1 \big|\widehat{Q}_{Y,\underline{p}\mid X}(t) - Q_{Y,\underline{p}\mid X}(t)\big|\,\mathrm{d}t,
\]
where the inequality uses $a \leqslant 1$ and $[b/a, 1] \subseteq [0,1]$. Note that this integrated form is precisely the conditional $W_1$ distance between $\widehat{F}_{Y,\underline{p}\mid X}$ and $F_{Y,\underline{p}\mid X}$, so no uniform-in-$\tau$ control is needed: integration absorbs any pointwise blow-up of the quantile error at $\tau \in \{0,1\}$ against the vanishing Lebesgue measure near the endpoints. We decompose via the triangle inequality through the localized target $Q_{Y,\delta_n\mid X}$ defined in \cref{asp:continuous_rates}(6):
\begin{align*}
    \int_0^1 \big|\widehat{Q}_{Y,\underline{p}\mid X}(t) &- Q_{Y,\underline{p}\mid X}(t)\big|\,\mathrm{d}t \\
    &\leqslant \underbrace{\int_0^1 \big|\widehat{Q}_{Y,\underline{p}\mid X}(t) - Q_{Y,\delta_n\mid X}(t)\big|\,\mathrm{d}t}_{\text{(II.a): statistical estimation error}} + \underbrace{\int_0^1 \big|Q_{Y,\delta_n\mid X}(t) - Q_{Y,\underline{p}\mid X}(t)\big|\,\mathrm{d}t}_{\text{(II.b): localization bias}},
\end{align*}
where $Q_{Y,\delta_n\mid X}(\tau)$ is the $\tau$-th conditional quantile of $Y$ given $p(Z,X) \leqslant \underline{p}(X) + \delta_n$ and $X$, as defined in \cref{asp:continuous_rates}(6).

For~(II.a), the estimator $\widehat{Q}_{Y,\underline{p}\mid X}$ is trained on the localized subset $\mathcal{I}_\delta = \{i \in I_1 : \hat{p}(Z_i, X_i) \leqslant \widehat{\underline{p}}(X_i) + \delta_n\}$, where $\hat{p}$ and $\widehat{\underline{p}}$ are estimated from the independent sample $I_0$. Conditional on $I_0$, the selection rule $\hat{p}(Z_i,X_i) \leqslant \widehat{\underline{p}}(X_i) + \delta_n$ is a fixed (non-random) function of each observation, so the selected observations $\{(Y_i, X_i)\}_{i \in \mathcal{I}_\delta}$ are conditionally i.i.d.\ from the distribution of $(Y,X)$ given $\hat{p}(Z,X) \leqslant \widehat{\underline{p}}(X) + \delta_n$. Since $r_{p,n} = o(\delta_n)$, the sup-norm error in $\hat{p}$ shifts the effective localization bandwidth by a negligible amount relative to $\delta_n$, so this conditional distribution is within $O(r_{p,n})$ total variation distance of the target population with $p(Z,X) \leqslant \underline{p}(X) + \delta_n$. By \cref{asp:continuous_regularity}, the density of $p(Z,X)$ near $\underline{p}(X)$ is bounded away from zero, so $|\mathcal{I}_\delta| = \Theta_P(n\delta_n)$. Applying the integrated quantile rate from \cref{asp:continuous_rates}(5) to this conditionally i.i.d.\ sample of effective size $\Theta_P(n\delta_n)$ gives, uniformly in $X$,
\[
    \text{(II.a)} = O_P\bigl(r_Q(n\delta_n)\bigr).
\]

For~(II.b), the $W_1$-Lipschitz smoothness condition in \cref{asp:continuous_rates}(6) applied at $\delta = \delta_n$ gives directly, uniformly in $x \in \mathcal{X}$,
\[
    \text{(II.b)} = \int_0^1 \big|Q_{Y,\delta_n\mid X}(t) - Q_{Y,\underline{p}\mid X}(t)\big|\,\mathrm{d}t \leqslant C_{\mathrm{Lip}}\,\delta_n = O(\delta_n).
\]

Combining Terms~(I) and~(II), and noting that $r_{p,n} = o(\delta_n)$:
\[
    |\hat{\psi}_1(O) - \psi_1(O)| = O_P\bigl(r_Q(n\delta_n) + \delta_n\bigr).
\]

\medskip
\noindent\textbf{Aggregation.}
Combining the three components and noting that $r_{p,n}$ is dominated by $r_{p,n}/h_n$ (since $h_n \leq 1$) and the parametric rate $n^{-1/2}$ is dominated by the nonparametric rates:
\begin{align*}
    |\hat{\underline{\theta}}_{\omega,1} - \underline{\theta}_{\omega,1}| &= O_P\!\left(\frac{r_{p,n}}{h_n} + r_{X,n} + h_n^2 + r_Q(n\delta_n) + \delta_n\right).
\end{align*}

\end{proof}

\section{Additional Simulation Details}\label{app:simulation}

\subsection{Synthetic Experiment: Data-Generating Process}

\paragraph{Continuous instrument.}
The DGP has the following components:
\begin{itemize}
    \item Instrument: $Z \sim \text{Unif}(0,1)$, independent of everything else.
    \item Latent resistance: $U \mid Z \sim \text{Unif}(0,1)$, independent of $Z$.
    \item Propensity score: $p(Z) = \operatorname{logistic}(\beta_0 + \beta_1 Z)$ with $\beta_0 = -1$ and $\beta_1 = 2$, so $p(Z)$ ranges from $\operatorname{logistic}(-1) \approx 0.27$ to $\operatorname{logistic}(1) \approx 0.73$.
    \item Treatment: $W = \indicator(U \leqslant p(Z))$.
    \item Potential outcomes: $Y(1) = aU + \theta + \varepsilon$ and $Y(0) = aU + \varepsilon$, where $a = 0.5$, $\theta = 0.5$, and $\varepsilon \sim \text{Unif}(-0.5, 0.5)$. The outcome support is $[y_{\min}, y_{\max}] = [-1, 2]$.
    \item Observed outcome: $Y = WY(1) + (1-W)Y(0)$.
\end{itemize}
The MTE is constant: $\text{MTE}(u) = \mathbb{E}[Y(1)-Y(0)\mid U=u] = \theta = 0.5$ for all $u$. The policy of interest is $q_\alpha(Z) = \operatorname{clip}(p(Z) + \alpha, 0, 1)$, and the target parameter is
\[
    \theta_\alpha = \mathbb{E}[Y^{q_\alpha} - Y] = \theta \cdot \mathbb{E}[\indicator(U \leqslant q_\alpha(Z)) - \indicator(U \leqslant p(Z))].
\]
For small $|\alpha|$ where boundary clipping is negligible, this simplifies to $\theta_\alpha \approx \theta \cdot \alpha = 0.5\alpha$.

\paragraph{Discrete instrument.}
The discrete setting uses a different DGP designed to highlight the performance of the closed-form bounds when the propensity score takes only finitely many values.

\begin{itemize}
    \item \textbf{Instrument and latent resistance.} $Z \sim \text{Bernoulli}(0.5)$ (binary) and $U \sim \text{Unif}(0,1)$, independent of each other.

    \item \textbf{Propensity score and treatment.} The propensity score is piecewise constant with two levels:
    \[
        p(Z) = \begin{cases} 0.25 & \text{if } Z = 0, \\ 0.75 & \text{if } Z = 1. \end{cases}
    \]
    Treatment follows the threshold-crossing rule $W = \indicator(U \leqslant p(Z))$.

    \item \textbf{Potential outcomes.} The control outcome is $Y(0) = 0.1 + 0.22U$ and the treatment effect is heterogeneous and increasing in $U$: $\tau(U) = 0.48 + 0.18U$, giving $Y(1) = Y(0) + \tau(U) = 0.58 + 0.40U$. Hence the MTE is $\text{MTE}(u) = 0.48 + 0.18u$, which ranges from $0.48$ to $0.66$. Measurement noise $\varepsilon \sim \text{Unif}(-0.1, 0.1)$ is added to the observed outcome $Y = WY(1) + (1-W)Y(0) + \varepsilon$. The outcome support used in estimation is $[y_{\min}, y_{\max}] = [0, 1]$.
\end{itemize}

\noindent The target parameter is $\theta_\alpha = \mathbb{E}[Y^{q_\alpha} - Y]$ under the same uniform policy shift $q_\alpha(Z) = \operatorname{clip}(p(Z)+\alpha, 0, 1)$. For small positive $\alpha$, the new compliers at each instrument value have $U \in (p(Z), p(Z)+\alpha)$, so the ground truth is approximately
\[
    \theta_\alpha \;=\; \tfrac{1}{2}\int_{0.25}^{0.25+\alpha}(0.48+0.18u)\,\mathrm{d}u \;+\; \tfrac{1}{2}\int_{0.75}^{0.75+\alpha}(0.48+0.18u)\,\mathrm{d}u \;=\; 0.57\alpha.
\]
Because the propensity score support consists of only two points $\{0.25, 0.75\}$, the MTE is identified only on the interval $(0.25, 0.75)$, and the IVOT and IVMTE bounds differ in how they treat the unidentified regions $(0, 0.25)$ and $(0.75, 1)$.

\paragraph{Sample size and comparison.}
The continuous instrument experiment uses $n=10{,}000$ observations and the discrete instrument experiment uses $n=5{,}000$ observations. IVMTE bounds are computed using the \texttt{ivmte} R package with moment constraints derived from the observed data. The PRTE weight function $\omega(u)$ corresponding to the $q_\alpha$-policy is used for both methods.

\subsection{Synthetic Experiment: Numerical Results}

\cref{tab:synthetic_continuous,fig:synthetic_continuous,tab:synthetic_discrete} report the numerical bounds for selected values of $\alpha > 0$, together with the IVMTE 95\% backward confidence interval. For the discrete instrument setting, the IVOT 95\% delta-method confidence interval is also reported. The IVOT bounds are uniformly tighter: in the continuous case by a factor of 5--8$\times$, and in the discrete case by approximately 2.4$\times$ across all tested $\alpha$. For the discrete instrument, the IVOT 95\% delta-method CI is also tighter than the IVMTE 95\% backward CI. The IVMTE 95\% CI is only slightly wider than the IVMTE identified set in most cases.

\begin{table}[htbp]
\centering
\caption{Synthetic experiment (continuous instrument, $n=10{,}000$): IVOT vs.\ IVMTE bounds for $\theta_\alpha = \mathbb{E}[Y^{q_\alpha}-Y]$. IVOT covers the truth for all $\alpha$. The IVMTE 95\% CI is a backward confidence interval computed by the \texttt{ivmte} package.}
\label{tab:synthetic_continuous}
\renewcommand{\arraystretch}{1.3}
\begin{tabular}{ccccccc}
\toprule
$\alpha$ & Truth & IVOT $[\ell,u]$ & IVOT width & IVMTE $[\ell,u]$ & IVMTE 95\% CI \\
\midrule
$0.01$ & $0.0050$ & $[0.0046,\; 0.0051]$ & $0.0005$ & $[-0.0066,\; 0.0074]$ & $[-0.0068,\; 0.0075]$ \\
$0.03$ & $0.0152$ & $[0.0121,\; 0.0164]$ & $0.0044$ & $[-0.0240,\; 0.0435]$ & $[-0.0292,\; 0.0490]$ \\
$0.05$ & $0.0253$ & $[0.0172,\; 0.0292]$ & $0.0119$ & $[-0.0313,\; 0.0651]$ & $[-0.0369,\; 0.0703]$ \\
$0.07$ & $0.0354$ & $[0.0204,\; 0.0433]$ & $0.0229$ & $[-0.0444,\; 0.0911]$ & $[-0.0526,\; 0.0965]$ \\
$0.10$ & $0.0502$ & $[0.0215,\; 0.0665]$ & $0.0450$ & $[-0.0613,\; 0.1287]$ & $[-0.0683,\; 0.1370]$ \\
$0.12$ & $0.0601$ & $[0.0198,\; 0.0830]$ & $0.0632$ & $[-0.0711,\; 0.1537]$ & $[-0.0893,\; 0.1678]$ \\
\bottomrule
\end{tabular}
\end{table}

\begin{figure}[htbp]
    \centering
    \includegraphics[width=0.60\textwidth]{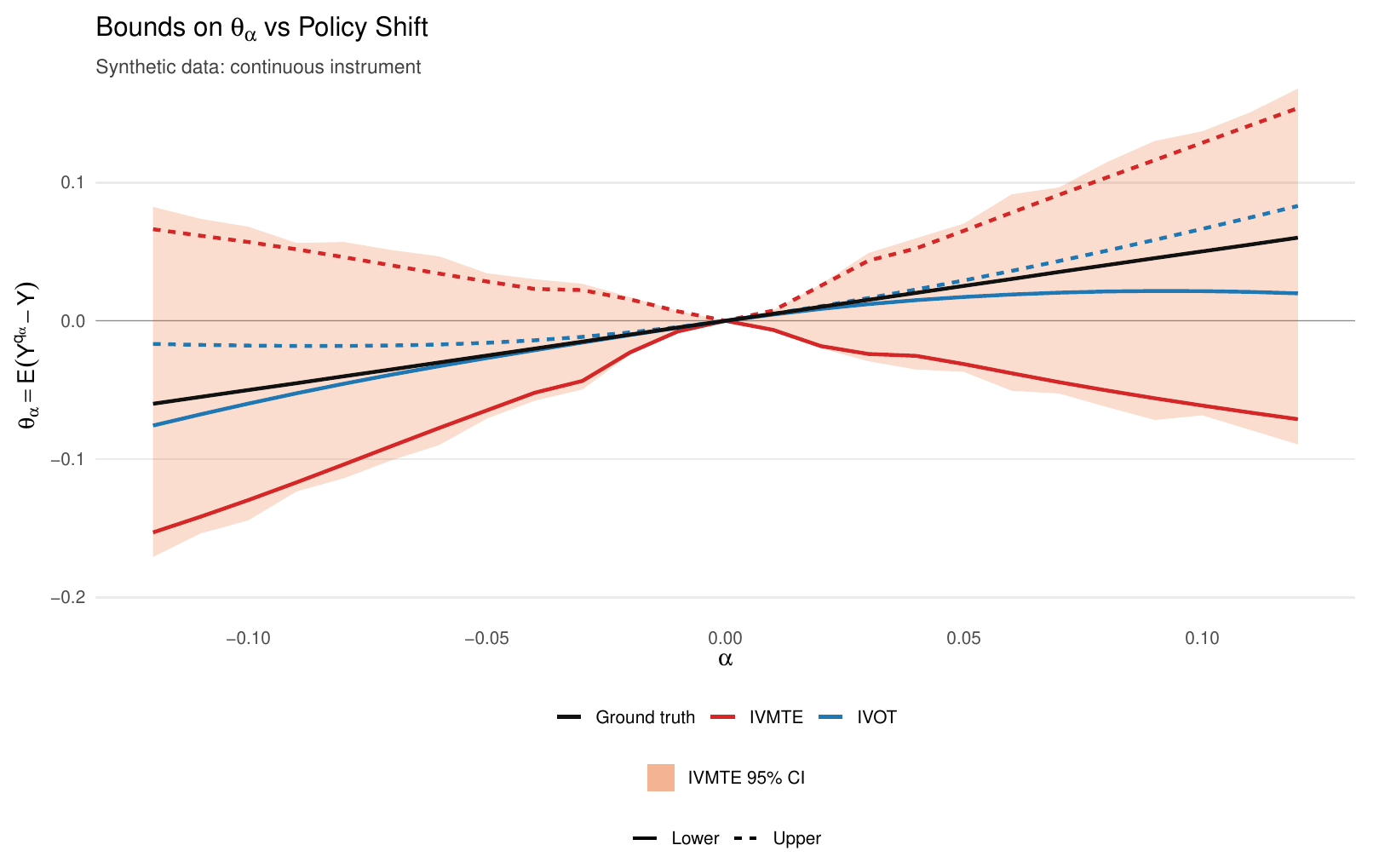}
    \caption{Continuous instrument ($n=10{,}000$): IVOT versus IVMTE identified sets and IVMTE 95\% backward CI for $\theta_\alpha = \mathbb{E}[Y^{q_\alpha}-Y]$ across $\alpha \in [-0.12, 0.12]$. The IVOT bounds cover the truth for all $\alpha$.}
    \label{fig:synthetic_continuous}
\end{figure}

\paragraph{Finite-sample coverage: effect of sample size.}
\cref{tab:synthetic_continuous_5000,fig:synthetic_continuous_5000} report the continuous instrument results at the smaller sample size of $n=5{,}000$ (compare with \cref{tab:synthetic_continuous,fig:synthetic_continuous} at $n=10{,}000$). While the IVOT bounds remain substantially tighter than IVMTE throughout, coverage drops below 100\% for three values of $\alpha$ near zero ($\alpha \in \{-0.02, -0.01, 0.01\}$), yielding an overall coverage rate of $88\%$ (22 out of 25 grid points). The near-misses occur where the IVOT interval is extremely tight (width ${\approx}\,0.002$) and the true $\theta_\alpha$ falls just outside the estimated bounds due to finite-sample error in propensity score estimation. At $n = 10{,}000$, coverage reaches 100\% across all $\alpha$, confirming that this is a finite-sample phenomenon rather than a systematic bias.

\begin{table}[htbp]
\centering
\caption{Synthetic experiment (continuous instrument, $n=5{,}000$): IVOT vs.\ IVMTE bounds for $\theta_\alpha = \mathbb{E}[Y^{q_\alpha}-Y]$. A $\dagger$ next to $\alpha$ indicates a coverage failure, which occurs only for small $|\alpha|$ where the IVOT interval is extremely narrow.}
\label{tab:synthetic_continuous_5000}
\renewcommand{\arraystretch}{1.3}
\begin{tabular}{ccccccc}
\toprule
$\alpha$ & Truth & IVOT $[\ell,u]$ & IVOT width & IVMTE $[\ell,u]$ & IVMTE 95\% CI \\
\midrule
$0.01^{\dagger}$ & $0.0050$ & $[0.0044,\; 0.0049]$ & $0.0005$ & $[-0.0057,\; 0.0063]$ & $[-0.0059,\; 0.0064]$ \\
$0.03$ & $0.0152$ & $[0.0121,\; 0.0163]$ & $0.0042$ & $[-0.0246,\; 0.0449]$ & $[-0.0280,\; 0.0485]$ \\
$0.05$ & $0.0252$ & $[0.0183,\; 0.0294]$ & $0.0112$ & $[-0.0290,\; 0.0666]$ & $[-0.0359,\; 0.0725]$ \\
$0.07$ & $0.0354$ & $[0.0225,\; 0.0439]$ & $0.0214$ & $[-0.0405,\; 0.0932]$ & $[-0.0504,\; 0.0992]$ \\
$0.10$ & $0.0502$ & $[0.0251,\; 0.0671]$ & $0.0419$ & $[-0.0593,\; 0.1329]$ & $[-0.0694,\; 0.1449]$ \\
$0.12$ & $0.0601$ & $[0.0241,\; 0.0830]$ & $0.0589$ & $[-0.0699,\; 0.1580]$ & $[-0.0922,\; 0.1755]$ \\
\bottomrule
\end{tabular}
\end{table}

\begin{figure}[htbp]
    \centering
    \includegraphics[width=0.60\textwidth]{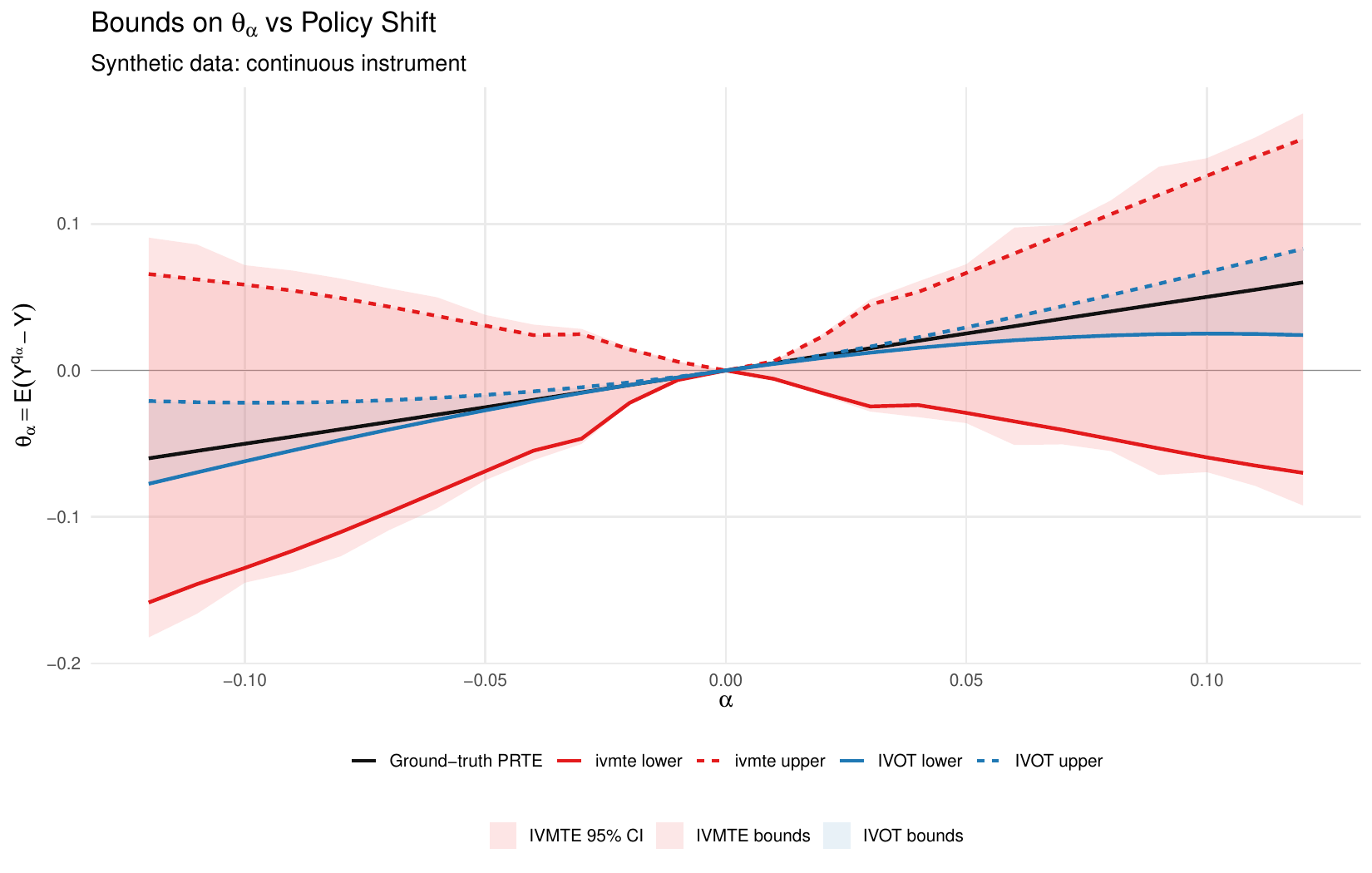}
    \caption{Continuous instrument ($n=5{,}000$): IVOT versus IVMTE identified sets for $\theta_\alpha = \mathbb{E}[Y^{q_\alpha}-Y]$ across $\alpha \in [-0.12, 0.12]$. The three values near $\alpha = 0$ where coverage fails are visible as points where the truth (solid line) falls just outside the IVOT bounds.}
    \label{fig:synthetic_continuous_5000}
\end{figure}

\begin{table}[htbp]
\centering
\caption{Synthetic experiment (discrete instrument): IVOT vs.\ IVMTE bounds for $\theta_\alpha = \mathbb{E}[Y^{q_\alpha}-Y]$. Coverage = 1 for all $\alpha$.}
\label{tab:synthetic_discrete}
\renewcommand{\arraystretch}{1.3}
\resizebox{\textwidth}{!}{%
\small
\begin{tabular}{ccccccc}
\toprule
$\alpha$ & Truth & IVOT $[\ell,u]$ & IVOT width & IVOT 95\% CI & IVMTE $[\ell,u]$ & IVMTE 95\% CI \\
\midrule
$0.01$ & $0.0057$ & $[-0.0007,\; 0.0083]$ & $0.0090$ & $[-0.0062,\; 0.0122]$ & $[-0.0106,\; 0.0109]$ & $[-0.0110,\; 0.0109]$ \\
$0.03$ & $0.0172$ & $[-0.0010,\; 0.0245]$ & $0.0255$ & $[-0.0069,\; 0.0293]$ & $[-0.0315,\; 0.0334]$ & $[-0.0338,\; 0.0335]$ \\
$0.05$ & $0.0287$ & $[-0.0008,\; 0.0401]$ & $0.0409$ & $[-0.0069,\; 0.0455]$ & $[-0.0470,\; 0.0503]$ & $[-0.0495,\; 0.0503]$ \\
$0.07$ & $0.0405$ & $[-0.0001,\; 0.0551]$ & $0.0552$ & $[-0.0063,\; 0.0610]$ & $[-0.0608,\; 0.0702]$ & $[-0.0632,\; 0.0702]$ \\
$0.10$ & $0.0581$ & $[0.0017,\; 0.0770]$ & $0.0753$ & $[-0.0047,\; 0.0835]$ & $[-0.0709,\; 0.1067]$ & $[-0.0770,\; 0.1069]$ \\
$0.12$ & $0.0697$ & $[0.0033,\; 0.0911]$ & $0.0878$ & $[-0.0032,\; 0.0978]$ & $[-0.0815,\; 0.1440]$ & $[-0.0890,\; 0.1448]$ \\
\bottomrule
\end{tabular}}
\end{table}

\subsection{Bed Net Application: Data and Setup}

\paragraph{Dataset.}
We use data from \citet{dupas2014subsidies}, which studies insecticide-treated bed net (ITN) take-up among Kenyan households. The instrument is the offered price $Z$ taking 17 distinct values spanning 0--250 Kenyan shillings. The treatment $W \in \{0,1\}$ indicates ITN purchase and the outcome $Y \in \{0,1\}$ indicates ITN usage at a one-year follow-up. The estimation sample consists of $n = 1078$ observations after merging the purchase and follow-up datasets.

\paragraph{Propensity score estimation.}
The propensity score is estimated via logistic regression of $W$ on $Z$, post-processed with isotonic regression to enforce the monotonicity assumption (higher price $\Rightarrow$ lower purchase probability). The estimated propensity scores range from approximately $0.23$ at the reference price of 150 KSh to $0.85$ at zero price.

\paragraph{Policy and target.}
The baseline policy corresponds to the reference price $z_0 = 150$ KSh with $\hat{p}(z_0) \approx 0.23$. The alternative policy $q_\alpha$ shifts the propensity score upward: $q_\alpha = \min(\hat{p}(z_0) + \alpha, 1)$ for $\alpha \in [0.05, 0.62]$. The maximum $\alpha_{\max} \approx 0.621$ equals the propensity at zero price minus the baseline propensity. The target is the policy-relevant treatment effect $\text{PRTE}_\alpha = \mathbb{E}[Y^{q_\alpha}-Y]/\alpha$, measuring the average per-unit effect of increasing compliance probability by $\alpha$.

\subsection{Bed Net Application: Numerical Results}

\cref{tab:bednet} reports the IVOT and IVMTE bounds for $\text{PRTE}_\alpha$ at selected subsidy levels, together with 95\% confidence intervals for both methods. IVOT yields tighter bounds throughout. Both methods indicate a positive and economically meaningful effect of price subsidies on bed net usage. The IVOT 95\% delta-method CI rules out non-positive per-unit effects for all subsidy levels considered, whereas the IVMTE 95\% backward CI remains inconclusive at small $\alpha$. At large $\alpha$ (near the maximum feasible shift), the IVOT bounds nearly point-identify $\text{PRTE}_\alpha$, whereas the IVMTE identified set and its 95\% CI remain substantially wider.

\begin{table}[htbp]
\centering
\caption{Bed net application: IVOT vs.\ IVMTE bounds for $\text{PRTE}_\alpha = \mathbb{E}[Y^{q_\alpha}-Y]/\alpha$ at selected policy shifts $\alpha$. The IVMTE point bounds reported in the table use degree-$10$ $u$-splines (degree-$20$ results are discussed in the main text and shown in \cref{fig:bednet_bounds}). The IVOT 95\% CI is computed via the delta method (influence function); the IVMTE 95\% CI is a backward confidence interval from the \texttt{ivmte} package based on the degree-$20$ specification.}
\label{tab:bednet}
\renewcommand{\arraystretch}{1.3}
\resizebox{\textwidth}{!}{%
\begin{tabular}{cccccc}
\toprule
$\alpha$ & IVOT lower & IVOT upper & IVOT 95\% CI & IVMTE $[\ell,u]$ & IVMTE 95\% CI \\
\midrule
$0.050$ & $0.740$ & $1.000$ & $[-1.000,\; 1.000]$ & $[0.129,\; 0.990]$ & $[-0.636,\; 1.002]$ \\
$0.094$ & $0.774$ & $0.774$ & $[-0.033,\; 1.000]$ & $[0.167,\; 0.949]$ & $[-0.380,\; 1.000]$ \\
$0.153$ & $0.703$ & $0.800$ & $[-0.084,\; 1.000]$ & $[0.262,\; 0.830]$ & $[-0.164,\; 1.000]$ \\
$0.211$ & $0.580$ & $0.601$ & $[0.158,\; 1.000]$ & $[0.406,\; 0.703]$ & $[-0.050,\; 1.000]$ \\
$0.299$ & $0.588$ & $0.637$ & $[0.290,\; 0.909]$ & $[0.575,\; 0.674]$ & $[0.183,\; 1.000]$ \\
$0.402$ & $0.622$ & $0.640$ & $[0.306,\; 0.960]$ & $[0.615,\; 0.753]$ & $[0.393,\; 0.987]$ \\
$0.504$ & $0.603$ & $0.624$ & $[0.410,\; 0.808]$ & $[0.656,\; 0.784]$ & $[0.518,\; 0.933]$ \\
$0.621$ & $0.597$ & $0.597$ & $[0.457,\; 0.737]$ & $[0.651,\; 0.782]$ & $[0.501,\; 0.861]$ \\
\bottomrule
\end{tabular}}
\end{table}